\documentclass[a4paper,11pt,DIV=calc,parskip=false]{scrartcl}

\usepackage[utf8]{inputenc}
\usepackage{amsmath,amsfonts,amssymb,amsthm,mathtools}
\usepackage{txfonts}
\usepackage{dsfont}
\usepackage[ngerman,english]{babel}
\usepackage[shortlabels]{enumitem}
\usepackage[all]{xy}
\usepackage{microtype}
\usepackage{mathrsfs,cite}
\usepackage{xcolor}
\usepackage{braket}
\numberwithin{equation}{section}
\def\cH{\mathcal{H}}
\def\cN{\mathcal{N}}
\def\cK{\mathcal{K}}
\def\cV{\mathcal{V}}
\def\cE{\mathcal{E}}

\def\cG{\mathcal{G}}
\def\Fock{\mathscr{F}}

\let\L=\Lambda

\DeclareMathOperator{\tr}{Tr}

\newtheorem{theorem}{Theorem}[section]
\newtheorem{proposition}[theorem]{Proposition}
\newtheorem{lemma}[theorem]{Lemma}
\newtheorem{corollary}[theorem]{Corollary}
\theoremstyle{definition}
\newtheorem{remark}[theorem]{Remark}

\newcommand{\dda}{\mathrm{d}}
\newcommand{\de}{\,\dda}

\renewcommand\rho\varrho
\renewcommand\epsilon\varepsilon

\definecolor{darkred}{rgb}{0.9,0,0.3}
\definecolor{darkblue}{rgb}{0,0.3,0.9}

\definecolor{vdarkred}{rgb}{0.7,0,0.2}
\definecolor{vdarkblue}{rgb}{0,0.2,0.7}
\usepackage[pdftex, colorlinks, linkcolor=vdarkblue,citecolor=vdarkred]{hyperref}

\def\cV{\mathcal{V}}

\def\cG{\mathcal{G}}

\def\cN{\mathcal{N}}
\def\cE{\mathcal{E}}
\def\cK{\mathcal{K}}
\def\cH{\mathcal{H}}

\def\fock{\mathscr{F}}

                  \let\p=\pi

\let\G=\Gamma        \let\L=\Lambda

\theoremstyle{definition}

\newcommand{\beq}{\begin{equation}}
\newcommand{\eeq}{\end{equation}}
\usepackage{pdfsync}

\begin{document}

\title{Upper bound for the grand canonical free energy of the Bose gas in the Gross--Pitaevskii limit}

\author{Chiara Boccato, Andreas Deuchert, David Stocker}

\date{\today}

\maketitle

\begin{abstract} 
	We consider a homogeneous Bose gas in the Gross--Pitaevskii limit at temperatures that are comparable to the critical temperature for Bose--Einstein condensation in the ideal gas. Our main result is an upper bound for the grand canonical free energy in terms of two new contributions: (a) the free energy of the interacting condensate is given in terms of an effective theory describing its particle number fluctuations, (b) the free energy of the thermally excited particles equals that of a temperature-dependent Bogoliubov Hamiltonian. 
\end{abstract}

\setcounter{tocdepth}{2}
\tableofcontents

\section{Introduction and main results}

\subsection{Background and summary}
\label{sec:Introduction}
The dilute Bose gas, that is, a bosonic system with rare but strong collisions, is one of the most fundamental and interesting models in quantum statistical mechanics. Its prominence is mostly due to the occurrence of the Bose–Einstein condensation (BEC) phase transition and its numerous phenomenological consequences. Triggered by the experimental realization of BEC in ultra cold alkali gases in 1995, see \cite{Wieetal1995,Davisetal1995}, and by the subsequent experimental progress,  in the past two decades there have been numerous mathematical investigations of dilute Bose gases in different parameter regimes. 

The most relevant parameter regime for the description of modern experiments with cold quantum gases is the Gross--Pitaevskii (GP) limit. Here the scattering length of the interaction between the particles is scaled with the particle number $N$ in such a way that the interaction energy, in the limit $N \to \infty$, is of the same order of magnitude as the spectral gap in the trap. It has been shown in \cite{LiSeiYng2000} that the ground state energy per particle can, in this limit, be approximated by the minimum of the GP energy functional. Moreover, approximate ground states of a trapped Bose gas display BEC and superfluidity, see \cite{LiSei2002,LiSeiYng2002}. The derivation of the GP energy functional has later been extended in \cite{Sei2003,LiSei2006} to the case of rotating gases, see also \cite{NamRouSei2016}. In such a system, the one-particle density matrices of approximate ground states can be shown to converge to a convex combination of projections onto the minimizers of the GP energy functional.

As predicted by Bogoliubov in 1947 in \cite{Bog1947}, the subleading correction to the ground state energy of a dilute Bose gas is given by the ground state energy of a certain quadratic Hamiltonian called Bogoliubov Hamiltonian. Recently, this claim has been proved in the GP limit for a homogeneous Bose gas in \cite{BocBreCeSchl2019,BocBreCeSchl2020}, for a homogeneous Bose gas with a slightly more singular interaction (Thomas-Fermi limit) in \cite{AdhBreSchl2020,BreCapoSchl2021}, and for a trapped Bose gas in \cite{NamNapRicTri2022,NamTriay2021,BreSchlSchr2021a,BreSchlSchr2021b}. The two-dimensional case has been investigated in \cite{CarCenaSchl2021,CarCenaSchl2023}. In all these works it was also possible to compute the low-lying eigenvalues of the Hamiltonian as well as the corresponding eigenfunctions. Simplified approaches in the homogeneous case have been provided in \cite{Hai2021,HaiSchlTriay2022}, and a second order upper bound for a system with hard core interactions was proved in \cite{BaCenaOlgPasqSchl2022}. A Bose gas in a box with Neumann boundary conditions has been studied in \cite{BocSei2023}. In case of mean-field interactions, Bogoliubov theory had previously been justified in \cite{Sei2011b,GrechSei2013}.

Low energy eigenstates provide an accurate description of a Bose gas at zero temperature. However, the understanding of the model at positive temperature is essential for the full description of experiments and crucial for the mathematical understanding of the BEC phase transition. In this case one is interested in the free energy and the Gibbs state, which are natural equivalents of the ground state energy and the corresponding eigenfunction. A trapped Bose gas in a combination of a thermodynamic limit in the trap and a GP limit was studied in \cite{DeuSeiYng2019}. There it could be shown that the difference between the free energy of the system and that of the ideal gas is approximately given by the minimum of the GP energy functional. Moreover, the one-particle density matrix of approximate minimizers of the free energy is, to leading order, given by the one of the ideal gas, where the condensate wave function has been replaced by the minimizer of the GP energy functional. This, in particular, establishes the existence of a BEC phase transition in the system. Comparable results have been obtained also for a homogeneous Bose gas, see \cite{DeuSeiGP2020}. 

Consider an approximate ground state of a trapped Bose gas in the GP limit. Its dynamics after the trapping potential has been switched off can be described by the time-dependent GP equation, see \cite{ErdSchlYau2009,ErdSchlYau2010,BenOlivSchl2015,Pickl2015,JeblLeopPickl2019,BreSchl2019}. The dynamics of approximate positive temperature states has so far been studied only for mean field (high density) systems. More information about the derivation of effective evolution equations for bosonic many-particle systems can be found in the reviews  \cite{BenPorSchl2015,Nap2021}.

The GP limit is appropriate to describe experiments with $10^2-10^6$ alkali atoms. In contrast, truly macroscopic samples with particle numbers of the order of the Avogadro constant $N_A \approx 6.022 \times 10^{23}$ are well-described by the thermodynamic limit followed by a dilute (i.e., low density) limit.  The asymptotic behavior of the specific energy in this setting has been obtained in \cite{Dyson,LiYng1998}. Results in two and one space dimensions can be found in \cite{LiYng2001} and \cite{AgersReuversSol2022}, respectively. Also the next-to-leading order correction (Lee--Huang--Yang (LHY) term) predicted in \cite{LHY1957} could recently be established, see \cite{YauYin2009,BastCenaSchlein2021} (upper bound), \cite{FouSol2020,FouSol2021} (lower bound), and \cite{Fouetal} (comparable result in two space dimensions). A two-term expansion for the free energy of the three-dimensional system has been proved in \cite{Yin2010} (upper bound) and \cite{Sei2008} (lower bound), and for the two-dimensional system in \cite{MaySei2020} (upper bound) and \cite{DeuMaySei2020} (lower bound). In the latter case the result depends on the critical temperature of the Berezinskii–-Kosterlitz–-Thouless critical temperature for superfluidity. Finally, a LHY type lower bound for the free energy at suitably low temperatures, where the contribution of the excitation spectrum and the LHY correction are of the same order, has been proved in \cite{HabHaiNamSei2023}. For a more extensive list of references concerning static properties of Bose gases we refer to \cite{LiSeiSolYng2005,Rou2015}.


In the present article we consider a homogeneous Bose gas in the GP limit at temperatures of the order of the critical temperature for BEC. Our main result is an upper bound for the grand canonical free energy in terms of two new contributions. The first is the free energy of the particle number fluctuations of the interacting condensate and is described by a suitable effective theory. The second new contribution is related to the free energy of thermal excitations over the condensate. For temperatures of the order of the critical temperature, the number of excited particles may be of the same order as the number of particles in the condensate, and Bogoliubov modes need to be described in terms of a \emph{temperature-dependent Bogoliubov Hamiltonian}. To obtain our upper bound, we construct a trial state as follows: particles in the condensate are described by a convex combination of coherent states, which allows us to increase their entropy. The excitations are described by a Gibbs state of free bosons with \emph{Bogoliubov dispersion relation}. The resulting state is a convex combination of quasi-free states, which we further transform to include two-body correlations. To do this we employ a suitable second quantized quartic operator. When computing the energy of our trial state, this operator allows us to renormalize the interaction potential and to show that the result only depends on the scattering length.

\subsection{Notation}
\label{sec:Notation}
For two functions $a$ and $b$ of the particle number and other parameters of the system, we use the notation $a \lesssim b$ to say that there exists a constant $C > 0$ independent of the parameters such that $a \leq C b$. If we want to highlight that $C$ depends on a parameter $k$ we use the symbol $\lesssim_k$. If $a \lesssim b$ and $b \lesssim a$ we write $a \sim b$ and $a \simeq b$ means that $a$ and $b$ are equal to leading order in the limit considered. By $C,c>0$ we denote generic constants, whose values may change from line to line. The Fourier coefficients of a periodic function $f : [0,L]^3 \to \mathbb{C}$ are denoted by $\hat{f}(p) = \int_{[0,L]^3} e^{-\mathrm{i} p x} f(x) \de x$, and for two Fourier coefficients $\hat{f}, \hat{g}$ we define their convolution as
\begin{equation}
	\hat{f} \ast \hat{g} (p) = L^{-3} \sum_{p \in (2 \pi/L) \mathbb{Z}^3} \hat{f}(p-q) \hat{g}(q).
	\label{eq:Convolution}
\end{equation} 
This, in particular, implies $\widehat{fg}(p) = \hat{f} \ast \hat{g}(p)$.
\subsection{Fock space and Hamiltonian}
\label{sec:FockSpace_Hamiltonian}
We consider a system of bosons confined to a three dimensional flat torus $\Lambda$ with side length $L$. In what follows, we could set $L = 1$ but we prefer to keep a length scale to explicitly display units in formulas. The one-particle Hilbert space of the system is given by $L^2(\Lambda, \de x)$, with $\de x$ denoting the Lebesgue measure. We are interested in the grand canonical ensemble, that is, in a system with a fluctuating particle number. The Hilbert space of the entire system is therefore given by the bosonic Fock space
\begin{equation}
	\mathscr{F}( L^2(\Lambda, \de x) ) = \bigoplus_{n=0}^{\infty} L^2_{\mathrm{sym}}(\Lambda^n, \de x).
	\label{eq:FockSpace}
\end{equation}
Here $L^2_{\mathrm{sym}}(\Lambda^n, \de x)$ denotes the closed linear subspace of $L^2(\Lambda^n, \de x)$ consisting of those functions $\Psi(x_1,...,x_n)$ that are invariant under any permutation of the coordinates $x_1, ... , x_n \in \Lambda$. As usual, we define $L^2_{\mathrm{sym}}(\Lambda^0, \de x) = \mathbb{C}$. 

On the $n$-particle Hilbert space $L^2_{\mathrm{sym}}(\Lambda^n, \de x)$ with $n \geq 1$ we define the Hamiltonian 
\begin{equation}
	\mathcal{H}_N^{(n)} = \sum_{i=1}^n -\Delta_i + \sum_{1 \leq i < j \leq n} v_N(d(x_i,x_j)),
	\label{eq:Hamiltonian(n)}
\end{equation}
where $\Delta$ denotes the Laplacian on the torus $\Lambda$ and $d(x,y)$ is the distance between two points $x,y \in \Lambda$. In the realization of $\Lambda$ as the set $[0,L]^3$, $\Delta$ is the usual Laplacian with periodic boundary conditions and $d(x,y) = \min_{k \in \mathbb{Z}^3} | x-y-k L |$. We also define $\mathcal{H}_N^{(0)} = 0$. The interaction potential is of the form 
\begin{equation}
	v_N(d(x,y)) = N^2 v(Nd(x,y))
	\label{eq:InteractionPotential}
\end{equation}
with a measurable, compactly supported function $v : [0,\infty) \to [0,\infty]$ and a parameter $N>0$. We will later choose $N$ as the expected number of particles in the system. Our assumptions on $v$ guarantee that it has a finite scattering length $\frak{a} \geq 0$. The scattering length is a combined measure for the range and the strength of an interaction potential. For its definition we refer to \cite[Appendix C]{LiSeiSolYng2005} and Appendix~\ref{app:ScatteringEquation}.
A simple scaling argument shows that the scattering length of $v_N$ is $\frak{a}_N = \frak{a}/N$.  Finally, the Hamiltonian $\mathcal{H}_N$ acting on $\mathscr{F}$ is defined by 
\begin{equation}
	\mathcal{H}_N = \bigoplus_{n=0}^{\infty} \mathcal{H}_N^{(n)}.
	\label{eq:Hamiltonian}
\end{equation}


\subsection{Grand canonical free energy, Gibbs state and Gibbs variational principle}
\label{sec:FreeEnergy_GibbsState}
We are interested in a gas of bosons in the grand canonical ensemble. The usual thermodynamic variables used to describe such a system are the inverse temperature, the chemical potential and the volume of the container. The chemical potential can later be chosen to obtain a desired particle number. In this article we replace the chemical potential in the above list of variables by the expected number of particles, which yields an equivalent description of the system. This motivates the following definitions.

The set of states on the bosonic Fock space $\mathscr{F}(L^2(\Lambda,\de x))$ with an expected number of $N \geq 0$ particles is defined by
\begin{equation}
	\mathcal{S}_N = \{ \Gamma \in \mathcal{B}(\mathscr{F}) \ | \ \Gamma \geq 0, \tr \Gamma = 1, \tr[\mathcal{N}\Gamma] = N \},
	\label{eq:States}
\end{equation}
where 
\begin{equation}
	\mathcal{N} = \bigoplus_{n=0}^{\infty} n
	\label{eq:ParticleNumberOp}
\end{equation}
denotes the number operator on $\mathscr{F}$. For a state $\Gamma \in \mathcal{S}_N$, the Gibbs free energy functional reads
\begin{equation}
	\mathcal{F}(\Gamma) = \tr[\mathcal{H}_N \Gamma] - \frac{1}{\beta} S(\Gamma) \quad \text{ with the von-Neumann entropy } \quad S(\Gamma) = - \tr[ \Gamma \ln(\Gamma) ]
	\label{eq:GibbsFreeEnergyFunctional}
\end{equation}
and the inverse temperature $\beta > 0$. The grand canonical free energy of the system is defined as the minimum of $\mathcal{F}$ in the set $\mathcal{S}_N$:
\begin{equation}
	F(\beta,N,L) = \min_{\Gamma \in \mathcal{S}_N} \mathcal{F}(\Gamma) = -\frac{1}{\beta} \ln\left( \tr[\exp(-\beta(\mathcal{H}_N - \mu \mathcal{N}))] \right) + \mu N.
	\label{eq:FreeEnergy}
\end{equation}
Here the chemical potential $\mu$ is chosen such that the unique minimizer 
\begin{equation}
	G = \frac{\exp(-\beta(\mathcal{H}_N - \mu \mathcal{N}))}{\tr[\exp(-\beta(\mathcal{H}_N - \mu \mathcal{N}))]}
	\label{eq:GibbsState}
\end{equation}
of $\mathcal{F}$ satisfies $\tr[\mathcal{N}G] = N$. The state $G$ is called the (grand canonical) Gibbs state. 
\subsection{The ideal Bose gas on the torus}
\label{sec:IdealBoseGas}
The bound that we prove for the free energy $F(\beta,N,L)$ in \eqref{eq:FreeEnergy} depends on several quantities related to the ideal (i.e., noninteracting) Bose gas on the torus. In this section we recall their definition and briefly discuss their behavior as a function of the inverse temperature $\beta$.

The chemical potential $\mu_0(\beta,N,L) < 0$ of the ideal gas is defined as the unique solution to the equation
\begin{equation}
	N = \sum_{p \in \Lambda^*} \frac{1}{\exp(\beta(p^2-\mu_0(\beta,N,L))) - 1}, 
	\label{eq:ChemicalPotentialIdealGas}
\end{equation}
where $\Lambda^* = (2\pi/L) \mathbb{Z}^3$. The expected number of particles with momentum $p=0$ and their density read
\begin{equation}
	N_0(\beta,N,L) = (\exp(-\beta \mu_0) - 1 )^{-1} \quad \text{ and } \quad \varrho_0(\beta,N,L) = N_0(\beta,N,L)/L^3,
	\label{eq:DensityBEC}
\end{equation}
respectively. The asymptotic behavior of $N_0$ in the limit $N \to \infty$ is given by
\begin{equation}
	\frac{N_0(\beta,N,L)}{N} \simeq \left[ 1 - \frac{\beta_{\mathrm{c}}}{\beta} \right]_+ \quad \text{ with } \quad \beta_{\mathrm{c}} = \frac{1}{4 \pi} \left( \frac{N}{L^3 \zeta(3/2) } \right)^{-2/3}.
	\label{eq:BECPhaseTransition}
\end{equation}
We note that $\beta$ in \eqref{eq:BECPhaseTransition} usually depends on $N$. By $\zeta$ we denote the Riemann zeta function and $[x]_+ = \max\{0,x \}$. The above formula implies that the ideal Bose gas displays a BEC phase transition: If $\beta = \kappa \beta_{\mathrm{c}}$ with $\kappa \in (1,\infty)$ then $N_0 \simeq N [1-1/\kappa]$ and $|\mu_0| \sim L^{-2} N^{-1/3}$. In contrast, for $\beta = \kappa \beta_{\mathrm{c}}$ with $\kappa \in (0,1)$ we have $N_0 \sim 1$ and $|\mu_0| \sim L^{-2} N^{2/3}$. Finally, the grand canonical free energy of the ideal gas is given by $F_0 = F_0^{\mathrm{BEC}} + F_0^+$. Here
\begin{equation}
	F_0^{\mathrm{BEC}}(\beta,N,L) = \frac{1}{\beta} \ln\left( 1 - \exp\left( \beta \mu_0 \right) \right) + \mu_0 N_0
	\label{eq:FreeEnergyIdealGasCond}
\end{equation}
denotes the free energy of the condensate and 
\begin{equation}
	F_0^+(\beta,N,L) = \frac{1}{\beta} \sum_{p \in \Lambda_+^*} \ln\left( 1 - \exp\left( -\beta (p^2 - \mu_0) \right) \right) + \mu_0 (N-N_0)
	\label{eq:FreeEnergyIdealGas}
\end{equation}
that of the non-condensed particles.
\subsection{Main results}
\label{sec:main}
Our main result is the following upper bound for the free energy of the homogeneous Bose gas in the GP limit.

\begin{theorem}
	\label{thm:Main}
	Assume that the function $v : [0,\infty) \to [0,\infty]$ is nonnegative, compactly supported, and satisfies $v(| \cdot |) \in L^3(\Lambda, \de x)$. By $\varrho = N/L^3$ we denote the particle density. In the combined limit $N \to \infty$, $\beta = \kappa \beta_{\mathrm{c}}$ with  $\kappa \in (0,\infty)$ and $\beta_{\mathrm{c}}$ in \eqref{eq:BECPhaseTransition}, the free energy in \eqref{eq:FreeEnergy} satisfies the upper bound
	\begin{align}
		F(\beta,N,L) \leq& F^+_0(\beta,N,L) + 8 \pi \mathfrak{a}_N L^{3} \varrho^2 + \min\{ F^{\mathrm{BEC}} - 8 \pi \mathfrak{a}_N L^3 \varrho_0^2, F_0^{\mathrm{BEC}} \} \nonumber \\
		&- \frac{1}{2 \beta} \sum_{p \in \Lambda^*_+} \left[ \frac{16 \pi \mathfrak{a}_N \varrho_0(\beta,N,L)}{p^2} - \ln\left( 1 + \frac{16 \pi \mathfrak{a}_N \varrho_0(\beta,N,L)}{p^2} \right) \right]  + O(L^{-2} N^{7/12}), 
		\label{eq:MainFreeEnenergyBound}
	\end{align}
	with $\varrho_0$ in \eqref{eq:DensityBEC}, $F_0^{\mathrm{BEC}}$ in \eqref{eq:FreeEnergyIdealGasCond}, $F^+_0$ in \eqref{eq:FreeEnergyIdealGas}, and
	\begin{equation}
		F^{\mathrm{BEC}}(\beta,N_0,L,\mathfrak{a}_N) = -\frac{1}{\beta} \ln\left( \int_{\mathbb{C}} \exp\left( - \beta \left( 4 \pi \mathfrak{a}_N L^{-3} |z|^4 - \mu |z|^2 \right) \right) \de z \right) + \mu N_0(\beta,N,L).
		\label{eq:FreeEnergyBEC}
	\end{equation}
	Here $\de z = \de x \de y/\pi$, where $x$ and $y$ denote the real and imaginary part of the complex number $z$, respectively. The chemical potential $\mu$ in \eqref{eq:FreeEnergyBEC} is chosen such that the Gibbs distribution
	\begin{equation}
		g(z) = \frac{\exp\left( - \beta \left( 4 \pi \mathfrak{a}_N L^{-3} |z|^4 - \mu |z|^2 \right) \right) }{ \int_{\mathbb{C}} \exp\left( - \beta \left( 4 \pi \mathfrak{a}_N L^{-3} |z|^4 - \mu |z|^2 \right) \right) \de z } \quad \text{ satisfies } \quad \int_{\mathbb{C}} |z|^2 g(z) \de z = N_0(\beta,N,L).
		\label{eq:GibbsDistribution}
	\end{equation}
\end{theorem}


The terms on the r.h.s. of \eqref{eq:MainFreeEnenergyBound} are listed in descending order concerning their order of magnitude in the limit $N \to \infty$. The free energy of the non-condensed particles satisfies $F_0^+ \sim L^{-2} N^{5/3}$. The second term is a density-density interaction of the order $L^{-2} N$. As we will see with Proposition~\ref{prop:FreeEnergyBECb} below, the energy of the interacting condensate (the third term), contributes on two orders of magnitude (if $\kappa > 1$): $L^{-2} N$ and $L^{-2} N^{2/3} \ln(N)$. The term in the second line is a correction to the free energy of the non-condensed particles coming from Bogoliubov theory, and is of the order $L^{-2} N^{2/3}$. 

The following proposition provides us with a simplified expression for $F^{\mathrm{BEC}}$ above and below the critical point. This, in particular, allows us to compute the minimum on the r.h.s. of \eqref{eq:MainFreeEnenergyBound}.

\begin{proposition}
	\label{prop:FreeEnergyBECb} We consider the limit $N \to \infty$, $\beta = \kappa \beta_{\mathrm{c}}$ with $\kappa \in (0,\infty)$ and $\beta_{\mathrm{c}}$ in \eqref{eq:BECPhaseTransition}.
	The following statements hold for given $\epsilon > 0$:
	\begin{enumerate}[(a)]
	\item Assume that $N_0 \gtrsim N^{5/6 + \epsilon}$. There exists a constant $c>0$ such that 
	\begin{equation}
		F^{\mathrm{BEC}}(\beta,N_0,L,\mathfrak{a}_N) = 4 \pi \mathfrak{a}_N L^{3} \varrho_0^2 + \frac{\ln \left( 4 \beta \mathfrak{a}_N/L^3 \right)}{2 \beta} + O\left( L^{-2} \exp\left(- c N^{\epsilon} \right) \right).
		\label{eq:FreeEnergyBECInteractingLimitb}
	\end{equation}
	\item Assume that $N_0 \lesssim N^{5/6 - \epsilon}$. Then 
	\begin{equation}
		F^{\mathrm{BEC}}(\beta,N_0,L,\mathfrak{a}_N) = - \frac{1}{\beta} \ln(N_0) - \frac{1}{\beta} + O\left( L^{-2} N^{2/3 - 2 \epsilon} \right)
		\label{eq:FreeEnergyBECNonInteractingLimitb}
	\end{equation}
	holds. In particular, $F^{\mathrm{BEC}}(\beta,N_0,L,\mathfrak{a}_N)$ is independent of $\mathfrak{a}_N$ at the given level of accuracy. 
\end{enumerate} 
\end{proposition}

The interpretation of Proposition~\ref{prop:FreeEnergyBECb} is as follows: if the number of particles in the BEC is sufficiently large, we see a contribution of the order $L^{-2} N^{2/3} \ln(N)$  in addition to the density-density interaction $4 \pi \mathfrak{a}_N L^{3} \varrho_0^2$. This new contribution (the second term on the r.h.s. of \eqref{eq:FreeEnergyBECInteractingLimitb})
is a consequence of the particle number fluctuations in the BEC and will be discussed in more detail in Remark~\ref{rmk:MainResults}.(b) below. In contrast, if the expected particle number inside the BEC satisfies $1 \ll N_0 \leq N^{5/6-\epsilon}$ its free energy equals that of an ideal gas to leading order. The appearance of the exponent $5/6$ is explained by the fact that $4 \pi \mathfrak{a}_N L^{3} \varrho_0^2 \sim L^{-2} N^{2/3}$ if $N_0 \sim N^{5/6}$. This should be compared to $1/\beta$ times the classical entropy of $g$ (for a definition see \eqref{eq:ClassicalEntropyb} below), which, for $N^{\epsilon} \leq N_0 \leq N^{5/6}$ with $\epsilon > 0$, is always of the order $\ln(N)/\beta \sim L^{-2} N^{2/3} \ln(N)$. That is, in the parameter region $N^{5/6 - \epsilon} \lesssim N_0 \lesssim N^{5/6 + \epsilon}$ the effective theory of the condensate transitions from a regime where the interaction is relevant to a regime where it is not. For those values of $N_0$ the free energy $F^{\mathrm{BEC}}$ does not have a form that is as simple as that in \eqref{eq:FreeEnergyBECInteractingLimitb} or \eqref{eq:FreeEnergyBECNonInteractingLimitb}. 

Proposition~\ref{prop:FreeEnergyBECb} allows us to bring our main result into a form that is better suited for a comparison to the existing literature, as stated in the following Corollary. 

\begin{corollary}
	\label{cor:MainCorollary}
	Assume that the function $v : [0,\infty) \to [0,\infty]$ is nonnegative, compactly supported, and satisfies $v(| \cdot |) \in L^3(\Lambda, \de x)$. By $\varrho = N/L^3$ we denote the particle density. We consider the combined limit $N \to \infty$, $\beta = \kappa \beta_{\mathrm{c}}$ with $\kappa \in (0,\infty)$ and $\beta_{\mathrm{c}}$ in \eqref{eq:BECPhaseTransition}. If $\kappa \in (1,\infty)$ the free energy in \eqref{eq:FreeEnergy} satisfies the upper bound
	\begin{align}
		F(\beta,N,L) \leq& F^+_0(\beta,N,L) + 4 \pi \mathfrak{a}_N L^{3} \left( 2 \varrho^2 - \varrho_0^2(\beta,N,L) \right) + \frac{\ln \left( 4 \beta \mathfrak{a}_N/L^3 \right)}{2 \beta} \nonumber \\
		&- \frac{1}{2 \beta} \sum_{p \in \Lambda^*_+} \left[ \frac{16 \pi \mathfrak{a}_N \varrho_0(\beta,N,L)}{p^2} - \ln\left( 1 + \frac{16 \pi \mathfrak{a}_N \varrho_0(\beta,N,L)}{p^2} \right) \right]  + O(L^{-2} N^{7/12}) 
		\label{eq:UpperBoundFinal}
	\end{align}
	and if $\kappa \in (0,1)$ we have
	\begin{equation}
		F(\beta,N,L) \leq F_0(\beta,N,L) + 8 \pi \mathfrak{a}_N L^{3} \varrho^2 + O(L^{-2} N^{1/2})
		\label{eq:UpperBoundFinalb}
	\end{equation}
	with $F_0$ defined above \eqref{eq:FreeEnergyIdealGasCond}.
\end{corollary}

If $\kappa \in (1,\infty)$ the minimum in \eqref{eq:MainFreeEnenergyBound} is attained by the first term and one obtains \eqref{eq:UpperBoundFinal}. In contrast, for $\kappa \in (0,1)$ it equals the second term, which leads to \eqref{eq:UpperBoundFinalb}. At the critical point ($\kappa = 1$, or $\kappa\to 1$ as $N\to\infty$, see also Remark~\ref{rmk:MainResults}.(h) below) the minimum is needed. We have the following remarks concerning the above statements.

\begin{remark}
	\label{rmk:MainResults} 
	\begin{enumerate}[(a)]
		\item The first two terms on the r.h.s. of \eqref{eq:UpperBoundFinal} and \eqref{eq:UpperBoundFinalb} already appeared in an asymptotic expansion of the canonical free energy in the GP limit in \cite{DeuSeiGP2020} (with a remainder of the order $o(L^{-2} N)$). To be precise, the result in \eqref{eq:UpperBoundFinal} has been stated with $F_0^+$ replaced by the canonical free energy $F_0^{\mathrm{c}}$ of the ideal gas. From \cite[Lemma~A1]{DeuSeiGP2020} we however know that $F_0^{\mathrm{c}}$ and $F_0^+$ agree up to a remainder of the order $L^{-2} N^{2/3} \ln(N)$. It is to be expected that the result in \cite{DeuSeiGP2020} also holds if the grand canonical ensemble is considered. That is, the two ensembles are expected to be equivalent if one allows for remainders of the order $o(L^{-2} N)$. We highlight that the first two terms on the r.h.s. of \eqref{eq:UpperBoundFinal} had for the first time been justified in the thermodynamic limit, see \cite{Yin2010} (upper bound) and \cite{Sei2008} (lower bound). The inclusion of the remaining two (negative) terms in the upper bound for the free energy in \eqref{eq:UpperBoundFinal} is therefore our main new contribution.  
		\item The third term on the r.h.s. of \eqref{eq:UpperBoundFinal} is related to the particle number fluctuations in the BEC. Let us explain this in some more detail: it is well known that a c-number substitution for one momentum mode in the spirit of \cite{LiSeiYng2005,DeuSeiGP2020} (method of coherent states) introduces only a small correction to the free energy. Motivated by this, we use a trial state of the form
		\begin{equation}
			\Gamma_0 = \int_{\mathbb{C}} |z \rangle \langle z | \ p(z) \de z, \quad \text{ where } \quad |z \rangle = \exp( z a_0^* - \overline{z} a_0 ) | \mathrm{vac} \rangle
			\label{eq:TrialStateCondensate}
		\end{equation}
		to describe the BEC. Here $a_0^*$ and $a_0$ denote the usual creation and annihilation operators of a particle in the $ p = 0 $ mode and $| \mathrm{vac} \rangle$ is the related vacuum vector. Moreover, $p(z)$ is a probability distribution on $\mathbb{C}$ w.r.t. the measure $\de z$ defined below \eqref{eq:FreeEnergyBEC}. Let us assume that the interaction energy of the BEC is described by the effective Hamiltonian $4 \pi \mathfrak{a}_N L^{-3} a_0^* a^*_0 a_0 a_0$. The free energy of $\Gamma_0$ is then given by
		\begin{equation}
			\mathcal{F}^{\mathrm{BEC}}(\Gamma_0) = 4 \pi \mathfrak{a}_N L^{-3} \int_{\mathbb{C}} |z|^4 p(z) \de z - \frac{1}{\beta}S(\Gamma_0).
		\end{equation}
		From the the Berezin--Lieb inequality, see e.g. \cite{Berezin1972,Lieb1973}, we know that the last term on the r.h.s. is bounded from above by $-1/\beta$ times
		\begin{equation}
			S(p) = -\int_{\mathbb{C}} p(z) \ln(p(z)) \de z.
			\label{eq:ClassicalEntropyb}
		\end{equation}
		When we minimize $\mathcal{F}^{\mathrm{BEC}}(\Gamma_0)$ with $S(\Gamma_0)$ replaced by $S(p)$ under the constraint $\int |z|^2 p(z) \de z = N_0$ over all probability distributions $p$, we obtain $F^{\mathrm{BEC}}$ in \eqref{eq:FreeEnergyBEC}. The unique minimizer is the Gibbs distribution $g$ in \eqref{eq:GibbsDistribution}. With the above considerations, Proposition~\ref{prop:FreeEnergyBECb}.(a), and $\int_{\mathbb{C}} |z|^2 g(z) \de z = N_0$ we conclude that
		\begin{equation}
			4 \pi \mathfrak{a}_N L^{-3} \left( \int_{\mathbb{C}} |z|^4 g(z) \de z - \left( \int_{\mathbb{C}} |z|^2 g(z) \de z \right)^2 \right) - \frac{1}{\beta}S(g) = \frac{\ln \left( 16 \beta \mathfrak{a}_N/L^3 \right)}{2 \beta} + O\left( L^{-2} \exp\left(- c N^{\epsilon/2} \right) \right) 
			\label{eq:EnergyFluctuationsBEC}
		\end{equation}
		provided that $N_0 \geq N^{5/6 + \epsilon}$ holds for some fixed $\epsilon > 0$. That is, the term on the r.h.s. of the above equation indeed describes the free energy related to the particle number fluctuations in the BEC. It is interesting to note that this contribution vanishes in the thermodynamic limit because it is bounded from above by a constant times $\ln(N)/\beta$. 
		\item The Gibbs distribution $g$ in \eqref{eq:GibbsDistribution} satisfies 
		\begin{equation}
			\mathrm{Var}_g(|z|^2) =  \int_{\mathbb{C}} |z|^4 g(z) \de z - \left( \int_{\mathbb{C}} |z|^2 g(z) \de z \right)^2 \sim N^{5/3}
			\label{eq:condfluct}
		\end{equation}
		for $\kappa > 1$, which should be compared to the grand canonical ideal Bose gas. Here the same quantity is of the order $N^2$. This decrease of the number fluctuations in the BEC is a well-known effect caused by the repulsive interaction between the particles. Motivated by the recent experimental realization of a system with grand canonical number statistics, see \cite{Schmittetal2014}, a discrete version of $g$ in \eqref{eq:GibbsDistribution} has recently been used in \cite{Wurffetal2014} to compute the particle number fluctuations in an interacting grand canonical trapped BEC. To rigorously justify the computations in \cite{Wurffetal2014}, it is necessary to show that $g(z)$ approximates $\tr[ |z \rangle \langle z | G ]$ with the interacting Gibbs state $G$ in \eqref{eq:GibbsState}. This is a very interesting mathematical problem, whose solution is beyond the scope of the present investigation. 
		\item The term in the second line of \eqref{eq:UpperBoundFinal} is a correction to the free energy of the non-condensed particles coming from Bogoliubov theory. It can be motivated by the following heuristic computation: We write the Hamiltonian $\mathcal{H}_N$ in \eqref{eq:Hamiltonian} in terms of creation and annihilation operators $a_p$ and $a_p^*$ of a particle with momentum $p \in \Lambda^*$. Next we replace $a_0$ and $a_0^*$ by $\sqrt{N_0}$, and $\hat{v}(p)$ by $4 \pi \mathfrak{a}_N L^{-3}$. When we additionally neglect cubic and quartic terms in $a_p$ and $a^*_p$, we obtain the Bogoliubov Hamiltonian
		\begin{equation}
			\mathcal{H}^{\mathrm{Bog}} = \sum_{p \in \Lambda^*_+} p^2 a_p^* a_p + 4 \pi \mathfrak{a}_N \varrho_0(\beta,N,L) \sum_{p \in \Lambda^*_+} \left( 2 a_p^* a_p + a_p^* a_{-p}^* + a_p a_{-p} \right).
			\label{eq:BogoliubovHamiltonian}
		\end{equation}
        The above heuristics is also supported by non-rigorous arguments in the physics article \cite{LY1958}.
		A careful analysis shows that the grand potential $\Phi^{\mathrm{Bog}}(\beta,\mu_0,L)$ associated to $\mathcal{H}^{\mathrm{Bog}}$ with $\mu_0$ in \eqref{eq:ChemicalPotentialIdealGas} satisfies (compare to Lemma~\ref{lem:freeEnergyBog} in Appendix~\ref{app:BogFreeEnergy})
		\begin{align}
			\Phi^{\mathrm{Bog}}(\beta,\mu_0,L) =& \frac{1}{\beta} \sum_{p \in \Lambda_+^*} \ln\left( 1 - \exp\left( \beta | p^2 - \mu_0 | \sqrt{ p^2 - \mu_0 + 16 \pi \mathfrak{a}_N \varrho_0 } \right) \right) \nonumber \\
			=& \frac{1}{\beta} \sum_{p \in \Lambda_+^*} \ln\left( 1 - \exp\left( -\beta (p^2 - \mu_0) \right) \right) + 8 \pi \mathfrak{a}_N L^3 (\varrho - \varrho_0 ) \varrho_0 \nonumber \\
			&- \frac{1}{2 \beta} \sum_{p \in \Lambda^*_+} \left[ \frac{16 \pi \mathfrak{a}_N \varrho_0(\beta,N,L)}{p^2} - \ln\left( 1 + \frac{16 \pi \mathfrak{a}_N \varrho_0(\beta,N,L)}{p^2} \right) \right]  + o(L^{-2} N^{2/3}).
			\label{eq:ExpansionGrandPotential}
		\end{align}
		The first term on the r.h.s. contributes to $F_0^+$, the second term is part of the density-density interaction energy, and the third term is the novel contribution in the second line of \eqref{eq:UpperBoundFinal}.
		\item In \cite{BocBreCeSchl2019} it has been shown that eigenvalues $e_B$ of $\mathcal{H}_N^{(N)}-E_N$ (with $\mathcal{H}_N^{(N)}$ in \eqref{eq:Hamiltonian(n)} and $E_N$ its ground state energy) that satisfy $e_B \ll L^{-2} N^{1/8}$ are, to leading order, as $N \to \infty$, approximated by those of a Bogoliubov Hamiltonian. If we compare this energy scale to our  temperature $1/\beta \sim 1/\beta_{\mathrm{c}} \sim L^{-2} N^{2/3}$, which is a measure for the \textit{energy per particle} in our system, we see that the result in \cite{BocBreCeSchl2019} is far from being sufficient to draw conclusions about the free energy.
		\item It is interesting to note that if one replaces $\mathfrak{a}_N$ by $\mathfrak{a}$ and takes the thermodynamic limit ($N,L \to \infty$ with $\varrho = N/L^3$ fixed) of the last term in \eqref{eq:ExpansionGrandPotential} divided by $L^3$, one obtains 
		\begin{equation}
			-\frac{1}{2 \beta (2 \pi)^3} \int_{\mathbb{R}^3} \left[ \frac{16 \pi \mathfrak{a} \varrho_0}{p^2} - \ln\left( 1 + \frac{16 \pi \mathfrak{a} \varrho_0}{p^2} \right) \right] \de p = - \frac{16 \sqrt{\pi}}{3\beta} (\mathfrak{a} \varrho_0)^{3/2}. \label{eq:born}
		\end{equation}
		The r.h.s. has been conjectured to appear in the asymptotic expansion of the specific free energy in the dilute limit, see \cite[Theorem~11]{NapReuvSol2017}. There it is shown that the restricted minimization of the free energy functional \eqref{eq:GibbsFreeEnergyFunctional} over the class of quasi-free states leads to \eqref{eq:born} with the scattering length replaced by its first Born approximation. This is also true for the natural equivalent of the second term on the r.h.s. of \eqref{eq:UpperBoundFinal} in the thermodynamic limit.
		\item The dependence of the third term on the r.h.s. of \eqref{eq:MainFreeEnenergyBound} on $F_0^{\text{BEC}}$ is needed because $F^{\mathrm{BEC}}- 8 \pi \mathfrak{a}_N L^3 \varrho_0^2$  fails to describe the free energy of the $p=0$ mode correctly if $N_0 \sim 1$ ($\Leftrightarrow \kappa < 1$), that is, if there is no BEC.
%
%
This is also related to the fact that we describe the discrete random variable associated to the operator $a^*_0 a_0$ by one that is continuous.
		\item Theorem~\ref{thm:Main} is stated and proved for fixed $\kappa \in (0,\infty)$. Our proof can, however, easily be extended to cover the case when $\kappa$ depends on $N$ provided $\kappa \gtrsim 1$ holds.
		\item We expect the upper bound in Theorem~\ref{thm:Main} to be accurate up to a remainder of the order $o(L^{-2} N^{2/3})$. That is, we expect it to be possible to prove a matching lower bound. 
		\item In case of the canonical ensemble we expect that $F_0^+ + F^{\mathrm{BEC}}$ needs to replaced by $F_0^{\mathrm{c}} + 4 \pi \mathfrak{a}_N \varrho_0^2$, where $F_0^{\mathrm{c}}$ denotes the free energy of the canonical ideal gas. This is a consequence of the fact the variance of the number of condensed particles in the canonical ideal gas lives, for $\beta = \kappa \beta_{\mathrm{c}}$ with $\kappa > 1$, on the scale $N^{4/3}$. This needs to be compared to \eqref{eq:condfluct} and \eqref{eq:EnergyFluctuationsBEC}. For a thorough analysis of condensate fluctuations in the canonical ideal gas we refer to \cite{ChattDiac2014}.
		\end{enumerate}
\end{remark}

\subsection{Organization of the article}
\label{sec:ProofStrategy}
We prove Theorem~\ref{thm:Main} with a trial state argument. In Section~\ref{sec:TrialState} we define our trial state and establish some of its properties that are needed for proving an upper bound for its free energy. In Section~\ref{sec:energy}, which is the core of our analysis, we provide an upper bound for the energy of our trial state, and Section~\ref{entr} is devoted to an estimate for its entropy. Finally, in Section~\ref{sec:ProofOfThm} we use these results to give the proof of Theorem~\ref{thm:Main}. To not dilute the main line of the argument, we deferred some technical parts of our proof to an Appendix. In  Appedix~\ref{app:ScatteringEquation} we collect known properties of the solution to the scattering equation in a ball with Neumann boundary conditions. Appendix~\ref{app:BogFreeEnergy} contains the proof of an expansion of the free energy related to a Bogoliubov Hamiltonian in the spirit of \eqref{eq:ExpansionGrandPotential}. In Appendix~\ref{app:CondensateFreeEnergy} we prove Proposition~\ref{prop:FreeEnergyBECb} as well as some lemmas concerning $F^{\mathrm{BEC}}$ in \eqref{eq:FreeEnergyBEC} and $g$ in \eqref{eq:GibbsDistribution}. Finally, in Appendix~\ref{app:ExpectedParticleNumber} we give the proof of a lemma that allows us to estimate the influence of the correlation structure on the expected number of particles in our trial state.
\section{The trial state}
\label{sec:TrialState}
In this section we define our trial state and collect some of its properties. 
\subsection{Definition of the trial state}
\label{sec:ConstructionTrialState}
We start our analysis with the definition of the trial state. To be able to distinguish between different parts of the system as e.g. the condensate, thermally excited particles, and the microscopic correlations between the particles induced by $v_N$, we start by introducing several subsets of the momentum space $\Lambda^*$. Let $\delta_{\mathrm{B}}, \delta_{\mathrm{L}}, \delta_{\mathrm{H}} > 0$ with $\delta_{\mathrm{B}} < 1/3$ and $\delta_{\mathrm{L}} + \delta_{\mathrm{H}} < 2/3$ and define 
\begin{align}
	P_{\mathrm{L}} &:= \left\{p \in \Lambda^* \ | \ |p| \leq N^{1/3 + \delta_{\mathrm{L}}} /L \right\}, \nonumber \\
	P_{\mathrm{B}} &:= \left\{p \in \Lambda^* \ | \ 0 < |p| \leq N^{\delta_{\mathrm{B}}} /L \right\}, \nonumber \\
	P_{\mathrm{I}} &:= \left\{p \in \Lambda^* \ | \ N^{\delta_{\mathrm{B}}} / L < |p| \leq N^{1/3 + \delta_{\mathrm{L}}} /L \right\}, \nonumber \\
	P_{\mathrm{H}} &:= \left\{p \in \Lambda^* \ | \ |p| \geq N^{1 - \delta_{\mathrm{H}}} /L \right\}. \label{eq:PH}
\end{align}
Our assumptions on the parameters assure that $P_{\mathrm{B}} \subset P_{\mathrm{L}}$ and $P_{\mathrm{L}} \cap P_{\mathrm{H}} = \emptyset$. Later, the parameters $\delta_{\mathrm{B}}, \delta_{\mathrm{L}}$ and $\delta_{\mathrm{H}}$ will be chosen independently of $N$. The meaning of our sets in \eqref{eq:PH} is the following: the set $P_{\mathrm{L}}$ is appropriate to describe the BEC and the thermally excited particles described by our trial state. To that end, it is sufficient to choose $\delta_{\mathrm{L}}>0$ as small as we wish. For any $\delta_{\mathrm{B}} > 0$, the set $P_{\mathrm{B}}$ is large enough to describe the Bogoliubov excitations in the system. The part of the trial state with support in $\mathscr{F}(L^2(P_{\mathrm{I}}))$ will be chosen as the Gibbs state of an ideal gas. That is, for these modes Bogoliubov theory is not relevant. Finally, the microscopic correlations between the particles induced by the singular interaction $v_N$ will be chosen to live in the set $P_{\mathrm{H}}$. 

It is convenient for us to introduce the following decomposition of the bosonic Fock space:
\begin{equation}
	\mathscr{F}(L^2(\Lambda,\de x)) \cong \mathscr{F}_{0} \otimes \mathscr{F}_{\mathrm{B}} \otimes \mathscr{F}_{\mathrm{I}} \otimes \mathscr{ F }_{>},
	\label{eq:FockSpaceDecomposition}
\end{equation} 
where $\mathscr{F}_{0}$ denotes the Fock space over the $p=0$ mode, $\mathscr{F}_{\mathrm{B}}$ is the Fock space over $L^2(P_{\mathrm{B}})$, $\mathscr{F}_{\mathrm{I}}$ denotes the Fock space over $L^2(P_{\mathrm{I}})$, and $\mathscr{ F }_{>}$ is the Fock space over all remaining momentum modes. Moreover, by $\cong$ we denote unitary equivalence. In the following we will use, without explicitly mentioning it, the same symbol for an operator acting on $\mathscr{F}$ and for its unitary image acting on $\mathscr{F}_{0} \otimes \mathscr{F}_{\mathrm{B}} \otimes \mathscr{F}_{\mathrm{I}} \otimes \mathscr{ F }_{>}$.
By $a^*(g)$ and $a(g)$ we denote the usual creation and annihilation operators of a particle in the function $g \in L^2 (\Lambda, \de x)$, which satisfy the canonical commutation relations
\begin{equation}\label{eq:ccr} 
[a (g), a^* (h) ] = \langle g,h \rangle , \quad [ a(g), a(h)] = 0 = [a^* (g), a^* (h) ]. 
\end{equation}
We also use the notation $a_p = a(\varphi_p)$ with the plane wave $L^{-3/2} e^{\mathrm{i}px}$. In this special case the first identity in \eqref{eq:ccr} reads $[a_p,a_q^*] = \delta_{p,q}$.

We are now prepared to define our trial state and start by introducing the Bogoliubov Hamiltonian
\begin{equation}
	\mathcal{H}^{\mathrm{B}} = \sum_{p \in P_{\mathrm{B}}}(p^2 - \mu_0) a^*_{p} a_{p} + \frac{\varrho_0(\beta,N,L)}{2}\sum_{p \in P_{\mathrm{B}}} \hat{v}_N*\hat{f}_{N}(p) \left[2 a^*_{p} a_{p}+ (z/|z|)^2 a^*_{p} a^*_{-p}+ (\overline{z}/|z|)^2 a_{p} a_{-p} \right]
	\label{eq:HB}
\end{equation}
with $\mu_0$ in \eqref{eq:ChemicalPotentialIdealGas}, $\varrho_0$ in \eqref{eq:DensityBEC}, and the Fourier coefficients $\hat{f}_{N} =\int_\L e^{-\mathrm{i} p x} v_N(x) \de x$ of the solution $f_N(x)$ to a version of the zero energy scattering equation that will be introduced more carefully below. We also recall our definition of the convolution in \eqref{eq:Convolution}. By
\begin{equation}
	G_{\mathrm{B}}(z) = \frac{\exp(-\beta \mathcal{H}^{\mathrm{B}})}{\tr_{\mathscr{F}_{\mathrm{B}}}[\exp(-\beta \mathcal{H}^{\mathrm{B}})]}
	\label{eq:Gbog}
\end{equation}
we denote the Gibbs state related to $\mathcal{H}^{\mathrm{B}}$, which acts on $\mathscr{F}_{\mathrm{B}}$. We also introduce the Gibbs state of the ideal gas 
\begin{equation}
	G_{\mathrm{free}} = \frac{\exp(-\beta \de \Gamma( \mathds{1}( -\mathrm{i} \nabla \in P_{\mathrm{I}})(- \Delta - \mu_0 ))}{\tr_{\mathscr{F}_{\mathrm{I}}} \exp(-\beta \de \Gamma( \mathds{1}( -\mathrm{i} \nabla \in P_{\mathrm{I}})(- \Delta - \mu_0 ))},
	\label{eq:Gfree}
\end{equation}
acting on $\mathscr{F}_{\mathrm{I}}$. With these definitions at hand, we define the state $\Gamma_0$ without microscopic correlations between the particles by
\begin{equation}
	\Gamma_0 = \int_{\mathbb{C}} | z \rangle \langle z | \otimes G_{\mathrm{B}}(z) \zeta(z) \de z \otimes G_{\mathrm{free}}.
	\label{eq:modfreestate}
\end{equation}
Here $|z \rangle$ is the coherent state in \eqref{eq:TrialStateCondensate}. The probability distribution  $\zeta(z)$ is given by
\begin{equation}
		\zeta(z) = \frac{\exp\left( - \beta \left( 4 \pi \mathfrak{a}_N L^{-3} |z|^4 - \widetilde{\mu} |z|^2 \right) \right) }{ \int_{\mathbb{C}} \exp\left( - \beta \left( 4 \pi \mathfrak{a}_N L^{-3} |z|^4 - \widetilde{\mu} |z|^2 \right) \right) \de z } 
		\label{eq:GibbsDistributionzeta}
\end{equation}
i.e., it equals $g(z)$ in \eqref{eq:GibbsDistribution} except for the fact that the chemical potential $\widetilde{\mu}$ in the definition of $\zeta$ will be used to adjust the expected number of particles in our final trial state $\Gamma$ that we define below in \eqref{trialstate}. To be able to adjust the particle number correctly, we will need the assumption $N_0 \geq N^{2/3}$. This is related to the fact that the correlation structure we add later changes the particle number on the scale $N^{\delta_{\mathrm{H}}}$ ($\delta_{\mathrm{H}} < 2/3$). In the parameter regime where $N_0 < N^{2/3}$ holds, we use a much simpler trial state. This is discussed at the end of Section~\ref{sec:ProofOfThm}. We define $\widetilde{N}_0$ by
\begin{equation}
\widetilde{N}_0=\int_{\mathbb{C}} |z|^2 \zeta(z) \de z.
		\label{eq:WidetildeN0}
\end{equation}
In combination, Lemmas~\ref{lm:NG}, \ref{lem:BoundWidetildeN0} below and our choice of parameters below \eqref{eq:Sec67} show that $\widetilde{N}_0 $ equals $N_0$ in \eqref{eq:DensityBEC} up to a correction of the order $N^{2/3}$ if $N_0 \sim N$. In the computation of the free energy of our trial state we obtain the term $F^{\mathrm{BEC}}(\beta,\widetilde{N}_0,L,\mathfrak{a}_N)$. To replace this free energy by the same expression with $\widetilde{N}_0$ replaced by $N_0$, we use Lemma~\ref{lem:ChemPotBEC} in Appendix~\ref{app:CondensateFreeEnergy}. It is important to note that the difference between these two free energies yields a contribution of the order $L^{-2} N^{2/3}$. More details concerning this issue can be found in Section~\ref{sec:ProofOfThm} in the analysis following \eqref{eq:Sec64}.

The definition of the condensate part and the part of our trial state related to the Bogoliubov modes $p\in P_{\mathrm{B}}$ have been motivated in Remark~\ref{rmk:MainResults}.(b) and (d), respectively. For higher momenta the Bogoliubov dispersion relation $\sqrt{p^2 - \mu_0} \sqrt{p^2-\mu_0 + 16 \pi \mathfrak{a}_N \varrho_0}$ resembles $p^2-\mu_0$, to leading order. We find it therefore more convenient to describe the thermally excited particles with momenta in $P_{\mathrm{I}}$ by $G_{\mathrm{free}}$. The Bogoliubov Hamiltonian in \eqref{eq:HB} depends on $z/|z|$ because the condensate is described by the coherent state $|z \rangle \langle z |$. The complex phase $z/|z|$ will cancel out in the computation of the energy but its inclusion here is crucial for certain terms not to vanish. 

In the final step, we dress our trial state with a correlation structure that describes the microscopic correlations introduced by $v_N$. Let $f_{N}$ denote the ground state solution to the Neumann problem
\begin{equation}
	(-\Delta + v_N(x)/2) f_{N}(x) = \lambda_{N} f_{N}(x)
	\label{eq:NP}
\end{equation}
on the ball $B_{\ell}(0)$ with some fixed $0 < \ell < L/2$. 
We assume that $f_{N}$ is normalized such that it equals $1$ on $\partial B_{ \ell}(0)$ and we interpret it as a function on $\Lambda$ by extending it by $1$ outside of $B_{\ell}(0)$. Eq~\eqref{eq:NP} is a finite volume version of the zero energy scattering equation $\Delta f(x) = v(x) f(x) /2$ with boundary condition $\lim_{|x| \to \infty} f(x) = 1$ on $\mathbb{R}^3$. We also define 
\begin{equation}
	\eta_p = \hat{f}_{N}(p) - L^3 \delta_{p,0}.
	\label{eq:eta}
\end{equation}
More information on the functions $f_{N}$ and $\eta_p$ can be found in Appendix~\ref{app:ScatteringEquation}. With $\eta_p$ at hand, we define the two-body operator
\begin{equation} 
	B = \frac{1}{2L^3} \sum_{p \in P_{\mathrm{H}}, \ u,v \in P_{\mathrm{L}}} \eta_p\, a_{u+p}^*a_{v-p}^* a_u a_v
	\label{eq:B}
\end{equation}
on $\mathscr{F}$. Except for the restriction of the momenta, it is a multiplication operator with the inverse Fourier transform of the function $\eta_p$. We apply the spectral theorem to write $\Gamma_0 = \sum_{\alpha} \lambda_{\alpha} | \psi_{\alpha} \rangle \langle \psi_{\alpha} |$ and define our trial state $\Gamma$ by
\begin{equation}
	\Gamma = \sum_{\alpha=1}^{\infty} \lambda_{\alpha} | \phi_{\alpha} \rangle \langle \phi_{\alpha} |, \quad \text{ where } \quad \phi_{\alpha} = \frac{(1+B) \psi_{\alpha}}{\Vert (1+B) \psi_{\alpha} \Vert}.
	\label{trialstate}
\end{equation}
The general idea behind the way we introduce correlations is as follows: let us for the sake of simplicity consider an $N$-particle wave function $\psi$ that we want to dress. A natural way to introduce correlations is to multiply $\psi$ with a Jastrow factor $\prod_{i<j} f_{N}(x_i-x_j)$. When we write $f_{N} = 1-w_{N}$ and expand the product in powers of $w_{N}$, we obtain $(1 - \sum_{i<j} w_N(x_i-x_j)) \psi$ plus higher order contributions in $w_{N}$. Except for our momentum cut-offs, these first two terms equal $(1+B) \psi$. Since higher order corrections in $w_{N}$ are not necessary to obtain the correct energy in the GP limit, we omit these contributions. The restrictions of the momentum sums in the definition of $B$ turn out to be mathematically convenient. Intuitively, $p \in P_{\mathrm{H}}$ and $u,v \in P_{\mathrm{L}}$ because $\eta_p$ can be well approximated with momenta in $ P_{\mathrm{H}}$ and $G_{\mathrm{B}}(z)$ and $G_{\mathrm{free}}$ can be well approximated with momenta in $P_{\mathrm{L}}$. Correlation structures that are similar to the one introduced by $B$ have been used at zero temperature in \cite{BocBreCeSchl2019,BocBreCeSchl2020}. A similar approach to describe correlations can be found in \cite{YauYin2009,Yin2010}. 

The introduction of our correlation structure changes the expected number of particles in the trial state (slightly) because $\Gamma_0$ does not commute with $\mathcal{N}$. The following lemma provides us with a bound relating the expected number of particles in $\Gamma$ and $\Gamma_0$.


\begin{lemma}\label{lm:NG}
	We consider the limit $N \to \infty$, $\beta = \kappa \beta_{\mathrm{c}}$ with $\kappa \in (0,\infty)$ and $\beta_{\mathrm{c}}$ in \eqref{eq:BECPhaseTransition}. The bound
	\begin{align}\label{eq:NG}
		|\tr [\cN \G] -\tr [\cN \G_0] |  \lesssim N^{\delta_{\mathrm{H}}}
	\end{align}
	holds uniformly in $0 \leq \widetilde{N}_0 \leq C N$.
\end{lemma}
We recall that $0 < \delta_{\mathrm{H}} < 2/3$. The proof of the above lemma is based on simpler versions of the techniques that we use to prove our upper bound for the energy of $\Gamma$ in Section~\ref{sec:energy}, and we therefore prefer to give it in Appendix~\ref{app:ExpectedParticleNumber}. 

The above lemma quantifies the change in expected particle number caused by the correlation structure but it does not guarantee the existence of $\widetilde{\mu}$ with $\tr[\mathcal{N}\Gamma] = N$. This is because we are missing the information that $\tr[\mathcal{N}\Gamma]$ is a continuous function of $\widetilde{\mu}$. To circumvent this problem, we use the fact that the free energy in \eqref{eq:FreeEnergy} is, for fixed $v_N$ and $\beta$, a convex function of $N$. To see this, we first note that
\begin{equation}
	\frac{\partial F(\beta,N,L)}{\partial N} = \mu(\beta,N,L).
	\label{eq:defmu}
\end{equation}
Moreover, differentiation of both sides of the equation $\tr[\mathcal{N} G] = N$ with $G$ in \eqref{eq:GibbsState} with respect to $N$ yields
\begin{equation}
	\frac{\partial \mu(\beta,N,L)}{\partial N} = \frac{1}{\beta \left\{ \tr[\mathcal{N}^2 G] - (\tr[\mathcal{N} G])^2 \right\}} > 0.
	\label{eq:diffmu}
\end{equation}
In combination, \eqref{eq:defmu} and \eqref{eq:diffmu} show that the map $N \mapsto F(\beta,N,L)$ is convex. This implies the following statement: for given $|M| > 0$ it is always possible to satisfy $F(\beta,N + M ,L) \geq F(\beta,N,L)$ by choosing the correct sign for $M$. Motivated by this, we choose $\widetilde{\mu}$ such that $\tr[\mathcal{N} \Gamma_0]= N + \Delta N$, where $|\Delta N|$ is twice as large as the error term in \eqref{eq:NG} and the sign of $\Delta N$ is chosen such that the free energy is increased. This is possible because the remainder in \eqref{eq:NG} is uniform in $0 \leq \widetilde{N}_0 \leq CN$ and $\widetilde{N}_0 \geq 0$ can be chosen arbitrarily by varying $\widetilde{\mu}$. A bound relating our choice of $\widetilde{N}_0$ and $N_0$ is provided by Lemma~\ref{lem:BoundWidetildeN0} below. 

In the remainder of this article, we prove an upper bound for the free energy of $\Gamma$ that implies Theorem~\ref{thm:Main}.

\subsection{Preparatory lemmas}
\label{eq:PreparatoryLemmas}
In this section we state and prove several lemmas that are needed for the computation of the free energy of $\Gamma$ in \eqref{trialstate}. We present them here to not interrupt the main line of the argument in Section \ref{sec:energy}.

\subsubsection*{Properties of the state $G_{\mathrm{B}}(z) \otimes G_{\mathrm{free}}$}

The first lemma provides us with a Bogoliubov transformation that diagonalizes the Bogoliubov Hamiltonian $\mathcal{H}^{\mathrm{B}}$ in \eqref{eq:HB}. Before we state it, we introduce the following notation. For fixed $z \in \mathbb{C}$ and $p \in \Lambda^*$, we define the functions $\varphi_{p,z}(x) = (z/|z|) L^{-3/2} e^{\mathrm{i} p x}$, which are planes waves with a $z$-dependent phase. By $a^*_{p,z} = a^*(\varphi_{p,z})$ and $a_{p,z} = a(\varphi_{p,z})$ we denotes the operators that create and annihilate a particle in the function $\varphi_{p,z}$, respectively. Since $\{ \varphi_{p,z} \}_{p \in 2 \pi \mathbb{Z}/L}$ is an orthonormal basis of $L^2(\Lambda, \de x)$ the operators $a^*_{p,z}$ and $a_{p,z}$ satisfy the canonical commutation relations in the form stated below \eqref{eq:ccr}. We also define the (unitary) Bogoliubov transformation $\mathcal{U}_z: \mathscr{F}\to \mathscr{F}$ (up to a global phase) by its action on our $z$-dependent creation and annihilation operators in the following way ($p \in P_{\mathrm{B}}$):

\begin{equation}
	\mathcal{U}_z^* a^*_{p,z} \mathcal{U}_z = u_p a^*_{p,z} + v_p a_{-p,z}, \quad \quad \mathcal{U}_z^* a_{p,z} \mathcal{U}_z = u_p a_{p,z} + v_p a^*_{-p,z}. 
	\label{eq:Btr}
\end{equation}
The $z$-independent coefficients $u_p$ and $v_p$ are defined by
\begin{align}
	u_{p} &= \frac{1}{2} \left( \frac{p^2-\mu_0}{p^2-\mu_0 + 2 \hat{v}_N \ast \hat{f}_{N}(p) \varrho_0} \right)^{1/4} + \frac{1}{2} \left( \frac{p^2-\mu_0}{p^2-\mu_0 + 2 \hat{v}_N \ast \hat{f}_{N}(p) \varrho_0} \right)^{-1/4} \quad \text{ and } \nonumber \\
	v_{p} &= \frac{1}{2} \left( \frac{p^2-\mu_0}{p^2-\mu_0 + 2 \hat{v}_N \ast \hat{f}_{N}(p) \varrho_0} \right)^{1/4} - \frac{1}{2} \left( \frac{p^2-\mu_0}{p^2-\mu_0 + 2 \hat{v}_N \ast \hat{f}_{N}(p) \varrho_0} \right)^{-1/4} 
	\label{eq:coefficientsBogtrafo}
\end{align}
with $\mu_0$ in \eqref{eq:ChemicalPotentialIdealGas} and $\varrho_0$ in \eqref{eq:DensityBEC}, respectively. The function $\hat{v}_N \ast \hat{f}_{N}(p)$ may take negative values. However, we claim that there exists $\widetilde{N} \in \mathbb{N}$ such that it is nonnegative uniformly in $p \in P_{\mathrm{B}}$ provided $N \geq \widetilde{N}$. To prove this, we note that
\begin{equation}
	\hat{v}_N \ast \hat{f}_{N}(p) \geq \hat{v}_N \ast \hat{f}_{N}(0) - |p| \int_0^1 | \nabla \hat{v}_N \ast \hat{f}_{N}(tp) | \de t \geq \hat{v}_N \ast \hat{f}_{N}(0) - N^{\delta_{\mathrm{B}}} L^{-1} \int_{\Lambda} v_N(x) f_{N}(x) |x| \de x.
	\label{eq:Posititvity}
\end{equation}
By Lemma~\ref{lemma:NPprop} we know that $0 \leq f_{N} \leq 1$; we use this and $\int v_N(x) N |x| \de x/L = N^{-1} \Vert \ | \cdot | v \Vert_1$ to see that the last term on the r.h.s.\ is bounded by a constant times $L N^{\delta_{\mathrm{B}} - 2}$. Since $0<\delta_{\mathrm{B}} < 1/3$ by assumption, and $\hat{v}_N \ast \hat{f}_{N}(0) =N^{-1}\int  v(x)f(x) \de x\gtrsim L N^{-1}$ by \eqref{eq:sl}, the claim is proved.  In particular, it assures that $u_p$ and $v_p$ are well defined. In the following we will always assume that $N \geq \widetilde{N}$, and hence $\hat{v}_N \ast \hat{f}_{N}(p) \geq 0$ for $p \in P_{\mathrm{B}}$. 

We are now prepared to state our first lemma.
\begin{lemma}
	\label{lem:DiagBogoH}
	The Bogoliubov Hamiltonian $\mathcal{H}^{\mathrm{B}}$ in \eqref{eq:HB} satisfies
	\begin{equation}
		\mathcal{U}_z^* \mathcal{H}^{\mathrm{B}} \mathcal{U}_z = E_0 + \sum_{p \in P_{\mathrm{B}}} \epsilon(p) a_p^* a_p
		\label{diag}
	\end{equation}
	with 
	\begin{equation}
		\epsilon(p) = \sqrt{p^2 - \mu_0} \sqrt{ p^2 - \mu_0 + 2 \hat{v}_N \ast \hat{f}_{N}(p) \varrho_0 } \quad \text{ and } \quad E_0 = -\frac{1}{2} \sum_{p \in P_{\mathrm{B}}} \left[ p^2 - \mu_0 + \varrho_0 \hat{v}_N \ast \hat{f}_{N}(p) - \epsilon(p) \right].
		\label{eq:BogoDispersion}
	\end{equation}  
\end{lemma}
\begin{proof}
	To see that the Bogoliubov transformation $\mathcal{U}_z$ diagonalizes $\mathcal{H}^{\mathrm{B}}$, we note that the latter can be written as
	\begin{equation}
		\mathcal{H}^{\mathrm{B}} = \sum_{p \in P_{\mathrm{B}}}(p^2 - \mu_0) a^*_{p,z} a_{p,z} + \frac{\varrho_0(\beta,N,L)}{2}\sum_{p \in P_{\mathrm{B}}} \hat{v}_N*\hat{f}_{N}(p) \left[2 a^*_{p,z} a_{p,z}+ a^*_{p,z} a^*_{-p,z}+ a_{p,z} a_{-p,z} \right].
		\label{eq:HBz}
	\end{equation}
	Eq.~\eqref{diag} now follows from a standard computation that uses \eqref{eq:Btr} and \eqref{eq:HBz}, see e.g. \cite[Lemma~5.2]{BocBreCeSchl2020b}.
\end{proof}
\begin{remark}
	It is worth noting that, although $\mathcal{H}^{\mathrm{B}}$ depends on $z \in \mathbb{C}$, the r.h.s. of \eqref{diag} is independent of $z$. This is related to the fact that the $z$-dependence of $\mathcal{H}^{\mathrm{B}}$ is quite simple: all functions of the plane wave basis are multiplied by the same complex phase.
\end{remark}
Next, we compute the 1-pdm and the pairing function of the state $G_{\mathrm{B}}(z) \otimes G_{\mathrm{free}}$. 
\begin{lemma}
	\label{lem:1pdm}
	The 1-pdm and the pairing function of the state $G_{\mathrm{B}}(z) \otimes G_{\mathrm{free}}$ with $G_{\mathrm{B}}(z)$ in \eqref{eq:Gbog} and $G_{\mathrm{free}}$ in \eqref{eq:Gfree} are for $p,q \in P_{\mathrm{B}} \cup P_{\mathrm{I}}$ given by
	\begin{equation}
		\tr_{\mathscr{F}_{\mathrm{B}} \otimes \mathscr{F}_{\mathrm{I}}} [ a_q^* a_p G_{\mathrm{B}}(z) \otimes G_{\mathrm{free}} ] = \delta_{p,q} \gamma(p) \quad \text{ and } \quad \tr_{\mathscr{F}_{\mathrm{B}} \otimes \mathscr{F}_{\mathrm{I}}} [ a_p a_q G_{\mathrm{B}}(z) \otimes G_{\mathrm{free}} ] = \delta_{p,-q} (z/|z|)^2 \alpha(p),
		\label{eq:TI}
	\end{equation}
	respectively. Here
	\begin{align}
		\gamma(p) &= \mathds{1}(p \in P_{\mathrm{B}}) \left( (u_p^2 + v_p^2) \frac{1}{\exp(\beta \epsilon(p))-1}  + v_p^2 \right) + \mathds{1}(p \in P_{\mathrm{I}}) \frac{1}{\exp(\beta (p^2-\mu_0))-1} \quad \text{and} \nonumber \\
		\alpha(p) &= \mathds{1}(p \in P_{\mathrm{B}}) u_p v_p \left( \frac{2}{\exp(\beta \epsilon(p)) - 1} + 1 \right) \label{eq:defGam}
	\end{align}
	with $\epsilon(p)$ in \eqref{eq:BogoDispersion}.
\end{lemma}
\begin{proof}
	We start by noting that the special form of the 1-pdm and the pairing function in \eqref{eq:TI} follows immediately from the translation invariance of the state $G_{\mathrm{B}}(z) \otimes G_{\mathrm{free}}$. To compute $\gamma(p)$, we write
	\begin{align}
		\tr_{\mathscr{F}_{\mathrm{B}} \otimes \mathscr{F}_{\mathrm{I}}} [ a_q^* a_p G_{\mathrm{B}}(z) \otimes G_{\mathrm{free}} ] = \tr_{\mathscr{F}_{\mathrm{B}} \otimes \mathscr{F}_{\mathrm{I}}} [ \mathcal{U}_z^* a_{q}^* a_{p} \mathcal{U}_z \mathcal{U}_z^* G_{\mathrm{B}}(z) \mathcal{U}_z \otimes G_{\mathrm{free}} ].
		\label{eq:Comp1}
	\end{align}
	Using $a_{p} = (z/|z|) a_{p,z}$ and Lemma~\ref{lem:DiagBogoH}, we see that
	\begin{align}
		\tr_{\mathscr{F}_{\mathrm{B}} \otimes \mathscr{F}_{\mathrm{I}}} [ a_{q}^* a_{p} \mathcal{U}_z^* G_{\mathrm{B}}(z) \mathcal{U}_z ] &= \delta_{p,q} \frac{1}{\exp(\beta \epsilon(p)-1}, \nonumber \\ \tr_{\mathscr{F}_{\mathrm{B}} \otimes \mathscr{F}_{\mathrm{I}}} [ a_{q}^* a_{p}^* \mathcal{U}_z^* G_{\mathrm{B}}(z) \mathcal{U}_z ] &= 0 = \tr_{\mathscr{F}_{\mathrm{B}} \otimes \mathscr{F}_{\mathrm{I}}} [ a_{q} a_{p} \mathcal{U}_z^* G_{\mathrm{B}}(z) \mathcal{U}_z ]
		\label{eq:Comp2}
	\end{align}
	holds for $p,q \in P_{\mathrm{B}}$. We also have
	\begin{align}
		\tr_{\mathscr{F}_{\mathrm{B}} \otimes \mathscr{F}_{\mathrm{I}}} [ a_{q}^* a_{p}  G_{\mathrm{free}} ] &= \delta_{p,q} \frac{1}{\exp(\beta (p^2 - \mu_0)-1}, \nonumber \\ \tr_{\mathscr{F}_{\mathrm{B}} \otimes \mathscr{F}_{\mathrm{I}}} [ a_{q}^* a_{p}^*  G_{\mathrm{free}} ] &= 0 = \tr_{\mathscr{F}_{\mathrm{B}} \otimes \mathscr{F}_{\mathrm{I}}} [ a_{q} a_{p} G_{\mathrm{free}} ],
		\label{eq:Comp3}
	\end{align}
	for $p,q \in P_{\mathrm{I}}$. Since $G_{\mathrm{B}}(z) \otimes G_{\mathrm{free}}$ is a quasi-free state we know that the expectations in \eqref{eq:TI} vanish if one momentum is in $P_{\mathrm{B}}$ and the other in $P_{\mathrm{I}}$. When we use \eqref{eq:Btr}, $a_{p} = (z/|z|) a_{p,z}$, \eqref{eq:Comp2} and \eqref{eq:Comp3} on the r.h.s. of \eqref{eq:Comp1}, we obtain the claimed formula for the 1-pdm. The formula for the pairing function follows from a similar computation.
\end{proof}

We highlight that the pairing function of $G_{\mathrm{B}}(z) \otimes G_{\mathrm{free}}$ depends on $z/|z|$, while the 1-pdm does not. We state now a result useful to estimate momentum sums. Its proof can be found in \cite[Lemma~3.3]{DeuSeiGP2020}.
\begin{lemma}
	\label{lem:RiemannSum}
	Let $f:[0,\infty) \to \mathbb{R}$ be a nonnegative and monotone decreasing function and choose some $\kappa \geq 0$. Then
	\begin{equation}
		\sum_{p \in \Lambda_+^*} f(|p|) \mathds{1}(|p| \geq \kappa) \leq \left( \frac{L^3}{2\pi} \right) \int_{|p| \geq \left[\kappa - \sqrt{3} \frac{2\pi}{L} \right]_+} f(|p|) \left( 1 + \frac{3 \pi}{L |p|} + \frac{6 \pi}{L^2 p^2} \right) \de p.		
		\label{eq:RiemannSum}
	\end{equation}
\end{lemma}
The next lemma provides us with bounds for the functions $\gamma$ and $\alpha$.
\begin{lemma}
	\label{lem:boundsgamma}
	The functions $\gamma$ and $\alpha$ in \eqref{eq:TI} satisfy the pointwise bounds 
	\begin{align}
		\gamma(p) &\lesssim \mathds{1}(p \in P_{\mathrm{L}} \backslash\{ 0 \}) \frac{1}{\exp(\beta p^2)-1} + \mathds{1}(p \in P_{\mathrm{B}}) \frac{1}{L^4 p^4} \quad \text{ and } \nonumber \\ 
		|\alpha(q)| &\lesssim \mathds{1}(p \in P_{\mathrm{B}}) \frac{1}{L^2 p^2} \left( 1 + \frac{1}{\beta p^2} \right).
		\label{eq:Comp4}
	\end{align}
	Moreover, for $n \in \{ 0,1 \}$ we have
	\begin{align}
		\sum_{p \in \Lambda_+^*} |p|^{n} \gamma(p) &\lesssim L^3 \beta^{-(3+n)/2} + L^{-n} c_n(N) \quad \text{ with } \quad c_n(N) = \begin{cases}
			1 & \text{ if } n = 0 \\ \ln(N) & \text{ if } n=1 
		\end{cases}  \quad \text{ as well as } \nonumber \\
	 	\sum_{p \in \Lambda_+^*} |p|^{n} |\alpha(p)| &\lesssim L^{-n+2} \beta^{-1} c_n(N) + L^{-n} N^{ (n+1) \delta_{\mathrm{B}}}.
	 	\label{eq:Comp4b}
	\end{align}
	Finally, the number of particles with momenta in $P_{\mathrm{B}}$ is bounded by
	\begin{equation}
		\sum_{p \in P_{\mathrm{B}}} \gamma(p) \lesssim 1 + \frac{L^2 N^{\delta_{\mathrm{B}}}}{\beta}.
		\label{eq:Comp4c}
	\end{equation}
\end{lemma}
\begin{proof}
	We start by noting that 
	\begin{equation}
		v_p^2 = \frac{1}{4} \left( \frac{p^2-\mu_0}{p^2-\mu_0 + 2 \hat{v}_N \ast \hat{f}_{N}(p) \varrho_0} \right)^{1/2} + \frac{1}{4} \left( \frac{p^2-\mu_0}{p^2-\mu_0 + 2 \hat{v}_N \ast \hat{f}_{N}(p) \varrho_0} \right)^{-1/2} - \frac{1}{2}.
		\label{eq:Comp5}
	\end{equation}
	As already remarked above, we can assume that $2 \hat{v}_N \ast \hat{f}_{N}(p) \geq 0$ holds uniformly in $p \in P_{\mathrm{B}}$ (see \eqref{eq:Posititvity}). In combination with the bound $0 \leq (1+x)^{-1/2} + (1+x)^{1/2} - 2 \leq x^2/4$ for $x \geq 0$ and $\mu_0 < 0$, this implies
	\begin{equation}
		v_p^2 \leq \frac{(\varrho_0 \hat{v}_N \ast \hat{f}_{N}(p))^2}{4 p^4 }.
		\label{eq:Comp6}
	\end{equation}
	Using $0 \leq f_{N} \leq 1$, we see that
	\begin{equation}
		| \hat{v}_N \ast \hat{f}_{N}(p) | \leq \int_{\Lambda} v_N(x) f_{N}(x) \de x \leq N^{-1} \int_{\Lambda} v(x) \de x,
		\label{eq:Comp7}
	\end{equation}
	and hence
	\begin{equation}
		v_p^2 \lesssim \frac{N^2_0}{N^2 L^4 p^4} \leq \frac{1}{L^4 p^4}.
		\label{eq:boundvp}
	\end{equation}
	The bounds for $\gamma(p)$ and $|\alpha(p)|$ now follow from \eqref{eq:defGam}, \eqref{eq:boundvp}, $u_p^2 - v_p^2 = 1$ and $\epsilon(p) \geq p^2 - \mu_0 \geq p^2$. The bounds in \eqref{eq:Comp4b} and \eqref{eq:Comp4c} are a direct consequence of the pointwise bounds for $\gamma(p)$ and $|\alpha(p)|$ and Lemma~\ref{lem:RiemannSum}.
\end{proof}

\subsubsection*{Properties of the state $\Gamma_0$}

Recall definition \eqref{eq:WidetildeN0} for $\widetilde{N}_0$, the expected number of particles in the condensate of our trial state $\Gamma_0$. Recall also that $\tr[\mathcal{N}\Gamma_0] = N+\Delta N$ with $| \Delta N | \lesssim N^{\delta_{\mathrm{H}}}$. We highlight that if we know $\widetilde{N}_0$ we can compute the chemical potential $\widetilde{\mu}$ in the definition of $\zeta$, as discussed below \eqref{eq:modfreestate}. 
%
In the following lemma, we prove a bound for $\widetilde{N}_0$ showing that it is close to $N_0$ in \eqref{eq:DensityBEC} in a suitable sense.

\begin{lemma}
	\label{lem:BoundWidetildeN0}
	Assume that $\beta \gtrsim \beta_{\mathrm{c}}$. There exists a constant $c>0$ such that $\widetilde{N}_0$ satisfies the bound
	\begin{equation}
		| \widetilde{N}_0 - N_0 | \lesssim N^{\delta_{\mathrm{H}}} + \frac{N_0 L^2}{N \beta} + \frac{N_0^2}{N^2} + \exp(- c N^{2 \delta_{\mathrm{L}}}).
		\label{eq:BoundWidetildeN0}
	\end{equation}
\end{lemma}
\begin{proof}
	The expected number of particles in the state $\Gamma_0$ equals $N +\Delta N$, that is,
	\begin{equation}
		N + \Delta N= \int_{\mathbb{C}} \tr[\mathcal{N} \ |z \rangle \langle z | \otimes G_{\mathrm{B}}(z) \otimes G_{\mathrm{free}}] \zeta(z) \de z = \int_{\mathbb{C}} |z|^2 \zeta(z) \de z + \sum_{p \in \Lambda_+^*} \gamma(p),
		\label{eq:Comp8}
	\end{equation}
	where we used Lemma~\ref{lem:1pdm} to obtain the second identity. We apply Lemma~\ref{lem:1pdm} and the identity $u_p^2 - v_p^2 = 1$ to see that the part of the sum on the r.h.s. that runs over $P_{\mathrm{B}}$ can be written as
	\begin{equation}
		\sum_{p \in P_{\mathrm{B}}} \gamma(p) = \sum_{p \in P_{\mathrm{B}}} \frac{1}{\exp(\beta \epsilon(p))-1} + 2 \sum_{p \in P_{\mathrm{B}}} \frac{1}{\exp(\beta \epsilon(p))-1} v_p^2 + \sum_{p \in P_{\mathrm{B}}} v_p^2
		\label{eq:Comp8a}
	\end{equation}
	with $\epsilon(p)$ in \eqref{eq:BogoDispersion} and $v_p^2$ in \eqref{eq:Comp5}. We use $\epsilon(p) \geq p^2$, $(\exp(x)-1)^{-1} \leq 1/x$, and the bound for $v_p^2$ in \eqref{eq:boundvp} to see that the second term on the r.h.s. is bounded by
	\begin{equation}
		2 \sum_{p \in P_{\mathrm{B}}} \frac{1}{\exp(\beta \epsilon(p))-1} v_p^2 \lesssim \sum_{p \in \Lambda_+^*} \frac{1}{\beta p^2} \frac{N_0^2}{N^2 L^4 p^4} \lesssim \frac{N_0^2 L^2}{N^2 \beta}.
		\label{eq:Comp8b}
	\end{equation}
	Moreover, for the third term
	\begin{equation}
		\sum_{p \in P_{\mathrm{B}}} v_p^2 \lesssim \frac{N_0^2}{N^2}
		\label{eq:Comp8c}
	\end{equation}
	holds.
	
	We also claim that 
	\begin{equation}
		\left| \sum_{p \in P_{\mathrm{B}}} \left( \frac{1}{\exp(\beta \epsilon(p))-1} - \frac{1}{\exp(\beta(p^2-\mu_0))-1} \right) \right| \lesssim \frac{N_0 L^2}{N \beta}.
		\label{eq:Comp9}
	\end{equation}  
	To see this, we write
	\begin{equation}
		\frac{1}{\exp(\beta \epsilon(p))-1} = \frac{1}{\exp(\beta(p^2-\mu_0))-1} - \int_0^1 \frac{ \beta (p^2-\mu_0) \left( \sqrt{ 1 + \frac{2 \varrho_0 \hat{v}_N \ast f_{N}(p)}{p^2-\mu_0} } -1 \right)}{4 \sinh^2\left( ( t(p^2-\mu_0) + (1-t) \epsilon(p) ) /2 \right)} \de t.
		\label{eq:Comp10}
	\end{equation}
	Using $|\sqrt{1+x} - 1| \leq x/2$ for $x \geq 0$ and \eqref{eq:Comp7}, we check that
	\begin{equation}
		\left| \sqrt{ 1 + \frac{2 \varrho_0 \hat{v}_N \ast f_{N}(p)}{p^2-\mu_0} } -1 \right| \leq \frac{\varrho_0 \hat{v}_N \ast f_{N}(p)}{p^2-\mu_0} \leq \frac{2 \varrho_0 \Vert v \Vert_1}{N ( p^2 - \mu_0 )} \lesssim \frac{N_0}{N L^2 (p^2-\mu_0)}.
		\label{eq:Comp11}
	\end{equation} 
	In combination, \eqref{eq:Comp10}, \eqref{eq:Comp11}, $\epsilon(p) \geq p^2 - \mu_0$ and $\mu_0 < 0$ imply
	\begin{equation}
		\left| \sum_{p \in P_{\mathrm{B}}} \left( \frac{1}{\exp(\beta \epsilon(p))-1} - \frac{1}{\exp(\beta(p^2-\mu_0))-1} \right) \right| \lesssim \frac{N_0 \beta}{N L^2} \sum_{p \in  P_{\mathrm{B}}} \frac{1}{\sinh^2( \beta (p^2 - \mu_0)/2 )} \leq \frac{N_0}{N \beta L^2} \sum_{p \in \Lambda_+^*} \frac{1}{p^4},
		\label{eq:Comp12}
	\end{equation}
	which proves \eqref{eq:Comp9}. 
	
	When we put \eqref{eq:Comp8}--\eqref{eq:Comp8c} and \eqref{eq:Comp9} together and use \eqref{eq:defGam}, we find
	\begin{equation}
		N + \Delta N= \widetilde{N}_0 + \sum_{p \in P_{\mathrm{L}} \backslash \{ 0 \}} \frac{1}{\exp(\beta(p^2 - \mu_0))} + O\left( \frac{N_0 L^2}{N \beta} \right)
		\label{eq:Comp13}
	\end{equation}
	with $\widetilde{N}_0$ in \eqref{eq:WidetildeN0}. The second term on the r.h.s. can be written as
	\begin{equation}
		\sum_{p \in P_{\mathrm{L}} \backslash \{ 0 \}} \frac{1}{\exp(\beta(p^2 - \mu_0))-1} = \sum_{p \in \Lambda_+^*} \frac{1}{\exp(\beta(p^2 - \mu_0))-1} - \sum_{p \in P_{\mathrm{L}}^{\mathrm{c}}} \frac{1}{\exp(\beta(p^2 - \mu_0))-1},
		\label{eq:Comp14}
	\end{equation}
	where $P_{\mathrm{L}}^{\mathrm{c}}$ denotes the complement of the set $P_{\mathrm{L}}$. The first term on the r.h.s. equals $N - N_0$ with $N_0$ in \eqref{eq:DensityBEC} and the second term satisfies the bound
	\begin{equation}
		\sum_{p \in P_{\mathrm{L}}^{\mathrm{c}}} \frac{1}{\exp(\beta(p^2 - \mu_0))-1} \leq \left( \frac{1}{\exp\left(\beta N^{2/3 + 2 \delta_{\mathrm{L}}}\right)-1} \right)^{1/2} \sum_{p \in \Lambda_+^*} \left( \frac{1}{\exp(\beta(p^2 - \mu_0))-1} \right)^{1/2} \lesssim \exp( -c N^{2 \delta_{\mathrm{L}}} ) N
		\label{eq:Comp15}
	\end{equation}
	for some $c>0$. To obtain \eqref{eq:Comp15}, we used $\beta \gtrsim \beta_{\mathrm{c}}$ and the definition of $P_{\mathrm{L}}$ in \eqref{eq:PH}. When we put \eqref{eq:Comp13}--\eqref{eq:Comp15} together, and use  $|\Delta N| \lesssim N^{\delta_{\mathrm{H}}}$ as well as the assumption $\delta_{\mathrm{L}} > 0$, we obtain a proof of \eqref{eq:BoundWidetildeN0}.
\end{proof}

For the computation of the energy of $\Gamma_0$ we need to know its 2-particle density matrix (2-pdm), which is stated in the next lemma.

\begin{lemma}
\label{lem:1-pdmGamma0}
The 2-pdm of the state $\Gamma_0$ in \eqref{eq:modfreestate} reads
\begin{align} 
\tr_{\fock} \big[a_{u_1}^*a_{v_1}^* a_{u_2} &a_{v_2} \Gamma_0 \big] = \delta_{u_1,0}\delta_{v_1,0}\delta_{u_2,0}\delta_{v_2,0} \int_{\mathbb{C}} |z |^4 \zeta(z) \de z  \nonumber \\
&+\widetilde{N}_0 \left[ \gamma(v_1) \delta_{v_1,v_2}\delta_{u_1,0}\delta_{u_2,0} +\gamma(u_1)\delta_{u_1,u_2}\delta_{v_1,0}\delta_{v_2,0} +\gamma(u_1)\delta_{u_1,v_2}\delta_{v_1,0}\delta_{u_2,0}+\gamma(v_1)\delta_{v_1,u_2}\delta_{u_1,0}\delta_{v_2,0} \right] \nonumber \\
&+ \widetilde{N}_0 \left[ \alpha(u_2) \delta_{u_2,-v_2}\delta_{u_1,0}\delta_{v_1,0} + \overline{\alpha(u_1)} \delta_{u_1,-v_1}\delta_{u_2,0}\delta_{v_2,0}  \right] \nonumber \\
&+\gamma(u_1)\gamma(v_1)\delta_{u_1,u_2}\delta_{v_1,v_2}+\gamma(u_1)\gamma(v_1)\delta_{u_1,v_2}\delta_{v_1,u_2} +\overline{\alpha(u_1)}\alpha(u_2)\delta_{u_1,-v_1}\delta_{u_2,-v_2}
\label{lm:GBGF}
\end{align}
with $\widetilde{N}_0$ in \eqref{eq:WidetildeN0} and $\gamma$, $\alpha$ in \eqref{eq:defGam}.
\end{lemma}

\begin{proof}
We denote by $\mathcal{W}_z=\exp( z a_0^* - \overline{z} a_0 )$ the Weyl transformation that implements the condensate. Using $\mathcal{W}_z^* a_0 \mathcal{W}_z = a_0 + z$, we find
\begin{equation} 
\tr_{\fock} [a_{u_1}^* a_{v_1}^* a_{u_2} a_{v_2}\Gamma_0 ] = \int_{\mathbb{C}} \tr_{\mathscr{F}_{\mathrm{B}} \otimes \mathscr{F}_{\mathrm{I}}}[A_{u_1, v_1, u_2,v_2} G_{\mathrm{B}}(z) \otimes G_{\mathrm{free}}] \zeta(z) \de z
\label{eq:gbgf}
\end{equation}
with the operator 
\begin{align}
A_{u_1, v_1, u_2,v_2}=&|z|^4\delta_{u_1,0}\delta_{v_1,0}\delta_{u_2,0}\delta_{v_2,0} \nonumber \\
&+|z|^2\Big(a_{v_1}^*a_{v_2}\delta_{u_1,0}\delta_{u_2,0}+a_{u_1}^*a_{u_2}\delta_{v_1,0}\delta_{v_2,0}+a_{u_1}^*a_{v_2}\delta_{v_1,0}\delta_{u_2,0}+a_{v_1}^*a_{u_2}\delta_{u_1,0}\delta_{v_2,0}\Big) \nonumber \\
&+\overline{z}^2a_{u_2} a_{v_2}\delta_{u_1,0}\delta_{v_1,0}+z^2a_{u_1}^* a_{v_1}^*\delta_{u_2,0}\delta_{v_2,0}+a_{u_1}^* a_{v_1}^* a_{u_2} a_{v_2}. 
\label{eq:Comp16}
\end{align}
An application of Lemma~\ref{lem:1pdm} allows us to compute the terms proportional to $|z|^2$, $z^2$ and $\overline{z}^2$. It remains to compute the expectation of the last term in \eqref{eq:Comp16}. Since $G_{\mathrm{B}}(z) \otimes G_{\mathrm{free}}$ is a quasi-free state we can apply the Wick theorem and find
\begin{align}
\tr_{\mathscr{F}_{\mathrm{B}} \otimes \mathscr{F}_{\mathrm{I}}}\big[a_{u_1}^* a_{v_1}^* a_{u_2} a_{v_2} G_{\mathrm{B}}(z) \otimes G_{\mathrm{free}} \big] =& \tr_{\mathscr{F}_{\mathrm{B}} \otimes \mathscr{F}_{\mathrm{I}}}\big[a_{u_1}^* a_{u_2} G_{\mathrm{B}}(z) \otimes G_{\mathrm{free}}\big] \tr_{\mathscr{F}_{\mathrm{B}} \otimes \mathscr{F}_{\mathrm{I}}}\big[ a_{v_1}^* a_{v_2} G_{\mathrm{B}}(z) \otimes G_{\mathrm{free}} \big] \nonumber \\
&+\tr_{\mathscr{F}_{\mathrm{B}} \otimes \mathscr{F}_{\mathrm{I}}}\big[a_{u_1}^* a_{v_2} G_{\mathrm{B}}(z) \otimes G_{\mathrm{free}} \big]\tr_{\mathscr{F}_{\mathrm{B}} \otimes \mathscr{F}_{\mathrm{I}}}\big[ a_{v_1}^* a_{u_2} G_{\mathrm{B}}(z) \otimes G_{\mathrm{free}} \big] \nonumber \\
&+\tr_{\mathscr{F}_{\mathrm{B}}}\big[a_{u_1}^* a_{v_1}^* G_{\mathrm{B}}(z) \big]\tr_{\mathscr{F}_{\mathrm{B}}}\big[ a_{u_2} a_{v_2} G_{\mathrm{B}}(z) \big]. \label{eq:QF}
\end{align}
The claimed identity in \eqref{lm:GBGF} follows when we apply Lemma~\ref{lem:1pdm} to compute the expectations in \eqref{eq:QF}.
\end{proof}

Our last preparatory lemma contains bounds for the 2, 3 and 4-pdms of $\Gamma_0$.

\begin{lemma}\label{lm:Tr}
The state $\Gamma_0$ in \eqref{eq:modfreestate} satisfies
\begin{align} 
\sum_{\substack{u_1,v_1\in P_{\mathrm{L}}\\u_2,v_2\in P_{\mathrm{L}}}}\big| \tr_{\mathscr{F}_{\mathrm{B}} \otimes \mathscr{F}_{\mathrm{I}}} \big[a_{u_1}^*a_{v_1}^* a_{u_2} a_{v_2}\Gamma_0\big]\big| &\lesssim N^2, \nonumber \\
\sum_{\substack{u_1,u_2,u_3\in P_{\mathrm{L}}\\v_1,v_2,v_3\in P_{\mathrm{L}}}} \big| \tr_{\mathscr{F}_{\mathrm{B}} \otimes \mathscr{F}_{\mathrm{I}}} [a^*_{v_1}a^*_{v_2}a^*_{v_3}  a_{u_1}a_{u_2}   a_{u_3} \Gamma_0] \big| &\lesssim N^3, \quad \text{ and } \nonumber \\
\sum_{\substack{u_1,u_2,u_3,u_4\in P_{\mathrm{L}} \\ v_1,v_2,v_3,v_4\in P_{\mathrm{L}}}} \big| \tr_{\mathscr{F}_{\mathrm{B}} \otimes \mathscr{F}_{\mathrm{I}}} [a^*_{v_1}a^*_{v_2}a^*_{v_3} a^*_{v_4}  a_{u_1}a_{u_2}   a_{u_3}a_{u_4}\Gamma_0]\big| &\lesssim N^4.
\end{align}
\end{lemma}
\begin{proof}
The first bound is a direct consequence of Lemmas~\ref{lem:boundsgamma}, \ref{lem:1-pdmGamma0} and \ref{lem:momentBoundsZeta}. To prove the second and the third bound we first need to compute the 3-pdm and the 4-pdm of $\Gamma_0$ as in the proof of Lemma~\ref{lem:1-pdmGamma0}. Afterwards, applications of the same lemmas prove the claim. Carrying out these steps is straightforward but a little lengthy. We therefore leave the details to the reader.
\end{proof}

\section{Bound for the energy}
\label{sec:energy}

We compute now the expectation of the Hamiltonian $\cH_N$, defined in \eqref{eq:Hamiltonian} and \eqref{eq:Hamiltonian(n)}, on our trial state. The main result of this section is Proposition \ref{trHG} below. This, together with Proposition \ref{prop:Entropy} for the entropy contribution (in Section \ref{entr}) will be the main ingredient to prove Theorem \ref{thm:Main}.

\begin{proposition}\label{trHG}
 Assume that $v : [0,\infty) \to [0,\infty]$ is nonnegative, compactly supported, and satisfies $v(| \cdot |) \in L^3(\Lambda, \de x)$. Let $\Gamma$ be defined in \eqref{trialstate} and $\beta=\kappa\beta_\mathrm{c}$, with $\kappa\in(0,\infty)$. Then we have
 \begin{equation} \label{eq:FH}
\tr\big[\cH_N\Gamma\big]-E_{\cH_N} \lesssim L^{-2}\cE_{\cH_N}\\
\end{equation}
where
\begin{align}
E_{\cH_N}=&\tr_{\Fock_{\mathrm{I}}} \Big[\Big(\sum_{p\in P_\mathrm{I}}p^2a^*_p a_p\Big)  G_{\mathrm{free}}\Big] \nonumber \\
&+ \tr_{\Fock_{\mathrm{B}}} \,\Big[\sum_{p\in P_\mathrm{B}} \Big( p^2a^*_p a_p+\frac{N_0(\beta,N,L)}{2L^3}(\hat v_N\ast\hat f_N)(p)\Big(2\,a^*_{p}a_{p}+(z^2/|z|^2)a^*_{p}a^*_{-p}+(\bar z^2/|z|^2)a_{p}a_{-p}\Big) \Big)\,G_{\mathrm{B}}(z)\Big] \nonumber \\
&+4\pi\frak{a}_NL^{-3}\biggl[\int_{\mathbb{C}} |z|^4 \zeta(z) \de z+2\widetilde N_0\sum_{\substack{ u\in P_{\mathrm{L}}\backslash\{0\}}}\gamma(u)+ 2\widetilde N_0\sum_{q\in P_\mathrm{I}}\gamma(q) +2\sum_{\substack{ u,v\in P_{\mathrm{L}}\backslash\{0\}}}\gamma(v)\gamma(u)\biggl] \label{eq:EH}
\end{align}
and
\begin{equation}\label{eq:FEr}
\cE_{\cH_N}=N^{1-\delta_{\mathrm{H}}}+N^{\delta_{\mathrm{H}}+2\delta_{\mathrm{B}}}+N^{-1/3+\delta_{\mathrm{H}}+2\delta_L}+N^{1/3+\delta_{\mathrm{B}}}.
\end{equation}
The parameter $\widetilde N_0$ has been introduced in \eqref{eq:WidetildeN0}, while $N_0(\beta,N,L)$ has been defined in \eqref{eq:DensityBEC}.
\end{proposition}

\begin{remark}
	The $z$-dependence of the Bogoliubov Hamiltonian and of $G_{\mathrm{B}}(z)$ in the second line of \eqref{eq:EH} cancel out exactly. This explains why this term is not integrated over $z$. 
\end{remark}

To prove Proposition \ref{trHG}, we split the Hamiltonian in two contributions: we define 
\begin{equation}\label{eq:excitHam} 
\cK=\sum_{p\in\L^*_{+}}p^2  a^*_p a_p\qquad\text{and}\qquad\cV_N=  \frac{1}{2L^3} \sum_{p,u,v \in \Lambda^*} \hat{v}_N(p) \,a_{u+p}^* a_{v-p}^* a_u a_v
\end{equation}
so that $\cH_N=\cK+\cV_N$. We have therefore
\begin{equation}
	\tr[\mathcal{H}_N \Gamma] = \sum_{\alpha} \lambda_{\alpha} \frac{\langle (1+B) \psi_{\alpha}, ( \cK + \cV_N ) (1+B) \psi_{\alpha} \rangle }{\langle (1+B) \psi_{\alpha}, (1+B) \psi_{\alpha} \rangle} \coloneqq \cG_\cK + \cG_\cV. 
	\label{eq:GKGV}
\end{equation}

We will prove Proposition \ref{trHG} in Section \ref{prop}, using the results of Lemma \ref{lm:GV} below for the analysis of $\cG_\cV$ and Lemma \ref{lm:GK}  for the analysis of $\cG_\cK$.

\subsection{Analysis of $\cG_\cV$}\label{GV}

In this section we prove an upper bound for $\cG_\cV$, as stated in the following lemma.

\begin{lemma}\label{lm:GV}
Under the assumptions of Proposition \ref{trHG}, we have
\begin{equation}\label{eq:gVV}
 \cG_V-E_{\cV_N} \lesssim L^{-2} \big(N^{1-\delta_{\mathrm{H}}}+N^{1/3}+N^{\delta_{\mathrm{H}}}\big),
\end{equation}
where
\begin{align}
  E_{\cV_N}=&\frac{4\pi\frak{a}_N}{L^3}\int_{\mathbb{C}} |z|^4 \zeta(z) \de z+\frac{8\pi\frak{a}_N\widetilde N_0}{L^3}\sum_{\substack{ u\in P_\mathrm{L}\backslash\{0\}}}\gamma(u) \nonumber \\
  &+\frac{\widetilde N_0}{2L^6}\sum_{\substack{p,q\in \L^*}} \hat{v}_N(p-q)\hat f_N(p)\Big[2\gamma(q)+ \alpha(q)+\overline{\alpha(q)}\Big]+\frac{8\pi\frak{a}_N}{L^3}\sum_{\substack{ u,v\in P_\mathrm{L}\backslash\{0\}}}\gamma(v)\gamma(u) \label{eq:mainGV} \\
  &+\frac{1}{2L^6}\sum_{\substack{p_1\in P_\mathrm{H}\\u_1,v_1,u_2,v_2\in P_\mathrm{L}}}\eta_{p_1}\Big[\hat v_N(p_1+u_1-u_2)+\frac{1}{L^3} \sum_{p_2\in P_\mathrm{H}}\hat v_N(p_1+p_2+u_1-u_2)\eta_{p_2}\Big]\tr[a^*_{v_1}a^*_{u_1}a_{v_2}a_{u_2}\Gamma_0]. \nonumber
\end{align}
The functions $\gamma(p)$ and $\alpha(p)$ are defined in \eqref{eq:defGam}.
\end{lemma}

\begin{proof}
Recall definition \eqref{eq:B} for  $B$. Acting on $\psi_\alpha$ (i.e., the eigenfunctions of $\Gamma_0$, defined in \eqref{eq:modfreestate}) with annihilation operators of momenta in $P_\mathrm{H}$ gives zero. Therefore  $\langle\psi_\alpha, B\,\psi_\alpha \rangle=\langle\psi_\alpha, B^*\,\psi_\alpha \rangle=0$, and we can estimate the denominator in \eqref{eq:cLn0} as $\|(1+B)\psi_{\alpha}\|^2=\langle\psi_\alpha, (1+B^*B)\,\psi_\alpha \rangle\geq 1$ so to have the upper bound
\begin{equation}
 \cG_V \leq\sum_{\alpha} \lambda_{\alpha}\braket{\psi_{\alpha},\mathcal{V}_N(1+B) \psi_{\alpha}}+\sum_{\alpha} \lambda_{\alpha}\braket{\psi_{\alpha},B^*\mathcal{V}_N(1+B) \psi_{\alpha}}=:\cG_V^{(1)}+\cG_V^{(2)}.
 \label{eq:cLn0} 
\end{equation}
With definitions \eqref{eq:excitHam} for $\mathcal{V}_N$ and \eqref{eq:B} for $B$ we write
\begin{align}
\cG_V^{(1)}&=\tr [ \mathcal{V}_N(1+B) \Gamma_0  ] =\frac{1}{2L^3}\sum_{p_1,u_1,v_1 \in \Lambda^*} \hat{v}_N(p_1) \,\tr[a_{u_1+p_1}^* a_{v_1-p_1}^* a_{u_1} a_{v_1}\Gamma_0] \nonumber \\
&\quad+\frac{1}{4L^6}\sum_{\substack{p_1,u_1,v_1 \in \Lambda^*\\p_2\in P_\mathrm{H},\, u_2,v_2\in P_\mathrm{L}}} \hat{v}_N(p_1)\eta_{p_2} \,\tr[a_{u_1+p_1}^* a_{v_1-p_1}^* a_{u_1} a_{v_1}a_{u_2+p_2}^* a_{v_2-p_2}^* a_{u_2} a_{v_2}\Gamma_0] \label{eq:G1} \nonumber \\
&=:\cG_V^{(1,1)}+\cG_V^{(1,2)}.
\end{align}
Using the commutation relations \eqref{eq:ccr}, we bring the monomial $a_{u_1} a_{v_1}a_{u_2+p_2}^* a_{v_2-p_2}^*$ in $\cG_V^{(1,2)}$ to normal order. When we exploit again that acting on $\Gamma_0$ with annihilation operators of momenta in $P_\mathrm{H}$ gives zero, we remain with
\begin{equation}
\cG_V^{(1,2)}=\frac{1}{2L^6}\sum_{\substack{p_1,u_1,v_1 \in \Lambda^*\\p_2\in P_\mathrm{H},\, u_2,v_2\in P_\mathrm{L}}} \hat{v}_N(p_1)\eta_{p_2} \,\tr[a_{u_1+p_1}^* a_{v_1-p_1}^*  a_{u_2} a_{v_2}\Gamma_0]\delta_{v_1,v_2-p_2}\delta_{u_1,u_2+p_2},
\end{equation}
where in addition  we used the symmetry under exchange of $u_1$ with $v_1$ and $p_1$ with $-p_1$. We add and subtract the contributions where $p_2\in P_\mathrm{H}^{\mathrm{c}}=\{p \in \Lambda^* \ | \ |p|< N^{1 - \delta_{\mathrm{H}}} /L \}$; using the definition of $\eta_p$ in \eqref{eq:eta} we find
\begin{align}
\cG_V^{(1,2)}
&=-\frac{1}{2L^3}\sum_{\substack{p_1,u_1,v_1 \in \Lambda^*}} \hat{v}_N(p_1) \,\tr[a_{u_1+p_1}^* a_{v_1-p_1}^*  a_{u_1} a_{v_1}\Gamma_0] \nonumber \\
&\quad+\frac{1}{2L^6}\sum_{\substack{p_1 \in \Lambda^*\\p_2\in \L^*,\, u_2,v_2\in P_\mathrm{L}}} \hat{v}_N(p_1)\hat f_N(p_2) \,\tr[a_{u_2+p_2+p_1}^* a_{v_2-p_2-p_1}^*  a_{u_2} a_{v_2}\Gamma_0] \nonumber \\
&\quad+\frac{1}{2L^6}\sum_{\substack{p_1,u_1,v_1 \in \Lambda^*\\p_2\in P_\mathrm{H}^c,\, u_2,v_2\in P_\mathrm{L}}} \hat{v}_N(p_1)\eta_{p_2} \,\tr[a_{u_2+p_2+p_1}^* a_{v_2-p_2-p_1}^*  a_{u_2} a_{v_2}\Gamma_0] \label{eq:G12}
\end{align}
The first contribution in \eqref{eq:G12} cancels with the first contribution in \eqref{eq:G1}. In the following, we denote the second and the third term on the r.h.s. of \eqref{eq:G12} by $\tilde\cG_V$ and $\cE^{(1)}_V$, respectively.

Using Lemma \ref{lm:Tr} to estimate the trace, equation \eqref{eq:sumPHCC} to estimate the sum of $\eta_{p_2}$ over $p_2$ and the bound $| \hat v_N(p_1) | \leq \int v_N(x) \de x \lesssim L N^{-1} $, we see that 
\begin{equation}\label{eq:erv1}
|\cE^{(1)}_V|  \leq\frac{1}{2 L^6}\biggl[\sup_{p_1\in \Lambda^*} | \hat{v}_N(p_1) | \biggl]\sum_{\substack{p_2\in P_\mathrm{H}^c}} |\eta_{p_2}| \sum_{\substack{p_1 \in \Lambda^*, u_2,v_2\in P_\mathrm{L}}} \big|\tr[a_{u_2+p_2+p_1}^* a_{v_2-p_2-p_1}^*  a_{u_2} a_{v_2}\Gamma_0]\big| \lesssim L^{-2}N^{1-\delta_{\mathrm{H}}}.
\end{equation}
We consider now $\tilde\cG_V$; we compute the trace using  Lemma \ref{lem:1-pdmGamma0} and obtain
\begin{align}
\tilde\cG_V &=\frac{1}{2L^6}\sum_{\substack{p_1,p_2\in \L^*}} \hat{v}_N(p_1)\hat f_N(p_2) \,\biggl[\delta_{p_1+p_2,0}\int_{\mathbb{C}}  |z|^4 \,\zeta(z) \de z\,+2\widetilde N_0\delta_{p_1+p_2,0}\sum_{\substack{ u\in P_\mathrm{L}\backslash\{0\}}}\gamma(u)+2\widetilde N_0\gamma(p_1+p_2) \nonumber \\
&\hspace{2cm}+\widetilde N_0\alpha(p_1+p_2)+\widetilde N_0\overline{\alpha(p_1+p_2)}+\sum_{\substack{ u\in P_\mathrm{B}\backslash\{0\}}}\overline{\alpha(u+p_2+p_1)}\alpha(u) \nonumber \\
&\hspace{2cm}+\delta_{p_1+p_2,0}\sum_{\substack{ u,v\in P_\mathrm{L}\backslash\{0\}}}\gamma(v)\gamma(u)+\sum_{\substack{ u\in P_\mathrm{L}\backslash\{0\}}}\gamma(u+p_2+p_1)\gamma(u)\biggl]=:\sum_{j=1}^8\tilde\cG_{V,j}. \label{eq:tildeGV}
\end{align}
Using  \eqref{eq:sl} we see that
 \begin{align}\label{eq:Vf}
\sum_{p\in\L^*}\hat v_N(p)\hat f_N(p)=\frac{L^3}{N}\int  v(x)f(x) \de x=\frac{L^3}{N}\left(8\pi\frak{a}+CL/N\right),
\end{align}
and therefore the first contribution in \eqref{eq:tildeGV} satisfies
\begin{equation}\label{eq:gV1}
\tilde\cG_{V,1}
=\frac{1}{2L^6}\sum_{\substack{p\in \L^*}} \hat{v}_N(p)\hat f_N(p) \int_{\mathbb{C}} |z|^4 \zeta(z) \de z \leq \frac{4\pi\frak{a}}{NL^3}\int_{\mathbb{C}} |z|^4 \zeta(z) \de z+ C L^{-2}.\\
\end{equation}
To obtain a bound for the integral over $|z|^4$ we applied Lemma~\ref{lem:momentBoundsZeta}. 

We consider now $\tilde\cG_{V,2}$. From $\sum_{u \in \Lambda^*_+} \gamma(u) \leq N$ and \eqref{eq:Vf} we know that
\begin{equation}\label{eq:gV2}
\tilde\cG_{V,2}
=\frac{\widetilde N_0}{L^6}\sum_{\substack{p\in \L^*}} \hat{v}_N(p)\hat f_N(p) \sum_{\substack{ u\in P_\mathrm{L}\backslash\{0\}}}\gamma(u) \leq \frac{8\pi\frak{a}\widetilde N_0}{NL^3}\sum_{\substack{ u\in P_\mathrm{L}\backslash\{0\}}}\gamma(u)+ CL^{-2} \\
\end{equation}
holds.
The sum of $\tilde\cG_{V,3}, \tilde\cG_{V,4}$ and $\tilde\cG_{V,5}$ is left untouched, i.e.,
\begin{equation}\label{eq:gV345}
\tilde\cG_{V,3}+\tilde\cG_{V,4}+\tilde\cG_{V,5}=\frac{\widetilde N_0}{2L^6}\sum_{\substack{p,q\in \L^*}} \hat{v}_N(p-q)\hat f_N(p)\Big[2\gamma(q)+\alpha(q)+\overline{\alpha(q)}\Big].\\
\end{equation}
Next, we consider $\tilde\cG_{V,6}$; from \eqref{eq:sl} it follows that
  \begin{equation}\label{eq:SSV}
\sup_{p\in\L^*}\Big|\sum_{q\in\L^*}\hat v_N(q)\hat f_N(p-q)\Big|\leq \frac{L^3}{N} \left(8\pi\frak{a}+\frac{CL}{N}\right);\\
 \end{equation}
using in addition \eqref{eq:Comp4b} and $\delta_{\mathrm{B}} < 1/3$ we see that
\begin{equation}\label{eq:gV6}
|\tilde\cG_{V,6}|
=\frac{1}{2L^6}\biggl|\sum_{\substack{p,q\in \L^*}} \hat{v}_N(q)\hat f_N(p-q) \sum_{\substack{ u\in P_\mathrm{B}\backslash\{0\}}}\bar\alpha(u+p)\alpha(u)\biggl|\leq CL^{-2}N^{1/3}.
\end{equation}
Using \eqref{eq:Vf} and \eqref{eq:SSV} we see that the last two terms in \eqref{eq:tildeGV} are equal at leading order:
\begin{align}
\tilde\cG_{V,7}+\tilde\cG_{V,8}
&=\frac{1}{2L^6}\sum_{\substack{p\in \L^*}} \hat{v}_N(p)\hat f_N(p) \sum_{\substack{ u,v\in P_\mathrm{L}\backslash\{0\}}}\gamma(v)\gamma(u) +\frac{1}{2L^6}\sum_{\substack{p,q\in \L^*}} \hat{v}_N(q)\hat f_N(p-q) \sum_{\substack{ u\in P_\mathrm{L}\backslash\{0\}}}\gamma(u+p)\gamma(u) \nonumber \\
&\leq \frac{8\pi\frak{a}}{NL^3} \sum_{\substack{ u,v\in P_\mathrm{L}\backslash\{0\}}}\gamma(v)\gamma(u) + C L^{-2}.
\label{eq:gV7}
\end{align}
It remains to consider $\cG_V^{(2)}$ in \eqref{eq:cLn0}.

To that end, we write
\begin{align}
\cG_V^{(2)}&= \tr [ B^*\mathcal{V}_N(1+B) \Gamma_0 ] = \tr [ B^*\mathcal{V}_N \Gamma_0  ] + \tr [ B^*\mathcal{V}_N B \Gamma_0  ] =:\cG_V^{(2,1)}+\cG_V^{(2,2)}.
\label{eq:Sect3Andi}
\end{align}
We have
 \begin{equation}
B^*\mathcal{V}_N =\frac{1}{4L^6}\sum_{\substack{p_1\in P_\mathrm{H}\\u_1,v_1\in P_L}}\eta_{p_1}\sum_{\substack{p_2,u_2,v_2\in \Lambda^*}}\hat v_N(p_2)\,a^*_{v_1}a^*_{u_1} a_{u_1+p_1} a_{v_1-p_1} a^*_{u_2+p_2} a^*_{v_2-p_2} a_{v_2}a_{u_2};
\end{equation}
commuting  $a_{u_1+p_1} a_{v_1-p_1} $ to the right and observing that only the contributions with $v_2,u_2\in P_\mathrm{L}$ give a non zero contribution, we arrive at
\begin{align}
\cG_V^{(2,1)}&=\frac{1}{2L^6}\sum_{\substack{p_1\in P_\mathrm{H}\\u_1,v_1\in P_\mathrm{L}}}\eta_{p_1}\sum_{\substack{p_2\in \Lambda^*\\u_2,v_2\in P_\mathrm{L}}}\hat v_N(p_2)\, \tr[ a^*_{v_1}a^*_{u_1}a_{v_2}a_{u_2} \Gamma_0 ] \delta_{u_2+p_2,u_1+p_1}\delta_{v_2-p_2,v_1-p_1} \nonumber \\
&=\frac{1}{2L^6}\sum_{\substack{p_1\in P_\mathrm{H}\\u_1,v_1,u_2,v_2\in P_\mathrm{L}}}\eta_{p_1}\hat v_N(p_1+u_1-u_2)\,\tr[a^*_{v_1}a^*_{u_1}a_{v_2}a_{u_2}\Gamma_0]\delta_{v_2,u_1+v_1-u_2}. \label{eq:BstarV}
\end{align}
This term contributes to \eqref{eq:mainGV}.
Note that $\delta_{v_2,u_1+v_1-u_2}$ can be dropped here because $\Gamma_0$ is translation invariant.

To compute $\cG_V^{(2,2)}$ we need to study
\begin{equation}
B^*\mathcal{V}_NB=\frac{1}{8L^9}\sum_{\substack{p_1,p_3\in P_\mathrm{H}\\u_1,v_1,u_3,v_3\in P_L\\p_2,u_2,v_2,\in \Lambda^*}}\eta_{p_1}\hat v_N(p_2)\eta_{p_3}\,a^*_{v_1}a^*_{u_1} a_{u_1+p_1} a_{v_1-p_1} a^*_{u_2+p_2} a^*_{v_2-p_2} a_{v_2}a_{u_2}a^*_{u_3+p_3} a^*_{v_3-p_3} a_{v_3}a_{u_3}.
\end{equation}
We bring the monomial $a_{v_2}a_{u_2}a^*_{u_3+p_3} a^*_{v_3-p_3}$ to normal order; the symmetries under exchange of $u_2$ with $v_2$, of $u_3$ with $v_3$ and $\eta_{-p_2} = \eta_{p_2}$ allow to organize the result of the commutation in the following three contributions:
\begin{equation}
\cG_V^{(2,2)}=\text{J}_1+\text{J}_2+\text{J}_3
\label{eq:Sect3Andi2}
\end{equation}
with
\begin{align}
\text{J}_1:=\frac{1}{8L^9}\sum_{\substack{p_1,p_3\in P_\mathrm{H}\\u_1,v_1,u_3,v_3\in P_\mathrm{L}\\p_2,u_2,v_2,\in \Lambda^*}}\eta_{p_1}\hat v_N(p_2)\eta_{p_3}\tr[a^*_{v_1}a^*_{u_1} a_{u_1+p_1} a_{v_1-p_1} a^*_{u_2+p_2} a^*_{v_2-p_2}   a^*_{u_3+p_3} a^*_{v_3-p_3} a_{v_2}a_{u_2}   a_{v_3}a_{u_3}\Gamma_0], \nonumber \\
\text{J}_2:=\frac{1}{2L^9}\sum_{\substack{p_1,p_3\in P_\mathrm{H}\\u_1,v_1,u_3,v_3\in P_\mathrm{L}\\p_2,u_2,v_2,\in \Lambda^*}}\eta_{p_1}\hat v_N(p_2)\eta_{p_3}\tr[a^*_{v_1}a^*_{u_1} a_{u_1+p_1} a_{v_1-p_1} a^*_{u_2+p_2} a^*_{v_2-p_2}   a^*_{u_3+p_3} a_{u_2}   a_{v_3}a_{u_3}\Gamma_0]\delta_{v_2,v_3-p_3}, \nonumber \\
\text{J}_3:=\frac{1}{4L^9}\sum_{\substack{p_1,p_3\in P_\mathrm{H}\\u_1,v_1,u_3,v_3\in P_\mathrm{L}\\p_2,u_2,v_2,\in \Lambda^*}}\eta_{p_1}\hat v_N(p_2)\eta_{p_3}\tr[a^*_{v_1}a^*_{u_1} a_{u_1+p_1} a_{v_1-p_1} a^*_{u_2+p_2} a^*_{v_2-p_2}  a_{v_3}a_{u_3}\Gamma_0]\delta_{u_2,v_3-p_3}\delta_{v_2,u_3+p_3}.
\end{align}
In $\text{J}_1$ we bring the monomial $a_{u_1+p_1} a_{v_1-p_1}a^*_{u_3+p_3} a^*_{v_3-p_3} $ to normal order; using that we obtain zero when we act with annihilations operators of momenta in $P_\mathrm{H}$ on $\Gamma_0$ and exploiting the symmetry under exchange of $v_1$ and $u_1$ we obtain
\begin{equation}
\text{J}_1=\frac{1}{4L^9}\sum_{\substack{p_1\in P_\mathrm{H}\\u_1,v_1,u_3\in P_\mathrm{L}\\p_2,u_2,v_2,\in \Lambda^*}}\eta_{p_1}\hat v_N(p_2)\eta_{p_1+u_1-u_3}\tr[a^*_{v_1}a^*_{u_1}a^*_{u_2+p_2} a^*_{v_2-p_2}  a_{v_2}a_{u_2}   a_{v_1+u_1-u_3}a_{u_3}\Gamma_0].
\end{equation}
When we apply Lemma \ref{lm:Tr} and the Cauchy-Schwarz inequality, and use the bound in \eqref{eq:sumPHC} as well as $|\hat v_N(p)| \lesssim L N^{-1}$, we see that
\begin{equation}\label{eq:gV8r}
|\text{J}_1| \lesssim \frac{L}{NL^9}\sum_{\substack{u_1,v_1,u_3\in P_\mathrm{L}\\p_2,u_2,v_2,\in \Lambda^*}}\big|\tr[a^*_{v_1}a^*_{u_1}a^*_{u_2+p_2} a^*_{v_2-p_2}  a_{v_2}a_{u_2}   a_{v_1+u_1-u_3}a_{u_3}\Gamma_0]\big|\sum_{\substack{p_1\in P_\mathrm{H}}}|\eta_{p_1}\eta_{p_1+u_1-u_3}| \lesssim L^{-2} N^{\delta_{\mathrm{H}}}
\end{equation}
holds. 

In $\text{J}_2$ we bring the monomial $a_{u_1+p_1} a_{v_1-p_1}a^*_{u_3+p_3} $ to normal order (so we obtain zero when $a^*_{u_3+p_3}$ acts on $\Gamma_0$, since $u_3+p_3\in P_\mathrm{H}$ ) and find
\begin{equation}
\text{J}_2=\frac{1}{L^9}\sum_{\substack{p_1,p_3\in P_\mathrm{H}\\u_1,v_1,u_3,v_3\in P_\mathrm{L}\\p_2,u_2,v_2,\in \Lambda^*}}\eta_{p_1}\hat v_N(p_2)\eta_{p_3}\tr[a^*_{v_1}a^*_{u_1}  a_{v_1-p_1} a^*_{u_2+p_2} a^*_{v_2-p_2}   a_{u_2}   a_{v_3}a_{u_3}\Gamma_0]\delta_{v_2,v_3-p_3}\delta_{u_1+p_1,u_3+p_3}\\
\end{equation}
Now we normal order $ a_{v_1-p_1} a^*_{u_2+p_2} a^*_{v_2-p_2} $ (with the aim of commuting $a_{v_1-p_1}$ to the right, since $v_1-p_1\in P_\mathrm{H}$) obtaining
\begin{align}
\text{J}_2&=\frac{2}{L^9}\sum_{\substack{p_1,p_3\in P_\mathrm{H}\\u_1,v_1,u_3,v_3\in P_\mathrm{L}\\p_2,u_2,v_2,\in \Lambda^*}}\eta_{p_1}\hat v_N(p_2)\eta_{p_3}\tr[a^*_{v_1}a^*_{u_1}  a^*_{u_2+p_2} a_{u_2}   a_{v_3}a_{u_3}\Gamma_0]\delta_{v_2,v_3-p_3}\delta_{u_1+p_1,u_3+p_3}\delta_{v_1-p_1,v_2-p_2} \nonumber \\
&=\frac{2}{L^9}\sum_{\substack{p_1\in P_\mathrm{H}\\u_1,v_1,u_3,v_3,u_2\in P_\mathrm{L}}}\eta_{p_1}\hat v_N(v_3+u_3-v_1-u_1)\eta_{p_1+u_1-u_3}\tr[a^*_{v_1}a^*_{u_1}  a^*_{u_2-v_1+v_3-u_1+u_3} a_{u_2}   a_{v_3}a_{u_3}\Gamma_0].
\end{align}
Again we exploited that $v_1$ can be exchanged with $u_1$ and  $v_2$ with $u_2$.
Using Lemma \ref{lm:Tr}, \eqref{eq:sumPHC} and $| \hat v_N(p) | \lesssim L N^{-1}$, we obtain the bound
\begin{equation}\label{eq:gV9r}
|\text{J}_2| \lesssim L^{-2} N^{-1+\delta_{\mathrm{H}}}.
\end{equation}
Normal ordering of $ a_{u_1+p_1} a_{v_1-p_1} a^*_{u_2+p_2} a^*_{v_2-p_2}$ and analogous considerations as above lead to
\begin{align}
\text{J}_3&=\frac{1}{2L^9}\sum_{\substack{p_1,p_3\in P_\mathrm{H}\\u_1,v_1,u_3,v_3\in P_\mathrm{L}\\p_2,u_2,v_2,\in \Lambda^*}}\eta_{p_1}\hat v_N(p_2)\eta_{p_3}\tr[a^*_{v_1}a^*_{u_1}   a_{v_3}a_{u_3}\Gamma_0]\delta_{u_2,v_3-p_3}\delta_{v_2,u_3+p_3}\delta_{v_1-p_1,v_2-p_2}\delta_{u_2+p_2,u_1+p_1} \nonumber \\
&=\frac{1}{2L^9}\sum_{\substack{p_1,p_3\in P_\mathrm{H}\\u_1,v_1,v_3,u_3\in P_\mathrm{L}}} \eta_{p_1} \hat{v}_N(p_1+p_3+u_3-v_1) \eta_{p_3} \tr[ a^*_{v_1} a^*_{u_1} a_{v_3} a_{u_3} \Gamma_0].
\end{align}
We combine $\text{J}_3$ with $\cG_V^{(2,1)}$ in \eqref{eq:BstarV} and obtain
\begin{align}
\cG_V^{(2,1)}+&\text{J}_3=\frac{1}{2L^6}\sum_{\substack{p_1\in P_\mathrm{H}\\u_1,v_1,v_2,u_2\in P_\mathrm{L}}}\eta_{p_1}\hat v_N(p_1+u_1-u_2)\,\tr[a^*_{v_1}a^*_{u_1}a_{v_2}a_{u_2}\Gamma_0] \nonumber \\
&\hspace{0.6cm}+\frac{1}{2L^9}\sum_{\substack{p_1,p_2\in P_\mathrm{H}\\u_1,v_1,v_2,u_2\in P_\mathrm{L}}}\eta_{p_1}\hat v_N(p_1+p_2+u_2-v_1)\eta_{p_2}\tr[a^*_{v_1}a^*_{u_1}   a_{v_2}a_{u_2}\Gamma_0] \label{eq:gV10} \\
=& \frac{1}{2L^6}\sum_{\substack{p_1\in P_\mathrm{H}\\u_1,v_1,u_2,v_2\in P_\mathrm{L}}}\eta_{p_1}\Big[\hat v_N(p_1+u_1-u_2)+ \frac{1}{| \Lambda| }\sum_{p_2\in P_\mathrm{H}}\hat v_N(p_1+p_2+u_1-u_2)\eta_{p_2}\Big]\tr[a^*_{v_1}a^*_{u_1}a_{v_2}a_{u_2}\Gamma_0]. \nonumber
\end{align}
In the last line we used the symmetry under exchange of $v_1$ with $u_1$.
Collecting the results of \eqref{eq:cLn0}, \eqref{eq:gV1}, \eqref{eq:gV2}, \eqref{eq:gV345}, \eqref{eq:gV7}, \eqref{eq:Sect3Andi}, \eqref{eq:Sect3Andi2}, \eqref{eq:gV10} and the bounds on the error terms in   \eqref{eq:erv1},  \eqref{eq:gV6}, \eqref{eq:gV8r}, \eqref{eq:gV9r} we obtain \eqref{eq:gVV}.
\end{proof}

\subsection{Analysis of $\cG_\cK$}\label{GK}

We recall from definitions of $\cK$ in \eqref{eq:excitHam} and $\cG_\cK$ in \eqref{eq:GKGV}. In this section we prove an upper bound for $\cG_\cK$, as stated in the following lemma.

\begin{lemma}\label{lm:GK}
Under the assumptions of Proposition \ref{trHG} we have
\begin{equation}
 \cG_\cK-E_{\cK} \lesssim L^{-2}\big(N^{\delta_{\mathrm{H}}+2\delta_{\mathrm{B}}}+N^{-1/3+\delta_{\mathrm{H}}+2\delta_L}+N^{1/3}\ln(N)\big)
\end{equation}
with
\begin{equation}\label{eq:mainGK}
 E_{\cK}= \sum_{p\in P_\mathrm{L}}p^2\gamma(p)+\frac{1}{L^6}\sum_{\substack{p_1\in P_\mathrm{H}\\u_1,v_1,u_2,v_2\in P_\mathrm{L}}}\eta_{p_1}(p_1+u_1-u_2)^2\eta_{p_1+u_1-u_2}\,\tr[a^*_{v_1}a^*_{u_1}a_{v_2}a_{u_2}\Gamma_0].
\end{equation}
The function $\gamma(p)$ is defined in \eqref{eq:defGam}.
\end{lemma}

\begin{proof}
It is convenient to introduce the operators
\begin{equation}\label{eq:KSKI} 
\cK_{\mathrm{B}}=\sum_{p\in P_\mathrm{B}}p^2  a^*_p a_p,\qquad \cK_{\mathrm{I}}=\sum_{p\in P_\mathrm{I}}p^2  a^*_p a_p\qquad\text{and}\qquad \cK_>=\sum_{p\in P_\mathrm{L}^c}p^2  a^*_p a_p
\end{equation}
and to denote the corresponding expectations w.r.t. $\Gamma$ by $\cG_{\cK_{\mathrm{B}}}$, $\cG_{\cK_{\mathrm{I}}}$, $\cG_{\cK_>}$, i.e., 
\begin{equation}
\cG_\cK = \sum_{\alpha} \lambda_{\alpha}\frac{\braket{(1+B)\psi_{\alpha}, \big(\cK_{\mathrm{B}}+\cK_{\mathrm{I}}+\cK_>\big)(1+B)\psi_{\alpha}}} {\braket{(1+B)\psi_{\alpha}, (1+B)\psi_{\alpha}}}=\cG_{\cK_{\mathrm{B}}}+\cG_{\cK_{\mathrm{I}}}+\cG_{\cK_>}.
\end{equation}
In the next subsections we will prove upper bounds for $\cG_{\cK_{\mathrm{B}}}, \cG_{\cK_{\mathrm{I}}}$ and $\cG_{\cK_>}$. We will need to compute the commutators of $\cK_{\mathrm{I}}$ and $\cK_{>}$ with $B$. Indicating with $P_\#$ either $P_\mathrm{I}$ or $P_\mathrm{L}^{\mathrm{c}}$ and with $\cK_\#$ either $\cK_{\mathrm{I}}$ or $\cK_>$, we will use the result
\begin{align}
[\cK_\#,B] &=\frac{1}{2L^3}\sum_{\substack{q\in P_\#,p\in P_\mathrm{H},\\\, u,v\in P_{\mathrm{L}}}}q^2\eta_{p}[a^*_qa_q, a^*_{u+p}a^*_{v-p}a_ua_v] \nonumber \\
&=\frac{1}{L^3}\sum_{\substack{q\in P_\#,p\in P_\mathrm{H},\\\, u,v\in P_{\mathrm{L}}}}q^2\eta_{p}\, \Big(\delta_{q,u+p}a^*_{q}a^*_{v-p}a_ua_v-\delta_{q,u}\,a^*_{u+p}a^*_{v-p}a_va_q\Big), \label{eq:Kcomm}
\end{align}
where we exploited the symmetry $\eta_p=\eta_{-p}$.

\subsubsection{Analysis of $\cG_{\cK_{\mathrm{B}}}$} 

Using the positivity of $\cK_{\mathrm{B}}$ we estimate the denominator by one;  observing that $\braket{\psi_{\alpha},(B^*\cK_{\mathrm{B}}+\cK_{\mathrm{B}} B)\,\psi_{\alpha}}$ vanishes (because the creation operators in $B$ commute with $\cK_{\mathrm{B}}$ and give zero when acting on $\psi_\alpha$), we have
\begin{align}
\cG_{\cK_{\mathrm{B}}}&=\sum_{\alpha} \lambda_{\alpha}\frac{\braket{(1+B)\psi_{\alpha}, \cK_\mathrm{B}(1+B)\psi_{\alpha}}} {\braket{(1+B)\psi_{\alpha}, (1+B)\psi_{\alpha}}}\leq\sum_{\alpha} \lambda_{\alpha}\braket{\psi_{\alpha},\Big(\sum_{p\in P_\mathrm{B}}p^2a^*_p a_p\Big)\psi_{\alpha}}+\sum_{\alpha} \lambda_{\alpha}\braket{\psi_{\alpha},B^*\cK_{\mathrm{B}}B\,\psi_{\alpha}} \nonumber \\
&=\sum_{p\in P_\mathrm{B}}p^2\gamma(p)+\cE_{\cK_{\mathrm{B}}}. \label{eq:estSmallMom}
\end{align}
To estimate the error term in \eqref{eq:estSmallMom} we contract the annihilation operators with momenta in $P_\mathrm{H}$ in $B^*$ with those in $B$, and obtain
\begin{align}
|\cE_{\cK_{\mathrm{B}}}|&=\frac{1}{2L^6}\biggl|\sum_{\substack{p_1\in P_\mathrm{H},\,q\in P_\mathrm{B}\\u_1,v_1,u_2\in P_\mathrm{L}}}q^2\,\eta_{p_1}\,\eta_{p_1+u_1-u_2}\,\tr[a^*_{u_1}a^*_{v_1}a^*_qa_qa_{u_2}a_{u_1+v_1-u_2}\Gamma_0]\biggl| \nonumber \\
&\lesssim \frac{N^{-3+\delta_\mathrm{H}+2\delta_\mathrm{B}}}{L^2}\sum_{\substack{q\in P_\mathrm{B}\\u_1,v_1,u_2\in P_\mathrm{L}}}\big|\tr[a^*_{u_1}a^*_{v_1}a^*_qa_qa_{u_2}a_{u_1+v_1-u_2}\Gamma_0]\big| \lesssim L^{-2}N^{\delta_{\mathrm{H}}+2\delta_\mathrm{B}}.
\label{eq:erK1}
\end{align}
The inequalities follow from $|q|\leq N^{\delta_\mathrm{B}}$, the Cauchy-Schwarz inequality, the bound in \eqref{eq:sumPHC} and Lemma \ref{lm:Tr}. We therefore have
\begin{equation}\label{eq:KB}
\cG_{\cK_{\mathrm{B}}}- \sum_{p\in P_\mathrm{B}}p^2\gamma(p) \lesssim L^{-2} N^{\delta_{\mathrm{H}}+2\delta_\mathrm{B}}.
\end{equation}

\subsubsection{Analysis of $\cG_{\cK_{\mathrm{I}}}$}

Here we need to exploit crucial cancellations between the numerator and the denominator. We commute $\cK_{\mathrm{I}}$ to the right and obtain
 \begin{equation}\nonumber
 \cG_{\cK_{\mathrm{I}}}= \sum_{\alpha} \lambda_{\alpha}\frac{\braket{(1+B)\psi_{\alpha}, \cK_{\mathrm{I}}(1+B)\psi_{\alpha}}} {\braket{(1+B)\psi_{\alpha}, (1+B)\psi_{\alpha}}}=\sum_{\alpha} \lambda_{\alpha}\frac{\braket{(1+B)\psi_{\alpha},(1+B)\cK_{\mathrm{I}}\psi_{\alpha}}}{\braket{(1+B)\psi_{\alpha}, (1+B)\psi_{\alpha}}}+ \cE_{\cG_{K_{\mathrm{I}}}}
 \end{equation}
with
\begin{equation}
\cE_{\cG_{K_{\mathrm{I}}}}:=\sum_{\alpha} \lambda_{\alpha}\frac{\braket{\psi_{\alpha},(1+B^*)[\cK_{\mathrm{I}},B]\psi_{\alpha}}}{\braket{(1+B)\psi_{\alpha}, (1+B)\psi_{\alpha}}}.
\end{equation}
Let us introduce the notation $\psi_\alpha=\xi_{\alpha_1}\otimes\nu_{\alpha_2}$ with $\alpha=(\alpha_1,\alpha_2)$, where $\xi_{\alpha_1}$ and $\nu_{\alpha_2}$ denote the eigenfunctions of $G_{\mathrm{B}}$ and $G_{\mathrm{free}}$, respectively. Calling $E_{\alpha_2}$ the eigenvalues of $\cK_\mathrm{I}$, so that $\cK_\mathrm{I}\nu_{\alpha_2}=E_{\alpha_2}\nu_{\alpha_2}$,
we see that
\begin{equation}\label{eq:Ktermsb}
\cG_{\cK_{\mathrm{I}}} =\sum_{\alpha}  \lambda_{\alpha} E_{\alpha_2}+\cE_{\cG_{K_{\mathrm{I}}}}
\end{equation}
holds.

Next, we estimate $\cE_{\cG_{K_{\mathrm{I}}}}$. With $\Vert (1+B) \psi_{\alpha} \Vert \geq 1$ we have
\begin{equation}
| \cE_{\cG_{K_{\mathrm{I}}}} | = \left| \sum_{\alpha} \lambda_{\alpha}\frac{\braket{\psi_{\alpha},B^*[\cK_{\mathrm{I}},B]\psi_{\alpha}}}{\braket{(1+B)\psi_{\alpha}, (1+B)\psi_{\alpha}}} \right| \leq \sum_{\alpha} \lambda_{\alpha} |  \braket{\psi_{\alpha},B^*[\cK_{\mathrm{I}},B]\psi_{\alpha}} |.
\label{eq:Sect3Andi3}
\end{equation}
When we contract the high momenta in $B^*[\cK_{\mathrm{I}},B]$, we see that the inner product inside the absolute value equals
\begin{equation}
\braket{\psi_{\alpha},B^*[\cK_{\mathrm{I}},B]\psi_{\alpha}}=-\frac{1}{L^6}\sum_{\substack{u_1,v_1\in P_\mathrm{L}\\p_1\in P_\mathrm{H},\, u_2\in P_\mathrm{I}}}u_2^2\,\eta_{p_1+v_1-u_2}\eta_{p_1}\braket{\psi_{\alpha},\,a^*_{u_1}a^*_{v_1}a_{u_2}a_{u_1+v_1-u_2}  \psi_{\alpha}}.
\end{equation}
Applying the Cauchy-Schwarz inequality,  using \eqref{eq:sumPHC}, and estimating $|u_2|\leq N^{1/3+\delta_\mathrm{L}}/L$ we find
\begin{align}
\sum_{\alpha} \lambda_{\alpha} |\braket{\psi_{\alpha},B^*[\cK_{\mathrm{I}},B]\psi_{\alpha}}|&\leq\frac{1}{L^6}\sum_{\alpha} \lambda_{\alpha} \sum_{\substack{u_1,v_1\in P_\mathrm{L}\\ u_2\in P_\mathrm{I}}} |\braket{\psi_{\alpha},a^*_{u_1}a^*_{v_1}a_{u_2}a_{u_1+v_1-u_2}\psi_{\alpha}}|\sum_{\substack{p_1\in P_\mathrm{H}}}|\eta_{p_1+u_1-u_2}\eta_{p_1}\,u_2^2|\, \nonumber \\
&\lesssim L^{-2} N^{-7/3+\delta_{\mathrm{H}}+2\delta_{\mathrm{L}}} \sum_{\alpha} \lambda_{\alpha} \sum_{\substack{u_1,v_1\in P_\mathrm{L}\\u_2\in P_\mathrm{I}}} |\braket{\psi_{\alpha},a^*_{u_1}a^*_{v_1}a_{u_2}a_{u_1-v_1+u_2}\psi_{\alpha}}|.
\label{eq:Section3Andi4}
\end{align}
We observe now that since $u_2\in P_\mathrm{I}$, at least one of the momenta $u_1$ or $v_1$ needs to be in $P_\mathrm{I}$ (this is due to the fact that the eigenfunctions of $G_{\mathrm{free}}$ are symmetric tensor products of plane waves). Let us assume without loss of generality that $v_1\in P_\mathrm{I}$. We distinguish now the three cases $u_1=0,\, u_1\in P_\mathrm{B}$ and $u_1\in P_\mathrm{I}$.

If $u_1=0$, then $v_1-u_2=0$ (this again follows from the structure of the orthonormal set $\{ \psi_{\alpha} \}_{\alpha \in \mathbb{N}}$), and the expectation $\braket{\psi_{\alpha},a^*_{0}a^*_{u_2}a_{u_2}a_{0}\psi_{\alpha}}$ is positive. In this case we  estimate the sum on the r.h.s. \eqref{eq:Section3Andi4} by
\begin{align}
&\sum_{\alpha} \lambda_{\alpha}\sum_{\substack{u_2\in P_\mathrm{I}}}\braket{\psi_{\alpha},a^*_{0}a^*_{u_2}a_{u_2}a_{0}\psi_{\alpha}}
=\sum_{\substack{u_2\in P_\mathrm{I}}}\tr\big[a^*_{0}a^*_{u_2}a_{u_2}a_{0}\,\Gamma_0\big] \nonumber \\ &\hspace{1cm}=\int_{\mathbb{C}} \tr_{\Fock_0} \big[a^*_{0}a_{0}\,|z\rangle\langle z|\big]\zeta(z)\de z \ \tr_{\Fock_\mathrm{I}}\Big[\sum_{\substack{u_2\in P_\mathrm{I}}}a^*_{u_2}a_{u_2}\, \Gamma_{\text{free}}\Big] = \int_{\mathbb{C}} |z|^2 \zeta(z)\de z \ \sum_{\substack{u_2\in P_\mathrm{I}}}\gamma(u_2)\leq N^2.
\end{align}
To obtain the last inequality we used $\widetilde{N}_0 + \sum_{p \in \Lambda_+^*} \gamma(p) = N$.

In the case $u_1\in P_\mathrm{B}$ we have $v_1=u_2$ and $u_1-v_1+u_2 = u_1$ because $G_{\mathrm{B}}$ and $G_{\mathrm{free}}$ are both translation-invariant states. In particular, the relevant expectation value is again positive. Using \eqref{eq:Comp4c} we obtain the following bound for the sum on the r.h.s. of \eqref{eq:Section3Andi4}:
\begin{align}
\sum_{\alpha} &\lambda_{\alpha}\sum_{\substack{u_1\in P_\mathrm{B}\\v_1,u_2\in P_\mathrm{I}}}|\braket{\psi_{\alpha},a^*_{u_1}a^*_{v_1}a_{u_2}a_{u_1-v_1+u_2}\psi_{\alpha}}| = \sum_{\alpha} \lambda_{\alpha}\sum_{\substack{u_1\in P_\mathrm{B}\\u_2\in P_\mathrm{I}}}\braket{\psi_{\alpha},a^*_{u_2}a_{u_2}a^*_{u_1}a_{u_1}\psi_{\alpha}} \nonumber \\
&\hspace{1.5cm}=\int_{\mathbb{C}}\tr_{\Fock_{\mathrm{B}}} \,\Big[\sum_{\substack{u_1\in P_\mathrm{B}}}a^*_{u_1}a_{u_1}\, G_{\mathrm{B}}(z)\Big] \zeta(z)\de z \ \tr_{\Fock_{\mathrm{I}}} \,\Big[\sum_{\substack{u_2\in P_\mathrm{I}}}a^*_{u_2}a_{u_2}\, G_{\mathrm{free}}\Big] \lesssim N^{5/3+\delta_\mathrm{B}}
\end{align}

To conclude the discussion, we examine the case $u_1\in P_\mathrm{I}$. This implies that $u_1-v_1+u_2\in P_\mathrm{I}$. Moreover, we have that either $v_1=u_2$ or $u_1=u_2$; in both cases the expectation is positive. An application of Lemma \ref{lm:Tr} shows 
\begin{equation}
\sum_{\alpha}\lambda_{\alpha}\sum_{\substack{u_1, u_2\in P_\mathrm{I}}}\braket{\psi_{\alpha},a^*_{u_1}a^*_{u_2}a_{u_2}a_{u_1}\psi_{\alpha}}=\sum_{\substack{u_1, u_2\in P_\mathrm{I}}}\,\tr_{\Fock_{\mathrm{I}}}\Big[a^*_{u_1}a^*_{u_2}a_{u_2}a_{u_1}\,\Gamma_0\Big]\lesssim  N^2
\label{eq:Section3Andi5}
\end{equation}
as a bound for the sum on the r.h.s. of \eqref{eq:Section3Andi4}.

In combination, \eqref{eq:Sect3Andi3}, \eqref{eq:Section3Andi4}--\eqref{eq:Section3Andi5}, and $\delta_{\mathrm{B}} < 1/3$ imply  
\begin{equation}
|\cE_{\cG_{K_{\mathrm{I}}}}| \lesssim L^{-2}N^{-1/3+\delta_{\mathrm{H}}+2\delta_L}
\end{equation}
as well as
\begin{equation}\label{eq:KIresult}
\big|\cG_{\cK_{\mathrm{I}}}
-\sum_{\alpha}  \lambda_{\alpha} E_{\alpha_2}\big|=\big|\cG_{\cK_{\mathrm{I}}}
- \sum_{p\in P_\mathrm{I}}p^2\gamma(p)\big| \lesssim L^{-2}N^{-1/3+\delta_{\mathrm{H}}+2\delta_L}.
\end{equation}
To obtain \eqref{eq:KIresult}, we additionally used \eqref{eq:Ktermsb} and Lemma~\ref{lem:1pdm}.
\subsubsection{Analysis of $\cG_{\cK_>}$} 
In the analysis of $\cG_{\cK_>}$ we estimate the denominator again by one. When we additionally commute $\cK_>$ to the right we find
 \begin{align}
  \cG_{\cK_>}&\leq \sum_{\alpha} \lambda_{\alpha}\braket{(1+B)\psi_{\alpha}, \cK_>(1+B)\psi_{\alpha}} \nonumber \\
  &=\sum_{\alpha} \lambda_{\alpha}\braket{(1+B)\psi_{\alpha},(1+B)\cK_>\psi_{\alpha}}+ \sum_{\alpha} \lambda_{\alpha}\braket{\psi_{\alpha},(1+B^*)[\cK_>,B]\psi_{\alpha}}.
 \end{align}
 The first contribution clearly vanishes because it contains annihilation operators with momenta in  $P_\mathrm{L}^{\mathrm{c}}$ acting on $\psi_{\alpha}$. 
Using additionally \eqref{eq:Kcomm}, we see that
\begin{equation}\label{eq:comm}
[\cK_>,B]=\frac{1}{L^3}\sum_{\substack{p\in P_\mathrm{H}\\u,v\in P_\mathrm{L}}}p^2\eta_{p}\, a^*_{u+p}a^*_{v-p}a_ua_v+\frac{2}{L^3}\sum_{\substack{p\in P_\mathrm{H}\\u,v\in P_\mathrm{L}}}\eta_p\,u\cdot p\, a^*_{u+p}a^*_{v-p}a_ua_v=:\cK_\eta+\cE_\cK,
\end{equation}
which implies
\begin{equation}\label{eq:Kterms}
\cG_{\cK_>}
\leq \sum_{\alpha} \lambda_{\alpha}\braket{\psi_{\alpha},(1+B^*)\big(\cK_\eta+\cE_\cK\big)\psi_{\alpha}}
= \tr[ B^*\cK_\eta \Gamma_0 ] + \tr [ B^*\cE_\cK \Gamma_0].
\end{equation}
When we contract operators with high momenta this allows us to write the first term on the r.h.s. as
\begin{align}
\tr[ B^*\cK_\eta \Gamma_0 ]&=\frac{1}{L^6}\sum_{\substack{p_1\in P_\mathrm{H}\\u_1,v_1\in P_\mathrm{L}}}\eta_{p_1}\sum_{\substack{p_2\in P_\mathrm{H}\\u_2,v_2\in P_\mathrm{L}}}p^2_2\eta_{p_2}\, \tr[ a^*_{v_1}a^*_{u_1}a_{v_2}a_{u_2} \Gamma_0] \delta_{u_2+p_2,u_1+p_1}\delta_{v_2-p_2,v_1-p_1} \nonumber \\
&=\frac{1}{L^6}\sum_{\substack{p_1\in P_\mathrm{H}\\u_1,v_1,u_2,v_2\in P_\mathrm{L}}}\eta_{p_1}(p_1+u_1-u_2)^2\eta_{p_1+u_1-u_2}\,\tr[a^*_{v_1}a^*_{u_1}a_{v_2}a_{u_2}\Gamma_0]\delta_{v_2,v_1+u_1-u_2}. \label{eq:BstarKeta}
\end{align}
Because $\Gamma_0$ is translation invariant the factor $\delta_{v_2,v_1+u_1-u_2}$ can be dropped.
This term contributes to \eqref{eq:mainGK}. 

We consider now the second term in \eqref{eq:Kterms}, that is,
\begin{equation}\label{eq:BsEK}
\tr [ B^*\cE_\cK\ \Gamma_0] =\frac{2}{L^6}\sum_{u_1,v_1\in P_\mathrm{L}}\sum_{\substack{p_2\in P_\mathrm{H}\\u_2,v_2\in P_\mathrm{L}}}\eta_{p_2+u_2-u_1}\eta_{p_2}p_2\cdot u_2\,\tr[a^*_{v_1}a^*_{u_1}a_{v_2}a_{u_2}\Gamma_0].
\end{equation}
An application of Lemma~\ref{lem:1-pdmGamma0} shows that the trace in \eqref{eq:BsEK} can be written as $\tr [ B^*\cE_\cK\ \Gamma_0] = \text{D}_1+\text{D}_2$
with
\begin{align}
\text{D}_1&=2\widetilde N_0L^{-6}\sum_{\substack{p_2\in P_\mathrm{H},u_2\in P_\mathrm{L}}}\Big[\big(\eta_{p_2}+\eta_{p_2+u_2}\big)\eta_{p_2}p_2\cdot u_2\,\gamma(u_2)+\eta_{p_2+u_2}\eta_{p_2}p_2\cdot u_2\,\alpha(u_2)\Big], \nonumber \\
\text{D}_2&=2L^{-6}\sum_{\substack{p_2\in P_\mathrm{H}\\v_2,u_2\in P_\mathrm{L}}}\Big[\big(\eta_{p_2}+\eta_{p_2-v_2+u_2}\big)\eta_{p_2}p_2\cdot u_2\,\gamma(u_2)\gamma(v_2)+\eta_{p_2-v_2+u_2}\eta_{p_2}p_2\cdot u_2\,\overline{\alpha(v_2)}\alpha(u_2)\Big].
\end{align}
Let us introduce the set $P_r:=\{p\in\L^*:|p|\geq rN\}$, where $r>1$;  for $p_2\in P_r$ we have 
\begin{equation}\label{eq:etap1}
\sup_{u\in P_\mathrm{L}}\sum_{\substack{p_2\in P_r}}|\eta_{p_2-u}\eta_{p_2}||p_2|\leq \biggl(\sum_{\substack{p\in P_{r/2}}}|\eta_{p}|^2\biggl)^{1/2}\biggl(\sum_{\substack{q\in P_r}}|\eta_{q}|^2q^2\biggl)^{1/2} \lesssim_r L^5 N^{-2},
\end{equation}
which follows from the Cauchy-Schwarz inequality, \eqref{eq:etabound} and \eqref{eq:Dw} (the latter  implies $\|\nabla f_N\|_{L^2} \lesssim L^2 N^{-1/2}$).
Instead, for  $p_2\in P_r^{\mathrm{c}}$,  \eqref{eq:etabound} implies
\begin{equation}\label{eq:etap2}
\sup_{u\in P_\mathrm{L}}\sum_{\substack{p_2\in P_\mathrm{H}\cap P_r^c}}|\eta_{p_2}\eta_{p_2-u}||p_2| \lesssim \frac{L^2}{N^2} \sup_{u\in P_\mathrm{L}}  \sum_{\substack{p_2\in P_\mathrm{H}\cap P_r^c}} \frac{1}{|p_2|(|p_2|-|u|)^2} \lesssim \frac{L^5 \ln(N)}{N^2}.
\end{equation}
To obtain the second bound we used $|p_2|-|u| \gtrsim 1$, which follows from the assumption $\delta_{\mathrm{L}} + \delta_{\mathrm{H}} < 2/3$. Using \eqref{eq:Comp4b}, \eqref{eq:etap1}, \eqref{eq:etap2} and $\delta_{\mathrm{B}} < 1/3$, we conclude that
\begin{align}
|\text{D}_1|& \lesssim (\widetilde N_0/L^6)\sup_{u\in P_\mathrm{L}}\sum_{\substack{p_2\in P_\mathrm{H}}}|\eta_{p_2-u}\eta_{p_2}||p_2|\sum_{\substack{u_2\in P_\mathrm{L}}}|u_2|\gamma(u_2)+(\widetilde N_0/L^6) \sup_{u\in P_\mathrm{L}}\sum_{\substack{p_2\in P_\mathrm{H}}}|\eta_{p_2-u}\eta_{p_2}||p_2|\sum_{\substack{u_2\in P_\mathrm{L}}}|u_2||\alpha(u_2)| \nonumber \\
&\lesssim L^{-2} N^{1/3} \ln(N). \label{eq:erD1}
\end{align}
With similar considerations we see that
\begin{align}
|\text{D}_2|
\lesssim& \frac{1}{L^6} \sup_{u\in P_\mathrm{L}}\sum_{\substack{p_2\in P_\mathrm{H}}}|\eta_{p_2-u}\eta_{p_2}||p_2|\sum_{\substack{u_2, v_2\in P_\mathrm{L}}}|u_2|\gamma(u_2)\gamma(v_2) \nonumber \\
&+ \frac{1}{L^6} \sup_{u\in P_\mathrm{L}}\sum_{\substack{p_2\in P_\mathrm{H}}}|\eta_{p_2-u}\eta_{p_2}||p_2|\sum_{\substack{u_2,v_2\in P_\mathrm{L}}}|u_2||\alpha(u_2)\alpha(v_2)| \lesssim L^{-2} N^{1/3} \ln(N). \label{eq:erD3}
\end{align}
Collecting the results of equations \eqref{eq:Kterms}--\eqref{eq:BsEK} and the bounds  \eqref{eq:erD1}, \eqref{eq:erD3}, we conclude that
 \begin{equation}\label{eq:gK+}
  \begin{split}
  \cG_{\cK_>}-\frac{1}{L^6}\sum_{\substack{p_1\in P_\mathrm{H}\\u_1,v_1,u_2,v_2\in P_\mathrm{L}}}\eta_{p_1}(p_1+u_1-u_2)^2\eta_{p_1+u_1-u_2}\,\tr[a^*_{v_1}a^*_{u_1}a_{v_2}a_{u_2}\Gamma_0] \lesssim L^{-2}N^{1/3}\ln(N).
  \end{split}
 \end{equation}
The bounds \eqref{eq:KB}, \eqref{eq:KIresult} and \eqref{eq:gK+}  imply \eqref{eq:mainGK} and conclude the proof of Lemma \ref{lm:GK}. 
\end{proof}

\subsection{Proof of Proposition \ref{trHG}}\label{prop}

The results of Lemma \ref{lm:GV} and Lemma \ref{lm:GK} imply
\begin{equation}\label{eq:enf}
\tr\big[\cH_N \Gamma\big] -\big(E_{\cK}+E_{\cV_N}\big) \lesssim L^{-2}\cE_N,
\end{equation}
with $E_{\cV_N}$ in \eqref{eq:mainGV}, $E_{\cK}$ in \eqref{eq:mainGK} and
\begin{equation}
\cE_N=N^{1-\delta_{\mathrm{H}}}+N^{\delta_{\mathrm{H}}+2\delta_{\mathrm{B}}}+N^{-1/3+\delta_{\mathrm{H}}+2\delta_L}+N^{1/3}\ln(N).
\end{equation}
Separating the contributions from momenta in $P_\mathrm{B}$ and $P_\mathrm{I}$, $E_{\cV_N} + E_{\cK}$ can be written as
\begin{align}
&E_{\cK}+E_{\cV_N}=\tr_{\Fock_{\mathrm{I}}} \Big[\Big(\sum_{p\in P_\mathrm{I}}p^2a^*_p a_p\Big)  G_{\mathrm{free}}\Big]+\frac{\widetilde N_0}{L^3}\sum_{q\in P_\mathrm{I}}(\hat v_N\ast\hat f_N)(q)\gamma(q) \nonumber \\
&\quad+\tr_{\Fock_{\mathrm{B}}} \,\Big[\sum_{p\in P_\mathrm{B}}p^2a^*_p a_p+\frac{\widetilde N_0}{2L^3}\sum_{q\in P_\mathrm{B}}(\hat v_N\ast\hat f_N)(q)\Big(2a^*_{q}a_{q}+( z/|z|)^2a^*_{q}a^*_{-q}+(\bar z/|z|)^2a_{q}a_{-q}\Big)\,G_{\mathrm{B}}(z)\Big] \nonumber \\
&\quad+\frac{4\pi\frak{a}_N}{L^3}\biggl[\int_{\mathbb{C}} |z|^4 \zeta(z) \de z+2\widetilde N_0\sum_{\substack{ u\in P_\mathrm{L}\backslash\{0\}}}\gamma(u)+2\sum_{\substack{ u,v\in P_\mathrm{L}\backslash\{0\}}}\gamma(v)\gamma(u)\biggl]+\text{E}_1
\label{eq:sumHkHv}
\end{align}
with
\begin{align}
\text{E}_1&=\frac{1}{L^6}\sum_{\substack{p_1\in P_\mathrm{H}\\u_1,v_1,u_2\in P_\mathrm{L}}}\hspace{-0.3cm}\eta_{p_1}\Big[(p_1+u_1-u_2)^2\eta_{p_1+u_1-u_2}+\frac{1}{2}\hat v_N(p_1+u_1-u_2) \nonumber \\
&\hspace{3cm}+\frac{1}{2L^3}\sum_{\substack{p_2\in P_\mathrm{H}}}\hat v_N(p_1+p_2+u_1-u_2)\eta_{p_2}\Big]\,\tr[a^*_{v_1}a^*_{u_1}a_{v_2}a_{u_2}\Gamma_0].
\end{align}
An application of the scattering equation in \eqref{npscatt} allows us to write $E_1$ as
\begin{align}
\text{E}_1&=\frac{1}{L^3}\sum_{\substack{p_1\in P_\mathrm{H}\\u_1,v_1,u_2\in P_\mathrm{L}}}\hspace{-0.3cm}\eta_{p_1}\biggl[\lambda_{N} \big(\hat{\mathds{1}}_{|x|\leq\ell}*\hat{f}_N\big)(p_1+u_1-u_2)-\frac{1}{2 L^3 }\sum_{\substack{p_2\in P_\mathrm{H}^c}}\hat v_N(p_1+p_2+u_1-u_2)\eta_{p_2}\biggl]\tr[a^*_{v_1}a^*_{u_1}a_{v_2}a_{u_2}\Gamma_0] \nonumber \\
&=:\text{E}_{11}+\text{E}_{12}. 
\label{eq:Sect3Andi6}
\end{align}
We now prove
\begin{equation}\label{eq:Escat}
|\text{E}_{11}| \lesssim L^{-2} N^{-1/2+\delta_{\mathrm{H}}/2} \quad \text{ and } \quad |\text{E}_{12}| \lesssim L^{-2} N^{1-\delta_{\mathrm{H}}}.
\end{equation}
To obtain the bound for $\text{E}_{11}$, we first note that \eqref{eq:Dw} and $f_N \leq 1$ imply
\begin{equation}
	\sum_{\substack{p_1\in \L^*}}| \big(\hat{\mathds{1}}_{|x|\leq\ell}*\hat{f}_N\big)(p_1)|^2 =  \|\mathds{1}_{|x|\leq\ell} f_N\|^2 \lesssim L^3.
\end{equation}
Applications of the above bound, Cauchy-Schwarz, \eqref{eq:el} (which implies $\lambda_N \lesssim 1/(N L^2)$), \eqref{eq:sumPHC} and Lemma~\ref{lm:Tr} prove the bound for $\text{E}_{11}$. An application of \eqref{eq:etabound} shows
\begin{equation}
	\sup_{u \in \L^*} \sum_{\substack{p\in P_\mathrm{H}}}|\hat v_N(p + u)||\eta_{p}|\leq \frac{L \|\hat v\|_\infty}{N^2}\sum_{\substack{p\in P_\mathrm{H},\\|p|\leq N}}\frac{1}{p^2}+\Biggl(\sum_{\substack{p\in \L^*}}|\hat v_N(p)|^2\Biggl)^{1/2}\Biggl(\sum_{\substack{|p|> N}}\frac{L^2}{N^2p^4}\Biggl)^{1/2} \lesssim L^4 N^{-1}.
\end{equation}
The bound for $\text{E}_{12}$ follows when we combine this bound, \eqref{eq:sumPHCC} and Lemma~\ref{lm:Tr}. 

Next, we use \eqref{eq:Comp4b} and \eqref{eq:SSV} to estimate
\begin{equation}\label{eq:gfin1}
\frac{\widetilde N_0}{L^3}\sum_{q\in P_\mathrm{I}}(\hat v_N\ast\hat f_N)(q)\gamma(q)\leq \frac{8\pi\frak{a}_N\widetilde N_0}{L^3}\sum_{q\in P_\mathrm{I}}\gamma(q)+ C L^{-2}.
\end{equation}
Finally, we can replace $\widetilde N_0$ by $N_0(\beta,N,L)$ in the second line of \eqref{eq:sumHkHv}. More precisely, we apply Lemmas~\ref{lem:boundsgamma}, \ref{lem:BoundWidetildeN0} and \eqref{eq:SSV} and find that the error term is bounded by
\begin{equation}\label{eq:gfin2}
\frac{|\widetilde N_0-N_0|}{L^3}\sum_{q\in P_\mathrm{B}}\big|(\hat v_N\ast\hat f_N)(q)\big|\big(2\gamma(q)+|\alpha(q)|+|\overline{\alpha(q)}|\big) \lesssim L^{-2} N^{1/3+\delta_{\mathrm{B}}}.
\end{equation}
In combination, \eqref{eq:enf}, \eqref{eq:sumHkHv}--\eqref{eq:Escat}, \eqref{eq:gfin1} and \eqref{eq:gfin2} prove \eqref{eq:FH}. 

\section{Bound for the entropy}\label{entr}

In this section we establish the following lower bound for the entropy of our trial state.

\begin{proposition}
	\label{prop:Entropy}
	There exists a constant $C>0$ such that the entropy of the state $\Gamma$ in \eqref{trialstate} satisfies
	\begin{equation}
		S(\Gamma) \geq \int_{\mathbb{C}} S(G_{\mathrm{B}}(z)) \zeta(z) \de z + S(G_{\mathrm{free}}) + S(\zeta) - C N^{-1 + \delta_{\mathrm{H}}},
		\label{eq:Entr}
	\end{equation}
	where 
	\begin{equation}
		S(\zeta) = - \int_{\mathbb{C}} \zeta(z) \ln(\zeta(z)) \de z
		\label{eq:ClassicalEntropy}
	\end{equation}
	denotes the classical entropy of the probability distribution $\zeta$.
\end{proposition}

The remainder of this section is devoted to the proof of Proposition~\ref{prop:Entropy}. In the first step, we estimate the influence of the correlation structure with the following lemma. It appeared for the first time in \cite[Lemma~2]{Sei2006}. 

\begin{lemma}
	\label{lem:entropy1}
	Let $\Gamma$ be a density matrix on some Hilbert space with eigenvalues $\{ \lambda_{\alpha} \}_{\alpha \in \mathbb{N}}$, let $\{ P_{\alpha} \}_{\alpha \in \mathbb{N}}$ be a family of one-dimensional orthogonal projection (for which $P_{\alpha_1} P_{\alpha_2} = \delta_{\alpha_1,\alpha_2} P_{\alpha_1}$ need not necessarily be true), and define $\hat{\Gamma} = \sum_{\alpha} \lambda_{\alpha} P_{\alpha}$. Then we have
	\begin{equation}
		S(\hat{\Gamma}) \geq S(\Gamma) - \ln \mathrm{Tr}\left( \sum_{\alpha} P_{\alpha} \hat{\Gamma} \right).
		\label{eq:LemEntropyCorr}
	\end{equation}
\end{lemma}

An application of Lemma~\ref{lem:entropy1} shows
\begin{equation}
	S(\Gamma) \geq S(\Gamma_0) - \ln \tr \left( \sum_{\alpha'} | \phi_{\alpha'} \rangle \langle \phi_{\alpha'} | \ \Gamma \right) = S(\Gamma_0) - \ln\left( \sum_{\alpha,\alpha' } \lambda_{\alpha} \ | \langle \phi_{\alpha}, \phi_{\alpha'} \rangle |^2 \right)
	\label{eq:entropy2b}
\end{equation}
with $\Gamma_0$ in \eqref{eq:modfreestate} and $\lambda_{\alpha}$, $\phi_{\alpha}$ in \eqref{trialstate}. Let us have a closer look at the term inside the logarithm. Using $\Vert (1+B) \psi_{\alpha} \Vert \geq 1$, $B^* \psi_{\alpha} = 0$ and the fact that $\{ \psi_{\alpha} \}_{\alpha \in \mathbb{N}}$ is an orthonormal set, we see that
\begin{align}
	\sum_{\alpha,\alpha'} \lambda_{\alpha} \ | \langle \phi_{\alpha}, \phi_{\alpha'} \rangle |^2 &\leq \sum_{\alpha,\alpha'} \lambda_{\alpha} \ | \langle \psi_{\alpha}, (1+B^*)(1+B) \psi_{\alpha'} \rangle |^2 = 1 + 2 \sum_{\alpha} \lambda_{\alpha} \langle \psi_{\alpha}, B^* B \psi_{\alpha} \rangle + \sum_{\alpha} \lambda_{\alpha} \langle \psi_{\alpha}, (B^* B)^2 \psi_{\alpha} \rangle \nonumber \\
	&\leq 1+\delta + (1+\delta^{-1}) \tr \left[ (B^* B)^2 \Gamma_0 \right] 	\label{eq:entropy3b}
\end{align}
holds for $\delta>0$. 

The last term on the r.h.s. reads
\begin{align}
	\tr&\left[ ( B^* B )^2 \Gamma_0 \right] = \frac{1}{16|\Lambda|^4} \sum_{p_i \in P_{\mathrm{H}}; u_i, v_i \in P_{\mathrm{L}} } \prod_{i=1}^4 \eta_{p_i} \nonumber \\
	&\hspace{2cm}\times \tr \left[ a^*_{u_1} a^*_{v_1} a_{u_1+p_1} a_{v_1-p_1} a^*_{u_2+p_2} a^*_{v_2-p_2} a_{u_2} a_{v_2} a^*_{u_3} a^*_{v_3} a_{u_3+p_3} a_{v_3-p_3} a^*_{u_4+p_4} a^*_{v_4-p_4} a_{u_4} a_{v_4}  \Gamma_0 \right]. 	\label{eq:entropy5b}
\end{align}
Since no momenta in $P_{\mathrm{H}} - P_{\mathrm{L}}$ are present in the state $\Gamma_0$, we know that the operators with momenta in $P_{\mathrm{H}} - P_{\mathrm{L}}$ need to be paired among each other in order to obtain a non-zero contribution. A short computation therefore shows that the r.h.s. of \eqref{eq:entropy5b} is bounded from above by a constant times
\begin{align}
	\frac{1}{|\Lambda|^4} \left( \sum_{p, q \in P_{\mathrm{H}} + P_{\mathrm{L}}} | \eta_{p} \eta_{q} | \right)^2  \sum_{u_i, v_i \in P_{\mathrm{L}}} \tr \left[ a^*_{u_1} a^*_{v_1} a_{u_2} a_{v_2} a^*_{u_3} a^*_{v_3} a_{u_4} a_{v_4} \Gamma_0 \right]. 
\end{align}
We apply the Cauchy-Schwarz inequality and \eqref{eq:sumPHC} to bound the proportional to $\eta_p \eta_q$ by $\sum_{p \in P_{\mathrm{H}} + P_{\mathrm{L}} } \eta_p^2 \lesssim L^{6} N^{-3 + \delta_{\mathrm{H}}}$. Afterwards, we use Lemma~\ref{lm:Tr} to show that the second factor is bounded by a constant times $N^4$. When we put the above considerations together together, use $\delta_{\mathrm{L}} + \delta_{\mathrm{H}} < 2/3$, and choose $\delta = N^{-1 + \delta_{\mathrm{H}}}$, we find
\begin{equation}
	\sum_{\alpha,\alpha'} \lambda_{\alpha} \ | \langle \phi_{\alpha}, \phi_{\alpha'} \rangle |^2 \leq 1 + C N^{-1 + \delta_{\mathrm{H}}} 
\end{equation}
as well as 
\begin{equation}
	S(\Gamma) \geq S(\Gamma_0) - C N^{-1 + \delta_{\mathrm{H}}}.
	\label{eq:IntermediateEntropyBound}
\end{equation}
It remains to find a lower bound for the entropy of $\Gamma_0$.

To that end, we need the following lemma, which provides us with a Berezin--Lieb inequality in the spirit of \cite{Berezin1972,Lieb1973}.
\begin{lemma}
	\label{lem:BerezinLieb}
	Let $\{ G(z) \}_{z \in \mathbb{C}}$ be a family of states on a Hilbert space, let $p: \mathbb{C} \to \mathbb{R}$ be a probability distribution and define the state
	\begin{equation}
		\Gamma = \int_{\mathbb{C}} |z \rangle \langle z | \otimes G(z) p(z) \de z.
		\label{eq:entropy2}
	\end{equation} 
	Then we have
	\begin{equation}
		S(\Gamma) \geq \int_{\mathbb{C}} S(G(z)) p(z) \de z + S(p) \quad \text{ with } \quad S(p) = - \int_{\mathbb{C}} p(z) \ln (p(z)) \de z.
		\label{eq:entropy3}
	\end{equation}
\end{lemma}
\begin{proof}
	We use the spectral theorem to write
	\begin{equation}
		G(z) = \sum_{\alpha} g_{\alpha}(z) \ | v_{\alpha}(z) \rangle \langle v_{\alpha}(z) | \quad \text{ as well as } \quad |z \rangle \langle z | \otimes G(z) = \sum_{\alpha} g_{\alpha}(z) \ | z \otimes v_{\alpha}(z) \rangle \langle z \otimes v_{\alpha}(z) |.
		\label{eq:entropy4}
	\end{equation} 
	Because $G(z)$ is a state for fixed $z \in \mathbb{C}$, we know that $\{ v_{\alpha}(z) \}_{\alpha \in \mathbb{N}}$ is an orthonormal basis. In combination with the completeness relation $\int |z \rangle \langle z | \de z = \mathds{1}$, this implies
	\begin{equation}
		\int_{\mathbb{C}} \sum_{\alpha=1}^{\infty} | \langle w, z \otimes v_{\alpha}(z) \rangle |^2 \de z = 1
		\label{eq:entropy5}
	\end{equation}
	for any fixed vector $w$ with $\Vert w \Vert = 1$. 
	
	For $x \in [0,1]$ we define the function $\varphi(x) = -x \ln(x)$ and denote by $\{ w_{\alpha} \}_{\alpha \in \mathbb{N}}$ the eigenbasis of $\Gamma$. An application of Jensen's inequality shows
	\begin{align}
		\tr \varphi(\Gamma) = \sum_{\alpha} \varphi\left( \langle w_{\alpha}, \Gamma w_{\alpha} \rangle \right) &= \sum_{\alpha} \varphi\left( \int_{\mathbb{C}} \sum_{\alpha'} g_{\alpha'}(z) | \langle w_{\alpha}, z \otimes v_{\alpha'}(z) \rangle |^2 p(z) \de z \right) \nonumber \\
		&\geq \sum_{\alpha} \int_{\mathbb{C}} \sum_{\alpha'} \varphi(g_{\alpha'}(z) p(z)) | \langle w_{\alpha}, z \otimes v_{\alpha'}(z) \rangle |^2 \de z = \int_{\mathbb{C}} \sum_{\alpha'}  \varphi(g_{\alpha'}(z) p(z)) \de z.
		\label{eq:entropy6}
	\end{align}
	This is justified because $x \mapsto \varphi(x)$ is concave and \eqref{eq:entropy5} holds. In the last step we used that $\{w_{\alpha}\}_{\alpha \in \mathbb{N}}$ is a complete orthonormal basis. With $xy \ln(xy) = xy \ln(x)+ x y \ln(y)$ for $x,y \geq 0$ and $\sum_{\alpha'} g_{\alpha'}(z) = 1$, we see that the r.h.s. of \eqref{eq:entropy6} equals the r.h.s. of the inequality in \eqref{eq:entropy3}, which proves the claim.
\end{proof}

An application of Lemma~\ref{lem:BerezinLieb} on the r.h.s. of \eqref{eq:IntermediateEntropyBound} and the additivity of the entropy w.r.t. tensor products prove Proposition~\ref{prop:Entropy}.
\section{Proof of the main results}
\label{sec:ProofOfThm}
Propositions~\ref{trHG} and \ref{prop:Entropy} imply the following upper bound for the free energy of our trial state:
\begin{align}
	\tr[\mathcal{H}_N \Gamma] - \frac{1}{\beta} S(\Gamma) \leq& \sum_{p \in P_{\mathrm{I}}} p^2 \tr_{\mathscr{F}_{\mathrm{I}}}[ a_p^*a_p G_{\mathrm{free}}] - \frac{1}{\beta} S(G_{\mathrm{free}}) + \int_{\mathbb{C}} \left( \tr_{\mathscr{F}_{\mathrm{B}}}[\mathcal{H}^{\mathrm{B}} G_{\mathrm{B}}(z)] - \frac{1}{\beta} S(G_{\mathrm{B}}(z)) \right) \zeta(z) \de z  \nonumber \\
	&+ \mu_0 \sum_{p \in P_{\mathrm{B}}} \gamma(p) + \frac{4 \pi \mathfrak{a}}{NL^3} \int_{\mathbb{C}} |z|^4 \zeta(z) \de z - \frac{1}{\beta} S(\zeta) \nonumber \\
	&+\frac{4\pi\frak{a}}{NL^3}\biggl[ 2\widetilde N_0\sum_{\substack{ u\in P_{\mathrm{L}}\backslash\{0\}}}\gamma(u)+ 2\widetilde N_0\sum_{u\in P_\mathrm{I}}\gamma(u) +2\sum_{\substack{ u,v\in P_{\mathrm{L}}\backslash\{0\}}}\gamma(v)\gamma(u)\biggl] + L^{-2} \mathcal{E}_{\mathcal{H}_N}.
	\label{eq:Sec51}
\end{align}
The Bogoliubov Hamiltonian $\mathcal{H}^{\mathrm{B}}$ and the error term $\mathcal{E}_{\mathcal{H}_N}$ are defined in \eqref{eq:HB} and \eqref{eq:FEr}, respectively. The first two terms on the r.h.s. can be written as 
\begin{equation}
	\frac{1}{\beta} \sum_{p \in P_{\mathrm{I}}} \ln\left( 1 - \exp(- \beta(p^2 - \mu_0)) \right) + \mu_0 \sum_{p \in P_{\mathrm{I}}} \frac{1}{\exp(\beta(p^2 - \mu_0))-1}
	\label{eq:Sec52}
\end{equation}
with $\mu_0$ in \eqref{eq:ChemicalPotentialIdealGas}, and an application of Lemma~\ref{lem:DiagBogoH} shows that the third term equals
\begin{equation}
	E_0 + \frac{1}{\beta} \sum_{p \in P_{\mathrm{B}}} \ln\left( 1 - \exp(- \beta \epsilon(p) ) \right).
	\label{eq:Sec53}
\end{equation}
We refer to the same lemma also for the definitions of $E_0$ and $\epsilon(p)$. One easily checks that $E_0$ is negative and can be dropped for an upper bound. Let us define 
\begin{equation}\label{eq:epstilde}
\tilde{\epsilon}(p) = \sqrt{p^2-\mu_0} \sqrt{p^2-\mu_0 + 16 \pi \mathfrak{a}_N \varrho_0}.
\end{equation} 
The function $x \mapsto \ln(1 - \exp(-x))$ is monotone increasing ($x \geq 0$). This and \eqref{eq:SSV} allow us to replace $\hat{v}_N \ast \hat{f}_N(p)$ in the definition of $\epsilon(p)$ by $8 \pi \mathfrak{a}_N (1 + C/N)$. Moreover, a first order Taylor expansion then shows
\begin{equation}
	\frac{1}{\beta} \sum_{p \in P_{\mathrm{B}}} \ln\left( 1 - \exp(- \beta \epsilon(p) ) \right) \leq \frac{1}{\beta} \sum_{p \in P_{\mathrm{B}}} \ln\left( 1 - \exp(- \beta \tilde\epsilon(p) ) \right) + \frac{C N_0}{L^2 N^2} \sum_{p \in P_{\mathrm{B}}} \frac{p^2 - \mu_0}{\exp(\beta(p^2-\mu_0))-1} \frac{1}{p^2}.
	\label{eq:Sec54}
\end{equation} 
Using $(\exp(x)-1)^{-1} \leq 1/x$ for $x \geq 0$ and $\delta_{\mathrm{B}} < 1/3$, we check that the second term on the r.h.s. is bounded by a constant times $1/L^2$. Moreover, from Lemma~\ref{lem:freeEnergyBog} we know that
\begin{align}
	&\frac{1}{\beta} \sum_{p \in P_{\mathrm{B}}} \ln\left( 1 - \exp(- \beta \tilde{\epsilon}(p) ) \right) \leq \frac{1}{\beta} \sum_{p \in P_{\mathrm{B}}} \ln\left( 1 - \exp(- \beta (p^2 - \mu_0) ) \right) + 8 \pi \mathfrak{a}_N \varrho_0 \sum_{p \in P_{\mathrm{B}}} \frac{1}{\exp(\beta(p^2 - \mu_0))-1} \nonumber \\
	&\hspace{3cm} - \frac{1}{2 \beta} \sum_{p \in \Lambda^*_+} \left[ \frac{16 \pi \mathfrak{a}_N \varrho_0(\beta,N,L)}{p^2} - \ln\left( 1 + \frac{16 \pi \mathfrak{a}_N \varrho_0(\beta,N,L)}{p^2} \right) \right] + \frac{C}{L^2} \left( N^{\delta_{\mathrm{B}}} + N^{2/3-\delta_{\mathrm{B}}} \right)
\end{align}
holds.

Next, we have a closer look at the first term in the second line of \eqref{eq:Sec51}. In \eqref{eq:Comp8a}--\eqref{eq:Comp9} we showed 
\begin{equation}
	\Bigg| \sum_{p \in P_{\mathrm{B}}} \left( \gamma(p) - \frac{1}{\exp(\beta(p^2-\mu_0))-1} \right) \Bigg| \lesssim \frac{N_0 L^2}{N \beta} + \frac{N_0^2}{N^2},
	\label{eq:Sec55}
\end{equation}
and hence
\begin{equation}
	\mu_0 \sum_{p \in P_{\mathrm{B}}} \gamma(p) \leq \mu_0 \sum_{p \in P_{\mathrm{B}}} \frac{1}{\exp(\beta(p^2-\mu_0))-1} + C L^{-2} N^{1/3}.
	\label{eq:Sec56}
\end{equation}
To obtain the second bound we also used $-\mu_0 = \ln(1+1/N_0)/\beta \leq 1/(\beta N_0)$ (which follows from \eqref{eq:DensityBEC}). In combination, the considerations in \eqref{eq:Sec52}--\eqref{eq:Sec56} and $\delta_{\mathmbox{B}} < 1/3$ imply
\begin{align}
	\sum_{p \in P_{\mathrm{I}}}& p^2 \tr_{\mathscr{F}_{\mathrm{I}}}[ a_p^*a_p G_{\mathrm{free}}] - \frac{1}{\beta} S(G_{\mathrm{free}}) + \int_{\mathbb{C}} \left( \tr_{\mathscr{F}_{\mathrm{B}}}[\mathcal{H}^{\mathrm{B}} G_{\mathrm{B}}(z)] - \frac{1}{\beta} S(G_{\mathrm{B}}(z)) \right) \zeta(z) \de z + \mu_0 \sum_{p \in P_{\mathrm{B}}} \gamma(p) \nonumber \\
	\leq& \frac{1}{\beta} \sum_{p \in P_{\mathrm{L}} \backslash \{ 0 \}} \ln\left( 1 - \exp(- \beta(p^2 - \mu_0)) \right) + \sum_{p \in P_{\mathrm{L}} \backslash \{ 0 \}} \frac{\mu_0}{\exp(\beta(p^2 - \mu_0))-1} + \sum_{p \in P_{\mathrm{B}}} \frac{8 \pi \mathfrak{a}_N \varrho_0}{\exp(\beta(p^2 - \mu_0))-1} \nonumber \\
	&- \frac{1}{2 \beta} \sum_{p \in \Lambda^*_+} \left[ \frac{16 \pi \mathfrak{a}_N \varrho_0(\beta,N,L)}{p^2} - \ln\left( 1 + \frac{16 \pi \mathfrak{a}_N \varrho_0(\beta,N,L)}{p^2} \right) \right] + C L^{-2} ( N^{1/3} + N^{2/3- \delta_{\mathrm{B}}} ). 
	\label{eq:Sec57}
\end{align}
In the first two terms on the r.h.s. it remains to replace the sums over $P_{\mathrm{L}} \backslash \{ 0 \}$ by sums over $\Lambda_+^*$. One easily checks that this can be done at the expense of an error term that is bounded by a constant times $L^{-2} \exp(-cN^{2\delta_{\mathrm{L}}})$ with some $c>0$.

The second and the third term in the second line of \eqref{eq:Sec51} equal 
\begin{equation}
	F^{\mathrm{BEC}}(\beta,\widetilde{N}_0,L,\mathfrak{a}_N) = -\frac{1}{\beta} \ln\left( \int_{\mathbb{C}} \exp\left( - \beta \left( 4 \pi \mathfrak{a}_N L^{-3} |z|^4 - \widetilde{\mu} |z|^2 \right) \right) \de z \right) + \widetilde{\mu} \widetilde{N}_0,
	\label{eq:Sec58}
\end{equation}
where the chemical potential $\widetilde{\mu}$ is chosen such that the Gibbs distribution $\zeta$ in \eqref{eq:GibbsDistributionzeta} satisfies \eqref{eq:WidetildeN0}. The first term on the r.h.s. is a concave function of $\widetilde{\mu}$. But this implies
\begin{align}
	-\frac{1}{\beta} \ln\left( \int_{\mathbb{C}} \exp\left( - \beta \left( 4 \pi \mathfrak{a}_N L^{-3} |z|^4 - \widetilde{\mu} |z|^2 \right) \right) \de z \right) \leq& -\frac{1}{\beta} \ln\left( \int_{\mathbb{C}} \exp\left( - \beta \left( 4 \pi \mathfrak{a}_N L^{-3} |z|^4 - \mu |z|^2 \right) \right) \de z \right) + \mu N_0  \nonumber \\
	& - \widetilde{\mu} N_0.
	\label{eq:Sec59}
\end{align}
Here we also used that the first derivative of the first term on the r.h.s. equals $-N_0$. 

The identity $\widetilde{N}_0 + \sum_{p \in P_{\mathrm{L}} \backslash \{ 0 \}} \gamma(p) = N+\Delta N$ allows us to bound the terms in the third line of \eqref{eq:Sec51} plus the third term on the r.h.s. of \eqref{eq:Sec57} as 
\begin{equation}
	\frac{4\pi\frak{a}}{NL^3}\biggl[ 2 N^2 - 2 N_0^2 + 2 ( N_0^2 - \widetilde{N}_0^2 )  + 2 N_0 \sum_{p \in P_\mathrm{B}} \frac{1}{\exp(\beta(p^2 - \mu_0))-1} - 2 \widetilde{N}_0 \sum_{u \in P_{\mathrm{B}}} \gamma(u) \bigg ]+CN^{\delta_\mathrm{H}}.
	\label{eq:Sec60}
\end{equation}
In the following, we denote $\gamma_0(p) = \exp(\beta(p^2-\mu_0)-1)^{-1}$. Another algebraic manipulation,  equations  \eqref{eq:NG}, \eqref{eq:defGam}, \eqref{eq:Comp4c} and \eqref{eq:Sec55}, the bound $\sum_{p \in P_{\mathrm{L}}^{\mathrm{c}}} \gamma(p) \lesssim \exp(-c N^{2\delta_{\mathrm{L}}})$ for some $c>0$, and $\delta_{\mathrm{H}} < 2/3, \delta_{\mathrm{L}} > 0$ imply
\begin{align}
	N_0^2 - \widetilde{N}_0^2 \leq&\, 2 \widetilde{N}_0 \sum_{p \in P_{\mathrm{B}}} \left( \gamma(p) - \gamma_0(p) \right) + 2 \,\bigg( \sum_{p \in P_{\mathrm{B}}} \gamma(p) \bigg) \sum_{p \in P_{\mathrm{B}}} \left( \gamma(p) - \gamma_0(p) \right)  \nonumber \\
	&+ \sum_{p \in P_{\mathrm{B}}} \left( \gamma_0(p) - \gamma(p) \right) \sum_{p \in P_{\mathrm{B}}} \left( \gamma(p) + \gamma_0(p) \right) + C [ N^{1+\delta_{\mathrm{H}}} + \exp(-c N^{2 \delta_{\mathrm{L}}}) ] \nonumber \\
	\leq&\, 2 \widetilde{N}_0 \sum_{p \in P_{\mathrm{B}}} \left( \gamma(p) - \gamma_0(p) \right) + C [ N^{1+\delta_{\mathrm{H}}} + N^{4/3 + \delta_{\mathrm{B}}}].
	\label{eq:Sec61}
\end{align}
A similar argument that additionally uses Lemma~\ref{lem:BoundWidetildeN0} and $\delta_{\mathrm{H}} < 2/3$ shows  
\begin{equation}
	2 N_0 \sum_{p \in \mathrm{B}} \gamma_0(p) - 2 \widetilde{N}_0 \sum_{p \in P_{\mathrm{B}}} \gamma(p) \leq -2 \widetilde{N}_0 \sum_{p \in P_{\mathrm{B}}} \left( \gamma(p) - \gamma_0(p) \right) + C N^{4/3+\delta_{\mathrm{B}}}.
	\label{eq:Sec62}
\end{equation}
When we collect the results in \eqref{eq:Sec60}--\eqref{eq:Sec62}, we find that \eqref{eq:Sec60} is bounded from above by
\begin{equation}
	\frac{4\pi\frak{a}}{NL^3}\biggl[ 2 N^2 - 2 N_0^2 \biggr] + \frac{8\pi\frak{a}}{NL^3}\biggl[ \widetilde{N}_0 \sum_{p \in P_{\mathrm{B}}} \left( \gamma(p) - \gamma_0(p) \right) \biggr] + C L^{-2} [ N^{\delta_{\mathrm{H}}} + N^{1/3 + \delta_{\mathrm{B}}}].
	\label{eq:Sec63}
\end{equation}

We combine now the second term above with the last terms on the r.h. sides of \eqref{eq:Sec58} and \eqref{eq:Sec59}, that is, we  consider
\begin{equation}
	\frac{8\pi\frak{a}}{NL^3}\biggl[ \widetilde{N}_0 \sum_{p \in P_{\mathrm{B}}} \left( \gamma(p) - \gamma_0(p) \right) \biggr] + \widetilde{\mu}(\widetilde{N}_0 - N_0).
	\label{eq:Sec64}
\end{equation}
We distinguish two cases and assume first that $N_0 < N^{5/6 + \delta}$ for some $\delta>0$. In this case applications of \eqref{eq:Sec55} and Lemma~\ref{lem:BoundWidetildeN0} show that the first term in \eqref{eq:Sec64} is bounded by a constant times $L^{-2} N^{1/2+\delta}$.  Inspection of \eqref{eq:AppBEC4} (recall that $\widetilde{N}_0 = \int |z|^2 \zeta(z) \de z$) shows $|\widetilde \mu | \lesssim 1/(\beta \widetilde{N}_0) + \widetilde{N}_0/(L^2N)$. We use this estimate, $N_0 \geq N^{2/3}$ (this implies, by Lemma~\ref{lem:BoundWidetildeN0}, $\widetilde N_0 \gtrsim N^{2/3}$ ), and again Lemma~\ref{lem:BoundWidetildeN0} to bound the second term in \eqref{eq:Sec64} by a constant times $L^{-2} [ N^{\delta_{\mathrm{H}}} + N^{-1/6 + \delta_{\mathrm{H}} + \delta} + N^{1/3 + 2 \delta} ]$. If $N_0 \geq N^{5/6 + \delta}$ we apply part (a) of Lemma~\ref{lem:ChemPotBEC} to bound the second term in \eqref{eq:Sec64} from above by
\begin{align}
	8 \pi \mathfrak{a}_N L^{-3} \widetilde{N}_0 (\widetilde{N}_0 - N_0) +  C \exp(-c N^{\delta}) \leq -8 \pi \mathfrak{a}_N L^{-3}\widetilde{N}_0 \sum_{p \in P_{\mathrm{B}}} \left( \gamma(p) - \gamma_0(p) \right) + C N^{\delta_{\mathrm{H}}}. \label{eq:Sec65}
\end{align}
To obtain the second bound, we also used $\sum_{p \in P_{\mathrm{L}}^{\mathrm{c}}} \gamma(p) \lesssim \exp(-c N^{2\delta_{\mathrm{L}}})$ for some $c>0$. We highlight that the first term on the r.h.s. of \eqref{eq:Sec65} cancels the first term in \eqref{eq:Sec64}. We collect the above considerations, make the assumption $0 < \delta < 1/6$, and find
\begin{equation}
	\frac{8\pi\frak{a}}{NL^3}\biggl[ \widetilde{N}_0 \sum_{p \in P_{\mathrm{B}}} \left( \gamma(p) - \gamma_0(p) \right) \biggr] + \widetilde{\mu}(\widetilde{N}_0 - N_0) \lesssim L^{-2} [ N^{\delta_\mathrm{H}} + N^{1/2 + \delta} ].
	\label{eq:Sec66}
\end{equation}
It remains to collect our results.

In combination, \eqref{eq:Sec57}--\eqref{eq:Sec59}, \eqref{eq:Sec63} and \eqref{eq:Sec66} imply the final upper bound
\begin{align}
	\tr[\mathcal{H}_N \Gamma] - \frac{1}{\beta} S(\Gamma) &\leq \frac{1}{\beta} \sum_{p \in \Lambda_+} \ln\left( 1 - \exp(- \beta(p^2 - \mu_0)) \right) + \mu_0 (N - N_0) + 8 \pi \mathfrak{a}_N L^3 (\varrho^2 - \varrho_0^2) + F^{\mathrm{BEC}}(\beta,N_0,L,\mathfrak{a}_N) \nonumber \\
	&- \frac{1}{2 \beta} \sum_{p \in \Lambda^*_+} \left[ \frac{16 \pi \mathfrak{a}_N \varrho_0(\beta,N,L)}{p^2} - \ln\left( 1 + \frac{16 \pi \mathfrak{a}_N \varrho_0(\beta,N,L)}{p^2} \right) \right] \nonumber \\
	&+ C L^{-2} [ N^{1- \delta_{\mathrm{H}}} + N^{\delta_{\mathrm{H}} + 2 \delta_{\mathrm{B}}} + N^{-1/3 + \delta_{\mathrm{H}} + 2 \delta_{\mathrm{L}}} + N^{1/3 + \delta_{\mathrm{B}}} + N^{1/2 + \delta} + N^{2/3 - \delta_{\mathrm{B}}} ]. \label{eq:Sec67}
\end{align}
The parameters $\delta_{\mathrm{L}}, \delta$ need to be strictly positive but can otherwise be chosen as small as we wish. The requirements $\delta_{\mathrm{L}} \leq 1/6, \delta \leq 1/12$ assure that they play no role in the optimization. The optimal choice $\delta_{\mathrm{H}} = 1/2 - \delta_{\mathrm{B}}$ with error $N^{1/2+\delta_{\mathrm{B}}}$ follows by combining the first and the second term. Moreover, the optimal choice $\delta_{\mathrm{B}} = 1/12$ results if we combine $N^{1/2+\delta_{\mathrm{B}}}$ and the last term in the last line of \eqref{eq:Sec67}. This leads to an overall error term that is bounded by a constant times $L^{-2} N^{7/12} \ll L^{-2} N^{2/3}$. We recall that the above bound holds under the assumption $N_0 \geq N^{2/3}$.

We now prove a second bound with another trial state (see also Remark~\ref{rmk:MainResults}.(g)) that holds without a restriction on $N_0$ (it, however, captures the correct behavior of the free energy only if $N_0 \ll N^{5/6}$). As undressed trial state we choose the Gibbs state 
\begin{equation}
	G_0 = \frac{\exp(-\beta( \de \Gamma(-\Delta - \mu_0) ))}{\tr_{\mathscr{F}}[\exp(-\beta( \de \Gamma(-\Delta - \mu_0) ))]}.
	\label{eq:Sec68}
\end{equation} 
We define the dressed trial state $\widetilde{\Gamma}$ as in \eqref{trialstate} with $\Gamma_0$ replaced by $G_0$. To obtain an upper bound for the free energy of $\widetilde{\Gamma}$ we can use a simpler version of the above proof. This is related to the following facts: (a) A coherent state in the definition of our trial state is not needed and the pairing function of $G_0$ equals zero. (b) The eigenfunctions of $G_0$ are also eigenfunctions of $\de \Gamma(-\Delta)$. Accordingly, the special treatment of momentum modes in $P_{\mathrm{B}}$ at several places in the proof is not needed. (c) Since $[G_0,\mathcal{N}] = 0$, we have $\tr[\mathcal{N} \widetilde{\Gamma}] = \tr[\mathcal{N} G_0]$. We therefore simply state the result and leave further details to the reader:
\begin{align}
	\tr[\mathcal{H}_N \widetilde{\Gamma}] - \frac{1}{\beta} S(\widetilde{\Gamma}) \leq F_0(\beta,N,L) + 8 \pi \mathfrak{a}_N L^{3} \varrho^2 + C L^{-2} N^{1/2}
	\label{eq:Sec69}
\end{align}
with $F_0$ defined above \eqref{eq:FreeEnergyIdealGasCond}.

We are now prepared to provide the missing proofs in Section~\ref{sec:main}. Theorem~\ref{thm:Main} follows from \eqref{eq:Sec67}, \eqref{eq:Sec69}, Proposition~\ref{prop:FreeEnergyBECb} and fact that the absolute value of the term in the second line of \eqref{eq:Sec67} is bounded by a constant times $N_0^2/(\beta N^2)$. The proof of Proposition~\ref{prop:FreeEnergyBECb} is provided in Appendix~\ref{app:CondensateFreeEnergy}, see Proposition~\ref{prop:FreeEnergyBEC}. Finally, Corollary~\ref{cor:MainCorollary} is a direct consequence of \eqref{eq:Sec67}, \eqref{eq:Sec69} and Proposition~\ref{prop:FreeEnergyBECb}.
\textbf{Acknowledgments.}
A. D. gratefully acknowledges funding from the Swiss National Science Foundation (SNSF) through the Ambizione grant PZ00P2 185851. It is our pleasure to thank Marco Caporaletti, Phan Thành Nam, Marcin Napiórkowski and Robert Seiringer for inspiring discussions.
\begin{appendix}

\begin{center}
	\huge \textsc{--- Appendix ---}
\end{center}

\section{The scattering equation}
\label{app:ScatteringEquation}

In this appendix we collect some known properties of the finite volume scattering equation \eqref{eq:NP}. It is convenient to define $f(Nx)=f_N(x)$, where $f$ satisfies the eigenvalue equation
\begin{equation} \label{eq:NP1}
	\bigg[-\Delta + \frac{v}{2}\bigg] f = \lambda_{\ell} f
\end{equation}
on the ball $|x| \leq N\ell$ with Neumann boundary conditions. It is normalized such that $f(x)=1$ holds for $|x|=N\ell$. By scaling, we have $N^2\lambda_\ell=\lambda_N$. In the next Lemma we collect the properties of $f_N,f$ and $\lambda_\ell$ that are useful for our analysis. The proof can be found in \cite[Appendix A]{BocBreCeSchl2020}.

\begin{lemma} \label{lemma:NPprop}
Let $v \in L^3(\mathbb{R}^3)$ be nonnegative, compactly supported and spherically symmetric. Fix $0<\ell <L/2 $ and let $f$ denote the solution to \eqref{eq:NP1} and $f_N$ the solution to \eqref{eq:NP}. For $N \in \mathbb{N}$ large enough the following properties hold true.
	\begin{enumerate}
		\item We have 
		\begin{equation}\label{eq:el}
			\lambda_{\ell} = \frac{3\frak{a}}{(N\ell)^3}(1+\mathcal{O}\big(\frak{a}/\ell N)\big).
		\end{equation}
		\item We have $0\leq f_{\ell}\leq 1$. Moreover there exists a constant $C > 0$ such that
		\begin{equation} \label{eq:sl}
			\left\lvert \int v(x)f(x)dx - 8\pi \frak{a} \right\lvert \leq \frac{C\frak{a}^2}{N\ell}.
		\end{equation}
		\item There exists a constant $C > 0$ such that, for all $x \in \mathbb{R}^3$,
		\begin{equation}\label{eq:Dw}
			1-f(x) \leq \frac{C}{1 + |x|} \quad \text{and} \quad \left|\nabla f(x)\right| \leq \frac{C}{1 + x^2}.
		\end{equation}
		\item There exists a constant $C > 0$ such that, for all $p \in \L^*_+$,
		\begin{equation}\label{eq:wp}
			\big|\widehat{(1-f_N)}(p)\big| \leq \frac{C}{N p^2}.
		\end{equation}
	\end{enumerate}
\end{lemma}

\medskip

%

By definition \eqref{eq:eta}, the function $\eta_p=-\widehat{(1-f_N)}(p)$ solves the equation
\begin{equation} \label{npscatt}
	p^2\eta_p + \frac{\hat v_N(p)}{2} + \frac{1}{2L^3} \sum_{q \in \Lambda^*} \hat{v}_N(p-q)\eta_q =  \frac{\lambda_{N}}{L^3} \sum_{q \in \Lambda^*}\hat{f}_N(p-q)\hat{\mathds{1}}_{|x|\leq\ell}(q),
\end{equation}
where $\hat{\mathds{1}}_{|x|\leq\ell}(q)$ is the Fourier coefficient of the characteristic function of the ball with radius $\ell$. Note that we have reinstated units in \eqref{npscatt}. Moreover, by \eqref{eq:Dw} and \eqref{eq:wp} we have
		\begin{equation} \label{eq:etabound}
			|\eta_p| \lesssim \frac{L}{Np^2}. 
		\end{equation}
 Inequality \eqref{eq:etabound} implies
\begin{equation} \label{eq:sumPHC}
	\sum_{p\in\L^*_+:|p|\geq \frac{1}{2L} N^{1-\delta_\mathrm{H}}}|\eta_p|^2 \lesssim L^6 N^{-3+\delta_\mathrm{H}}
\end{equation}
as well as
\begin{equation} \label{eq:sumPHCC}
	\sum_{p\in P_\mathrm{H}^{\mathrm{c}}}|\eta_p| \lesssim L^3 N^{-\delta_\mathrm{H}},
\end{equation}
where $P_\mathrm{H}^{\mathrm{c}}$ denotes the complement of $P_\mathrm{H}$ in \eqref{eq:PH}.
\section{Bogoliubov free energy}
\label{app:BogFreeEnergy}
The goal of this section is to prove the following lemma.
\begin{lemma}
	\label{lem:freeEnergyBog}
	We consider the limit $N \to \infty$, $\beta = \kappa \beta_{\mathrm{c}}$ with $\kappa \in (0,\infty)$ and $\beta_{\mathrm{c}}$ in \eqref{eq:BECPhaseTransition}. Recall  definition \eqref{eq:epstilde} for $\tilde{\epsilon}(p)$. There exists a constant $C>0$ such that
	\begin{align}
		\frac{1}{\beta} \sum_{p \in P_{\mathrm{B}}} \ln \left( 1- \exp(-\beta \tilde{\epsilon}(p)) \right) \leq& \frac{1}{\beta} \sum_{p \in P_{\mathrm{B}}} \ln \left( 1- \exp(-\beta (p^2-\mu_0)) \right) + 8 \pi \mathfrak{a}_N \varrho_0 \sum_{p \in P_{\mathrm{B}}} \frac{1}{\exp(\beta(p^2-\mu_0))-1} \label{eq:AppB1} \\
		&- \frac{1}{2 \beta} \sum_{p \in \Lambda_+^*} \left[ \frac{16 \pi \mathfrak{a}_N \varrho_0}{p^2} - \ln\left( 1 + \frac{16 \pi \mathfrak{a}_N \varrho_0}{p^2} \right) \right] + \frac{CN_0^2}{N^2} \left[ \frac{N^{\delta_{\mathrm{B}}}}{L^2} + \frac{1}{\beta N^{\delta_{\mathrm{B}}}} + \frac{L^2}{\beta^2 N_0} \right]. \nonumber
	\end{align}
\end{lemma}
\begin{proof}
	We first assume $\mu_0 = 0$ and then comment on how to adjust the proof to $\mu_0 < 0$. Let us define the function
	\begin{equation}
		F(\alpha) = \sum_{p \in \beta^{1/2} P_{\mathrm{B}}} \ln\left(1-\exp\left(-|p| \sqrt{p^2 + \alpha}\right)\right).
		\label{eq:AppB2}
	\end{equation}
	For $\alpha = 16 \pi \mathfrak{a}_N \varrho_0 \beta$ it equals $\beta$ times the l.h.s. of \eqref{eq:AppB1}. In the following we derive an asymptotic expansion of $F$ for small values of $\alpha$. We also define the functions
	\begin{equation}
		g_1(x) = \frac{1}{\exp(x)-1} \quad \text{ and } \quad g_2(x) = - \frac{1}{4 \sinh^2(x/2)}
		\label{eq:AppB3}
	\end{equation}
	and note that the bounds
	\begin{equation}
		 g_1(x) \geq \frac{1}{x} - C \quad \text{ and } \quad g_2(x) \leq \frac{-1}{x^2} + \frac{C}{x}
		\label{eq:AppB3b}
	\end{equation}
	hold for $0 < x \leq 1$. The first and the second derivative of $F$ can be written in terms of $g_1$ and $g_2$ as
	\begin{align}
		F'(\alpha) &= \sum_{p \in \beta^{1/2} P_{\mathrm{B}}} g_1\left( |p| \sqrt{p^2 + \alpha} \right) \frac{|p|}{2 \sqrt{p^2 + \alpha}}, \nonumber \\ 
		F''(\alpha) &= \frac{1}{4} \sum_{p \in \beta^{1/2} P_{\mathrm{B}}} \left[ g_2\left( |p| \sqrt{p^2 + \alpha} \right) \frac{p^2}{p^2 + \alpha} - g_1\left( |p| \sqrt{p^2 + \alpha} \right) \frac{|p|}{(p^2 + \alpha)^{3/2}} \right],
		\label{eq:AppB4}
	\end{align}
	and hence
	\begin{equation}
		F(\alpha) - F(0) - F'(0) \alpha = \frac{1}{4} \int_0^{\alpha} \sum_{p \in \beta^{1/2} P_{\mathrm{B}}} \left[ g_2\left( |p| \sqrt{p^2 + t} \right) \frac{p^2}{p^2 + t} - g_1\left( |p| \sqrt{p^2 + t} \right) \frac{|p|}{(p^2 + t)^{3/2}} \right] (\alpha-t) \de t.
		\label{eq:AppB5}
	\end{equation}
	It remains to investigate the r.h.s. of this above identity.
	
	Using the bounds in \eqref{eq:AppB3b}, we see that it is bounded from above by
	\begin{equation}
		-\frac{1}{2} \int_0^{\alpha} \sum_{p \in \beta^{1/2} P_{\mathrm{B}}} \frac{(\alpha-t)}{(p^2 + t)^2} \de t + C \alpha \int_0^{\alpha} \sum_{p \in \beta^{1/2} P_{\mathrm{B}}} \frac{|p|}{(p^2+t)^{3/2}} \de t. 
		\label{eq:AppB6}
	\end{equation}
	Here, the integral in the second term is bounded by
	\begin{equation}
		\sum_{p \in \beta^{1/2} P_{\mathrm{B}}} \int_0^{\alpha} \frac{1}{p^2+t} \de t = \sum_{p \in \beta^{1/2} P_{\mathrm{B}}} \ln\left( 1 + \frac{\alpha}{p^2} \right) \leq \sum_{p \in \beta^{1/2} P_{\mathrm{B}}} \frac{\alpha}{p^2} \lesssim \frac{L^2 \alpha N^{\delta_{\mathrm{B}}}}{\beta}.
		\label{eq:AppB7}
	\end{equation}
	A straightforward computation also shows
	\begin{equation}
		\int_0^{\alpha} \frac{(\alpha-t)}{(p^2 + t)^2} \de t = \frac{\alpha}{p^2} - \ln\left( 1 + \frac{\alpha}{p^2} \right).
		\label{eq:AppB8}
	\end{equation}
	In combination, \eqref{eq:AppB5}--\eqref{eq:AppB8} imply
	\begin{equation}
		F(\alpha) - F(0) - F'(0) \alpha \leq - \frac{1}{2} \sum_{p \in \beta^{1/2} P_{\mathrm{B}}} \left[  \frac{\alpha}{p^2} - \ln\left( 1 + \frac{\alpha}{p^2} \right) \right] + \frac{C L^2 \alpha^2 N^{\delta_{\mathrm{B}}}}{\beta}.
		\label{eq:AppB9}
	\end{equation}
	Finally, using $\ln(1+x) \geq x - x^2/2$ for $x \geq 0$ we see that
	\begin{equation}
		\sum_{p \in \beta^{1/2} P^{\mathrm{c}}_{\mathrm{B}}} \left[  \frac{\alpha}{p^2} - \ln\left( 1 + \frac{\alpha}{p^2} \right) \right] \leq \frac{\alpha^2}{2 \beta^2} \sum_{p \in P^{\mathrm{c}}_{\mathrm{B}}} \frac{1}{p^4} \lesssim \frac{\alpha^2 L^4}{\beta^2 N^{\delta_{\mathrm{B}}}}.
		\label{eq:AppB10}
	\end{equation}
	When we put our findings together, we obtain a proof of \eqref{eq:AppB1} if $\mu_0=0$ (the last error term excluded). 
	
	If $\mu_0 < 0$ our proof applies without changes and we obtain the first term in the second line of \eqref{eq:AppB1} with $p^2$ replaced by $p^2-\mu_0$. It is not difficult to check that the difference between these two terms is bounded by a constant times $N_0 L^2/(\beta^2 N^2)$, which proves the claim of the lemma.
\end{proof}

\section{Properties of the free energy of the condensate}
\label{app:CondensateFreeEnergy}
In this appendix we prove several statements concerning the effective condensate free energy in \eqref{eq:FreeEnergyBEC}, one of which is Proposition~\ref{prop:FreeEnergyBECb}. The other statements are needed for the proof of Theorem~\ref{thm:Main}. We start our discussion with a lemma that provides us with the asymptotic behavior of the chemical potential.

\begin{lemma}
	\label{lem:ChemPotBEC}
	We consider the limit $N \to \infty$, $\beta = \kappa \beta_{\mathrm{c}}$ with $\kappa \in (0,\infty)$ and $\beta_{\mathrm{c}}$ in \eqref{eq:BECPhaseTransition}. Let $g$ be the Gibbs distribution in \eqref{eq:GibbsDistribution} and assume that $\int_{\mathbb{C}} |z|^2 g(z) \de z = M$. The chemical potential $\mu$ related to $g$ satisfies the following statements for a given $\epsilon > 0$:
	\begin{enumerate}[(a)]
		\item If $M \gtrsim N^{5/6 + \epsilon}$ then there exists a constant $c>0$ such that
		\begin{equation}
			\left| \mu - 8 \pi \mathfrak{a}_N M/L^3 \right| \lesssim L^{-2} \exp\left( - c N^{\epsilon} \right).
			\label{eq:LemChemPotBEC1}
		\end{equation}
		\item If $M \lesssim N^{5/6 - \epsilon}$ then we have
		\begin{equation}
			\left| \mu + \frac{1}{\beta M} \right| \lesssim \frac{N^{-2 \epsilon}}{\beta M}.
			\label{eq:LemChemPotBEC2}
		\end{equation}
	\end{enumerate}
\end{lemma}
\begin{proof}
	We write the two-dimensional integration over $\mathbb{C}$ w.r.t. the measure $\de z = \de x \de y/ \pi$ in polar coordinates $(r,\varphi)$ and afterwards introduce the variable $x = r^2$. This allows us to write
	\begin{equation}
		M = \int_{\mathbb{C}} |z|^2 g(z) \de z = \frac{\int_0^{\infty} x \exp\left(- \beta \left( h x^2 - \mu x \right) \right) \de x}{\int_0^{\infty} \exp\left(- \beta \left( h x^2 - \mu x \right) \right) \de x},
		\label{eq:AppBEC1}
	\end{equation}
	where $h = 4 \pi \mathfrak{a}_N /L^3 \sim L^{-2} N^{-1}$. A short computation shows the integral in the numerator equals
	\begin{equation}
		\frac{1}{2 \beta h} + \frac{\mu}{4 h} \sqrt{ \frac{\pi}{\beta h} } \exp\left( \frac{\beta \mu^2}{4 h} \right) \mathrm{erfc}\left( - \sqrt{ \frac{\beta}{h}} \frac{\mu}{2}  \right),
		\label{eq:AppBEC2}
	\end{equation}
	where $\mathrm{erfc}(x) = (2/\sqrt{\pi}) \int_x^{\infty} \exp(-t^2) \de t$ denotes the complementary error function. For the integral in the denominator we find
	\begin{equation}
		\frac{1}{2} \sqrt{ \frac{\pi}{\beta h} } \exp\left( \frac{\beta \mu^2}{4 h} \right) \mathrm{erfc}\left( - \sqrt{ \frac{\beta}{h}} \frac{\mu}{2} \right).
		\label{eq:AppBEC3}
	\end{equation}
	Let us introduce the notation $\eta =\mu \sqrt{\beta/(4h)}$. Using \eqref{eq:AppBEC2} and \eqref{eq:AppBEC3}, we bring \eqref{eq:AppBEC1} to the form
	\begin{equation}
		\sqrt{\pi \beta h} M = \frac{ 1 + \sqrt{\pi} \eta \exp(\eta^2) \mathrm{erfc}(- \eta) }{\exp(\eta^2) \mathrm{erfc}(- \eta)} \eqqcolon \Upsilon(\eta).
		\label{eq:AppBEC4}
	\end{equation}
	The function $\Upsilon$ is strictly positive, strictly monotone increasing, and satisfies $\lim_{x \to - \infty} \Upsilon(x) = 0 $ as well as $\lim_{x \to \infty} \Upsilon(x) = + \infty$. In the following we study the asymptotic behavior of the (unique) solution to this equation. We start with the parameter regime $M \gtrsim N^{5/6 + \epsilon}$, which implies $\sqrt{\pi \beta h} M \gtrsim N^{\epsilon}$.
	
	In this case the l.h.s. of \eqref{eq:AppBEC4} diverges in the limit $N \to \infty$, and hence $\eta \to \infty$. From \cite[Eq.~7.1.13]{AbraStegun1972} we know that
	\begin{equation}
		\frac{1}{x + \sqrt{x^2 + 2}} < \exp\left( x^2 \right) \int_x^{\infty} \exp\left( -t^2 \right) \de t \leq \frac{1}{x + \sqrt{x^2 + 4/\pi}}
		\label{eq:AppBEC5}
	\end{equation}
	holds for $x \geq 0$. In combination with $\mathrm{erfc}(-\eta) = 2 - \mathrm{erfc}(\eta)$, this implies
	\begin{equation}
		2 \exp\left( \eta^2 \right) - \frac{2}{ \sqrt{\pi} \left( \eta + \sqrt{\eta^2 + 4/\pi } \right)}  \leq \exp\left( \eta^2 \right) \mathrm{erfc}(-\eta) < 2 \exp\left( \eta^2 \right) - \frac{2}{ \sqrt{\pi} \left( \eta + \sqrt{\eta^2 + 2 } \right)}
		\label{eq:AppBEC6}
	\end{equation}
	as well as
	\begin{equation}
		\sqrt{\pi \beta h} M = \frac{ 1 + \sqrt{\pi} \eta \left[2 \exp\left( \eta^2 \right) + O(1/\eta) \right] }{2 \exp\left( \eta^2 \right) + O(1/\eta)}.
		\label{eq:AppBEC7}
	\end{equation}
	We already know that $\eta \gg 1$, and hence $\eta \simeq \sqrt{\beta h} M$. Using this and our assumption $M \gtrsim N^{5/6 + \epsilon}$, which implies $\eta \gtrsim N^{\epsilon}$, we easily check that \eqref{eq:LemChemPotBEC1} holds. It remains to prove \eqref{eq:LemChemPotBEC2}.
	
	If $M \lesssim N^{5/6 - \epsilon}$ the l.h.s. of \eqref{eq:AppBEC4} satisfies $\sqrt{\pi \beta h} M \lesssim N^{-\epsilon}$ and we therefore have $\eta \to -\infty$. To obtain the leading order behavior of $\eta$, the approximation provided by \eqref{eq:AppBEC5} is not sufficiently accurate. A more precise approximation is provided by \cite[Eq.~7.1.23]{AbraStegun1972}, which implies 
	\begin{equation}
		\sqrt{\pi} \exp( x^2 ) \mathrm{erfc}(x) = \frac{1}{x} - \frac{1}{2 x^3} + Q(x), \quad \text{ where $Q$ satisfies } \quad |Q(x)| \leq \frac{3}{4 x^5}
		\label{eq:AppBEC8}
	\end{equation}
	for $x \geq 0$. We use this approximation in \eqref{eq:AppBEC4} and find
	\begin{equation}
		\sqrt{\beta h} M = \frac{1}{2 | \eta |} \left( 1 + O\left(\eta^{-2} \right) \right).
		\label{eq:AppBEC9}
	\end{equation}
	Eq.~\eqref{eq:LemChemPotBEC2} is a direct consequence of \eqref{eq:AppBEC9}. This proves our claim.
\end{proof}

We are now prepared to give the proof of Proposition~\ref{prop:FreeEnergyBECb}. Because of technical reasons, we prove it in a slightly more general situation.

\begin{proposition}
	\label{prop:FreeEnergyBEC}
	We consider the limit $N \to \infty$, $\beta = \kappa \beta_{\mathrm{c}}$ with $\kappa \in (0,\infty)$ and $\beta_{\mathrm{c}}$ in \eqref{eq:BECPhaseTransition}. The following statements hold for given $\epsilon > 0$:
	\begin{enumerate}[(a)]
		\item Assume that $M \gtrsim N^{5/6 + \epsilon}$. There exists a constant $c>0$ such that 
		\begin{equation}
			F^{\mathrm{BEC}}(\beta,M,L,\mathfrak{a}_N) = 4 \pi \mathfrak{a}_N L^{-3} M^2 + \frac{\ln \left( 4 \beta \mathfrak{a}_N/L^3 \right)}{2 \beta} + O\left( L^{-2} \exp\left(- c N^{\epsilon} \right) \right).
			\label{eq:FreeEnergyBECInteractingLimit}
		\end{equation}
		\item Assume that $M \lesssim N^{5/6 - \epsilon}$. Then 
		\begin{equation}
			F^{\mathrm{BEC}}(\beta,M,L,\mathfrak{a}_N) = - \frac{1}{\beta} \ln(M) - \frac{1}{\beta} + O\left( L^{-2} N^{2/3 - 2 \epsilon} \right)
			\label{eq:FreeEnergyBECNonInteractingLimit}
		\end{equation}
		holds. In particular, $F^{\mathrm{BEC}}(\beta,M,L,\mathfrak{a}_N)$ is independent of $\mathfrak{a}_N$ at the given level of accuracy. 
	\end{enumerate} 
\end{proposition}

\begin{proof}
	The free energy $F^{\mathrm{BEC}}(\beta,M,L,\mathfrak{a}_N)$ in \eqref{eq:FreeEnergyBEC} consists of two terms. In the following, we denote the first by $\Phi(\beta,M,L,\mathfrak{a}_N)$. When we apply the same coordinate transformations that led to \eqref{eq:AppBEC1}, we can write it as
	\begin{equation}
		\Phi(\beta,M,L,\mathfrak{a}_N) = - \frac{1}{\beta} \ln\left( \int_0^{\infty} \exp\left( -\beta \left( h x^2 - \mu x \right) \right) \de x \right) = - \frac{1}{\beta} \ln \left( \frac{1}{2} \sqrt{ \frac{\pi}{\beta h} } \exp\left( \frac{\beta \mu^2}{4h} \right) \mathrm{erfc}\left( - \sqrt{ \frac{\beta}{h} } \frac{\mu}{2} \right) \right),
		\label{eq:AppBEC10}
	\end{equation}
	where the second identity follows from the fact that the denominator in \eqref{eq:AppBEC1} is given by \eqref{eq:AppBEC3}. 
	
	We first consider the parameter regime $M \gtrsim N^{5/6 + \epsilon}$, where $\eta \simeq \sqrt{\beta h} M \gtrsim N^{\epsilon}$. An application of \eqref{eq:AppBEC6} shows that the r.h.s. of \eqref{eq:AppBEC10} equals
	\begin{equation}
		-\frac{1}{\beta} \ln\left( \sqrt{ \frac{\pi}{\beta h} } \exp( \eta^2 ) \left( 1 + O\left( \exp(-\eta^2)/\eta \right) \right) \right) = \frac{1}{2 \beta} \ln\left( \frac{4 \beta \mathfrak{a}_N}{L^3} \right) - 4 \pi \mathfrak{a}_N M^2 L^{-3} + O\left( L^{-2} \exp \left(-c N^{2 \epsilon} \right) \right).
		\label{eq:AppBEC11}
	\end{equation}
	From Lemma~\ref{lem:ChemPotBEC} we know that 
	\begin{equation}
		\mu M = 8 \pi \mathfrak{a}_N M^2 L^{-3} + O\left( L^{-2} \exp\left( -c N^{\epsilon} \right) \right). 
		\label{eq:AppBEC12}
	\end{equation}
	In combination, these consideration show
	\begin{equation}
		F^{\mathrm{BEC}}(\beta,M,L,\mathfrak{a}_N) = \Phi(\beta,M,L,\mathfrak{a}_N) + \mu M =  \frac{1}{2 \beta} \ln\left( \frac{4 \beta \mathfrak{a}_N}{L^3} \right) + 4 \pi \mathfrak{a}_N M^2 L^{-3} + O\left( L^{-2} \exp \left(-c N^{\epsilon} \right) \right), 
		\label{eq:AppBEC13}
	\end{equation}
	which proves \eqref{eq:FreeEnergyBECInteractingLimit}. 
	
	Next, we consider the case $M \lesssim N^{5/6 - \epsilon}$, where $\eta \simeq -1/(2 \sqrt{\beta h} M) \lesssim -N^{\epsilon}$. We use \eqref{eq:AppBEC8} to write $\Phi$ as
	\begin{equation}
		\Phi(\beta,M,L,\mathfrak{a}_N) = -\frac{1}{\beta} \ln\left( \sqrt{ \frac{1}{\beta h} } \frac{1}{2 | \eta |} \left( 1 + O \left( \eta^{-2} \right) \right) \right) = -\frac{ \ln( M ) }{\beta} + O\left( N^{-2 \epsilon} / \beta \right).
		\label{eq:AppBEC14}
	\end{equation}
	To obtain the second equality we applied Lemma~\ref{lem:ChemPotBEC}. Another application of the same lemma yields
	\begin{equation}
		\mu M = -\frac{1}{\beta} \left( 1 + O\left( N^{-2 \epsilon} \right) \right).
		\label{eq:AppBEC15}
	\end{equation}
	In combination, \eqref{eq:AppBEC14} and \eqref{eq:AppBEC15} prove \eqref{eq:FreeEnergyBECNonInteractingLimit}. 
\end{proof} 

The last lemma provides us with a large deviations bound as well as with bounds for the moments of the distribution $\zeta$. The large deviations bound is needed in the proof of Lemma~\ref{lm:NG} in Appendix~\ref{app:ExpectedParticleNumber}, whereas the moment bound finds application in Section~\ref{eq:PreparatoryLemmas} in our proof of Lemma~\ref{lm:Tr}. We recall that $\zeta$ equals $g$ in \eqref{eq:GibbsDistribution} except that the chemical potential $\widetilde{\mu}$ is chosen s.t. $\int_{\mathbb{C}} |z|^2 \zeta(z) \de z = \widetilde{N}_0$ holds with $\widetilde{N}_0$ in \eqref{eq:WidetildeN0}.

\begin{lemma}
	\label{lem:momentBoundsZeta}
	We consider the limit $N \to \infty$, $\beta = \kappa \beta_{\mathrm{c}}$ with $\kappa \in (0,\infty)$ and $\beta_{\mathrm{c}}$ in \eqref{eq:BECPhaseTransition}, and assume that $0 \leq \widetilde{N}_0 \lesssim N$ holds. Then there exist constants $c,\tilde{c}>0$ such that 
	\begin{equation}
		\int_{\mathbb{C}} (1+ |z|^{2} ) \mathds{1}(|z|^2 \geq cN) \zeta(z) \de z \lesssim \exp(-\tilde{c} N^{1/3} ).
		\label{eq:largeDeviationsZeta}
	\end{equation}
	Moreover,
	\begin{equation}
		\int_{\mathbb{C}} |z|^{2k} \zeta(z) \de z \lesssim_k N^k
		\label{eq:momentBoundsZeta}
	\end{equation}
	holds for all $k \in \mathbb{N}$.
\end{lemma}
\begin{proof}
	Let us again use the notation $h = 4 \pi \mathfrak{a}_N L^{-3}$. We first consider \eqref{eq:largeDeviationsZeta} with $|z|^2$ replaced by $|z|^{2k}$, $k \in \mathbb{N}_0$ and $\lesssim$ replaced by $\lesssim_k$, that is, we need to derive a bound for
	\begin{equation}
		\int_{\mathbb{C}} |z|^{2k} \mathds{1}(|z|^2 \geq cN) \zeta(z) \de z = \frac{\int_{cN}^{\infty} x^k \exp( -\beta(h x^2 - \widetilde{\mu} x) ) \de x}{\int_{0}^{\infty} \exp( -\beta(h x^2 - \widetilde{\mu} x) ) \de x} = \frac{\int_{cN}^{\infty} x^k \exp( -\beta h( x - \widetilde{\mu}/(2h))^2 ) \de x}{\int_{0}^{\infty} \exp( -\beta h( x - \widetilde{\mu}/(2h))^2 ) \de x}.
		\label{eq:AppBEC16}
	\end{equation} 
	To obtain the first equality, we used the same coordinate transformations as above \eqref{eq:AppBEC1}. Inspection of \eqref{eq:AppBEC4} shows that the chemical potential assumes its largest (positive) values when $\widetilde{N}_0 \sim N$. This follows from the fact that the l.h.s. of \eqref{eq:AppBEC4} is strictly increasing in $M$ and that the two maps $\widetilde{\eta} \mapsto \Upsilon(\widetilde{\eta})$ with $\Upsilon$ in \eqref{eq:AppBEC4} and $\widetilde{\mu} \mapsto \widetilde{\eta} =\widetilde{\mu} \sqrt{\beta/(4h)}$ are strictly increasing. Application of part (a) of Lemma~\ref{lem:ChemPotBEC} and the bound $\widetilde{N}_0 \lesssim N$ therefore show that $\widetilde{\mu}$ can be bounded from above by a constant times $\mathfrak{a}_N N L^{-3} \lesssim L^{-2}$. 
	
	Using this, we see that for $c>0$ large enough and $x \geq cN$, we have $x - \mu/2h \geq x/2$. We insert this bound on the r.h.s. of \eqref{eq:AppBEC16} and find
	\begin{equation}
		\int_{\mathbb{C}} |z|^{2k} \mathds{1}(|z|^2 \geq cN)  \zeta(z) \de z \lesssim_k \exp\left( -\frac{\beta h c^2 N^2}{16} \right) \frac{ \int_{0}^{\infty} \exp( -(\beta h/2)( x - \widetilde{\mu}/(2h))^2 ) \de x}{\int_{0}^{\infty} \exp( -\beta h( x - \widetilde{\mu}/(2h))^2 ) \de x},
		\label{eq:AppBEC17}
	\end{equation}
	where the fraction on the r.h.s. equals $\sqrt{2}$ times
	\begin{equation}
		\frac{ \int_{-(\widetilde{\mu}/2) \sqrt{\beta/(2h)}}^{\infty} \exp( -x^2 ) \de x}{\int_{-(\widetilde{\mu}/2) \sqrt{\beta/h}}^{\infty} \exp( -x^2 ) \de x}. 
		\label{eq:AppBEC18}
	\end{equation}
	If $\widetilde{\mu} \geq 0 $ we obtain an upper bound when we replace the lower integration boundary in the numerator by $-\infty$ and that in the denominator by $0$. This yields an upper bound of order $1$. If $\widetilde{\mu}<0$ we apply \eqref{eq:AppBEC5} on the r.h.s. of \eqref{eq:AppBEC18} and see that it is bounded from above by a constant times $\exp(\beta \widetilde{\mu}^2/(8h))$. In combination, these considerations imply the bound
	\begin{equation}
		\int_{\mathbb{C}} |z|^{2k} \mathds{1}(|z|^2 \geq cN) \zeta(z) \de z \lesssim_k \exp\left( -\frac{\beta h c^2 N^2}{16} \right) \max\{ 1, \exp(\beta \widetilde{\mu}^2/(8h))  \}.
		\label{eq:AppBEC19}
	\end{equation}
	When we assume that $0 > \widetilde{\mu} \geq -C/L^2$ for some $C>0$ and choose $c$ large enough, then \eqref{eq:AppBEC19} proves our claim. It remains to consider the case $\widetilde{\mu}<-C/L^2$.
	
	In this case we start with the term after the first equality sign in \eqref{eq:AppBEC16}. We pick $c_1>0$ and realize that it is bounded from above by
	\begin{align}
		\frac{\int_{cN}^{\infty} x^k \exp( \beta \widetilde{\mu} x ) \de x}{\int_{0}^{c_1 N} \exp( -\beta(h x^2 - \widetilde{\mu} x )) \de x} &\leq  \frac{ \exp(\beta h c_1^2 N^2) \int_{cN}^{\infty} x^k \exp( \beta \widetilde{\mu} x ) \de x}{  \int_{0}^{c_1N} \exp( \beta \widetilde{\mu} x )) \de x} \nonumber \\
		&\lesssim_k \frac{\exp(\beta h c_1^2 N^2) (|\beta \widetilde{\mu}|^{-k-1} + N^k) \exp(\beta \widetilde{\mu} cN)}{1-\exp(\beta \widetilde{\mu} c_1 N)} \lesssim_k \exp(-\tilde{c} N^{1/3}). \label{eq:AppBEC20}
	\end{align}
	for some $\tilde{c}>0$. In the last step we used $\widetilde{\mu}<-C/L^2$ and that $c_1$ can be chosen as small as we wish. In combination with our previous considerations, this proves 
	\begin{equation}
		\int_{\mathbb{C}} |z|^{2k} \mathds{1}(|z|^2 \geq cN) \zeta(z) \de z  \lesssim_k \exp(-\tilde{c} N^{1/3}) \label{eq:AppBEC21}
	\end{equation}
	for all $k \in \mathbb{N}_0$. In particular, \eqref{eq:largeDeviationsZeta} holds.
	
	Eq.~\eqref{eq:momentBoundsZeta} follows from \eqref{eq:AppBEC21} when we use the decomposition
	\begin{equation}
		\int_{\mathbb{C}} |z|^{2k} \zeta(z) \de z = \int_{\{ |z|^2 \leq cN \}} |z|^{2k} \zeta(z) \de z + \int_{\{ |z|^2 > cN \}} |z|^{2k} \zeta(z) \de z \lesssim_k N^k + \exp(-\tilde{c} N^{1/3} ).
	\end{equation}
\end{proof}

\section{The expected particle number in the trial state}
\label{app:ExpectedParticleNumber}

In this appendix we prove Lemma~\ref{lm:NG}. An essential ingredient of the proof are large deviations bounds for $G_{\mathrm{B}}(z)$ in \eqref{eq:Gbog} and $G_{\mathrm{free}}$ in \eqref{eq:Gfree}. Before we state them, we define
\begin{equation}\label{eq:defNBNI}
	\cN_{\mathrm{B}}=\sum_{p\in P_\mathrm{B}}a^*_pa_p\qquad\text{and}\qquad \cN_{\mathrm{I}}=\sum_{p\in P_\mathrm{I}}a^*_pa_p.
\end{equation}  

\begin{lemma}\label{lm:NGBGF} We consider the limit $N \to \infty$, $\beta = \kappa \beta_{\mathrm{c}}$ with $\kappa \in (0,\infty)$ and $\beta_{\mathrm{c}}$ in \eqref{eq:BECPhaseTransition}. For any $c>0, r \in \mathbb{N}$ we have
\begin{equation}
\tr[(1+\cN_\mathrm{B})\mathds{1}(\cN_\mathrm{B}\geq cN)G_{\mathrm{B}}(z)] \lesssim_r N^{-r}\big(N^{2/3+\delta_B}\big)^{r+1}.
\label{eq:plybound}
\end{equation}
Moreover, there exist positive constants $c,\tilde{c}>0$ such that
\begin{equation}
	\tr[(1+\cN_\mathrm{I})\mathds{1}(\cN_\mathrm{I}\geq cN)G_{\mathrm{free}}] \lesssim \exp(-\tilde{c} N^{1/3} ).
	\label{eq:expgf}
\end{equation}

\end{lemma}

\begin{proof} For the sake of simplicity, we give the proof of the first bound with $1+\mathcal{N}_{\mathrm{B}}$ replaced by $\mathcal{N}_{\mathrm{B}}$, and similarly for the second bound. For any $r\geq 1$, we have $\mathds{1}(\cN_\mathrm{B}\geq cN)\leq \cN_\mathrm{B}^r(cN)^{-r}$. Hence, 
\begin{align}
 \tr[\cN_\mathrm{B}\mathds{1}(\cN_\mathrm{B}\geq cN)G_{\mathrm{B}}(z)]&\lesssim (cN)^{-r}\tr[\cN_\mathrm{B}^{r+1}G_{\mathrm{B}}(z)] \nonumber \\
 &= (cN)^{-r}\sum_{p_1, \dots, p_{r+1}\in P_\text{B}}\tr[a^*_{p_1}a_{p_1}\dots a^*_{p_{r+1}}a_{p_{r+1}} G_{\mathrm{B}}(z)].
 \end{align}
After normal ordering and an application of Wick's theorem, we can use \eqref{eq:Comp4} and \eqref{eq:Comp4c} to see that
\begin{align}
 \tr[\cN_\mathrm{B}\mathds{1}(\cN_\mathrm{B}\geq cN)G_{\mathrm{B}}(z)]&\lesssim_r N^{-r}\big(N^{2/3+\delta_B}\big)^{r+1}
 \end{align}
holds. This proves the first bound in \eqref{eq:expgf}

Next, we prove the second bound. Let $0<k\leq L^{-2}\beta$ and observe that
\begin{align}
 k\tr[\cN_\mathrm{I}\mathds{1}(\cN_\mathrm{I}\geq cN)G_{\mathrm{free}}]&\leq  k\tr[(\cN_\mathrm{I}-cN)\mathds{1}(\cN_\mathrm{I}\geq cN)G_{\mathrm{free}}]+ckN \tr[\mathds{1}(\cN_\mathrm{I}\geq cN)G_{\mathrm{free}}]\nonumber\\
 &\leq (1+ckN)\tr[\exp(k(\cN_\mathrm{I}-cN))G_{\mathrm{free}}].\label{eq:d1}
 \end{align}
The trace on the r.h.s. can be written as
\begin{align}
\tr[\exp(k(\cN-cN)G_{\mathrm{free}})] &= \exp(-k cN)\frac{\tr_{\mathscr{F}_{\mathrm{I}}}[ \exp(-\beta \de \Gamma( \mathds{1}( -\mathrm{i} \nabla \in P_{\mathrm{I}})(- \Delta - \mu_0-k\beta^{-1} )))]}{\tr_{\mathscr{F}_{\mathrm{I}}}[ \exp(-\beta \de \Gamma( \mathds{1}( -\mathrm{i} \nabla \in P_{\mathrm{I}})(- \Delta - \mu_0 )))]} \nonumber \\
&= \exp(-k cN) \exp\left(\beta\,\big(\Phi(\mu_0)-\Phi(\mu_0+k\beta^{-1})\big)\right),\label{eq:d2}
\end{align}
where
\begin{equation}
	\Phi(\mu)=-\beta^{-1}\ln\tr_{\mathscr{F}_{\mathrm{I}}}\Big[\exp\Big(-\beta \sum_{p\in P_{\mathrm{I}}}(p^2-\mu)a^*_pa_p)\Big)\Big] = \beta^{-1} \sum_{p\in P_{\mathrm{I}}}\Big[\ln\big(1-\exp(-\beta(p^2-\mu)) \big) \Big].
\end{equation}
Using that $\Phi(\mu)$ is a concave and monotone decreasing function of $\mu$, we obtain the lower bound
\begin{align}
	\beta \Phi(\mu_0+k\beta^{-1}) & \geq \beta \Phi(\mu_0) - \sum_{p\in P_{\mathrm{I}}}\frac{k}{\exp(\beta(p^2-\mu_0-k\beta^{-1}))-1} \nonumber \\
	&\geq \beta \Phi(\mu_0) - CN k
	\label{eq:d3}
\end{align}
for some $C>0$. To come to the last line, we used used $0<k\leq L^{-2}\beta$ and applied Lemma~\ref{lem:RiemannSum}. In combination, these considerations show
\begin{equation}\label{eq:expgb1}
 \tr[e^{k(\cN-cN)}G_{\mathrm{free}}]\lesssim  e^{-k N (c - C) }.
 \end{equation}
When we choose $k= L^{-2} \beta$ and $c>C$ in the above equation this proves the second bound in \eqref{eq:expgf}. 
\end{proof}

We are now prepared to give the proof of Lemma~\ref{lm:NG}.

\begin{proof}[Proof of Lemma~\ref{lm:NG}]

We use $[\cN,B]=0$ to write
\begin{align}\label{eq:trng}
\tr [\cN \G] = \sum_{\alpha} \lambda_{\alpha}\frac{\braket{(1+B)\psi_{\alpha},\cN(1+B)\psi_{\alpha}}} {\braket{(1+B)\psi_{\alpha}, (1+B)\psi_{\alpha}}}= \sum_{\alpha} \lambda_{\alpha}\frac{\langle\psi_{\alpha},(\cN+B^*\cN B)\psi_{\alpha}\rangle} {\braket{(1+B)\psi_{\alpha}, (1+B)\psi_{\alpha}}},
\end{align}
and apply the lower bound $\|(1+B)\psi_{\alpha}\|^2=\langle\psi_\alpha, (1+B^*B)\,\psi_\alpha \rangle\geq 1$ to see that the r.h.s. is bounded from above by $\mathrm{Tr}[ (\cN+B^* B\cN) \Gamma_0 ]$. Using the definition of $B$ in \eqref{eq:B},  Lemma~\ref{lm:Tr} and the bound for $\eta_p$ in \eqref{eq:sumPHC}, we find
\begin{align}\label{eq:NBB}
\text{Tr} [B^* B\cN \Gamma_0] &= \frac{1}{4L^6} \sum_{\substack{q\in\L^*,p_1,p_2 \in P_{\mathrm{H}},\\ \ u_1,v_1,u_2,v_2 \in P_{\mathrm{L}}}}\eta_{p_1}\eta_{p_2} \tr [ a_{u_1}^* a_{v_1}^* a_{u_1+p_1}a_{v_1-p_1} a_{u_2+p_2}^*a_{v_2-p_2}^* a_{u_2 }a_{v_2}a^*_q a_q \Gamma_0 ] \nonumber\\
&=\frac{1}{2 L^6} \sum_{\alpha} \lambda_{\alpha}\sum_{\substack{p_1 \in P_{\mathrm{H}},\\ \ q,u_1,v_1,u_2 \in P_{\mathrm{L}}}}\eta_{p_1}\eta_{p_1+u_1-u_2} \tr[ a^*_{u_1} a_{v_1}^* a_{u_2 }a_{v_1+u_1-u_2}a^*_qa_q \Gamma_0 ] \lesssim N^{\delta_H},
\end{align}
 and therefore $\tr [\cN \G] - \tr [\cN \G_0]  \lesssim N^{\delta_H}$. It remains to prove a lower bound.
 
To that end, we note that $B^*\cN B \geq 0$ and $\mathds{1}(\cN_0\leq cN)\mathds{1}(\cN_{\text{B}}\leq cN)\mathds{1}(\cN_{\text{I}}\leq cN) \leq 1$ implies 
\begin{equation}
	\tr [\cN \G] \geq \sum_{\alpha} \lambda_{\alpha}\frac{\langle\psi_{\alpha},\cN\mathds{1}(\cN_0\leq cN)\mathds{1}(\cN_{\text{B}}\leq cN)\mathds{1}(\cN_{\text{I}}\leq cN)\psi_{\alpha}\rangle} {\braket{(1+B)\psi_{\alpha}, (1+B)\psi_{\alpha}}}
	\label{eq:AppD2}
\end{equation}
for sufficiently large $c>0$. Let $\cN_0=a^*_0a_0$ and recall the definition of $\cN_{\mathrm{B}}$ and $\cN_{\mathrm{I}}$ in \eqref{eq:defNBNI}. In combination, \eqref{eq:AppD2}, $1/(1+x) \geq 1 - x$ for $x\geq 0$ and the bound $\langle\psi_{\alpha},\cN\mathds{1}(\cN_0\leq cN)\mathds{1}(\cN_{\text{B}}\leq cN)\mathds{1}(\cN_{\text{I}}\leq cN)\psi_{\alpha}\rangle\leq 3cN$ allow us to show that
\begin{align}
\tr [\cN \G] &\geq \sum_{\alpha} \lambda_{\alpha}\langle\psi_{\alpha},\cN\mathds{1}(\cN_0\leq cN)\mathds{1}(\cN_{\text{B}}\leq cN)\mathds{1}(\cN_{\text{I}}\leq cN)\psi_{\alpha}\rangle(1-\langle\psi_\alpha, B^*B\,\psi_\alpha \rangle) \nonumber \\
&\geq \tr [\cN \G_0]- \tr\left[ \cN\big\{\mathds{1}(\cN_0\geq cN)+\mathds{1}(\cN_\textsc{B}\geq cN)+\mathds{1}(\cN_\textsc{I}\geq cN)\big\} \Gamma_0 \right] -3cN \tr [ B^*B\, \Gamma_0 ]. \label{eq:lbng}
\end{align}
The last contribution can be bounded from below by $-CN^{\delta_H}$ (this can be seen similarly as for \eqref{eq:NBB}). 

To obtain a bound for the second term on the r.h.s. we apply Lemmas~\ref{lem:momentBoundsZeta} and~\ref{lm:NGBGF}: we have
\begin{equation}
	\tr\left[ \cN\big[\mathds{1}(\cN_\textsc{B}\geq cN)+\mathds{1}(\cN_\textsc{I}\geq cN)\big] \Gamma_0 \right] \lesssim_r N^{1-r}\big(N^{2/3+\delta_{\mathrm{B}}}\big)^{r+1} + \exp(-\tilde{c} N^{1/3} ) \lesssim 1,
	\label{eq:AppD5b}
\end{equation}
provided $c>1$ and $r \in \mathbb{N}$ are chosen large enough. To obtain the second bound we also used the assumption $\delta_{\mathrm{B}} < 1/3$. Next, we consider
\begin{equation}
	\int_{\mathbb{C}} \tr[ (1 + \mathcal{N}_0) \mathds{1}(\mathcal{N}_0 \geq cN) |z \rangle \langle z | ] \zeta(z) \de z = \int_{\mathbb{C}} (1 + |z|^2) \tr[ \mathds{1}(\mathcal{N}_0 + 1 \geq cN ) |z \rangle \langle z | ] \zeta(z) \de z,
	\label{eq:AppD3}
\end{equation}
where we used $a_0 \mathds{1}(\mathcal{N}_0 \geq cN) = \mathds{1}(\mathcal{N}_0 + 1 \geq cN) a_0$. 
Pick $c'>0$. An application of Lemma~\ref{lem:momentBoundsZeta} shows that the term on the r.h.s. of \eqref{eq:AppD3} is bounded from above by 
\begin{align}
	&\int_{\{ |z|^2 \leq c' N \} } (1 + |z|^2) \tr[ \mathds{1}(\mathcal{N}_0 + 1 \geq cN) |z \rangle \langle z | ] \zeta(z) \de z + \int_{\{ |z|^2 > c' N \}} (1+|z|^2) \zeta(z) \de z \nonumber \\
	&\hspace{4cm}\leq (1 + c' N) \int_{\{ |z|^2 \leq c' N \} } \tr[ \mathds{1}(\mathcal{N}_0 + 1 \geq cN) |z \rangle \langle z | ] \zeta(z) \de z + C \exp(- \tilde{c} N^{1/3})
	\label{eq:AppD4} 
\end{align}
for some $\tilde{c} > 0$ as long as $c>0$ is chosen large enough. It remains to consider the first term on the r.h.s. of \eqref{eq:AppD4}.

We evaluate the trace in the eigenbasis $\{ |n \rangle \}_{n \in \mathcal{N}_0}$ of $\mathcal{N}_0$:
\begin{align}
	\int_{\{ |z|^2 \leq c' N \} } \tr[ \mathds{1}(\mathcal{N}_0 + 1 \geq cN) |z \rangle \langle z | ] \zeta(z) \de z &= \int_{\{ |z|^2 \leq c' N \} } \sum_{n=0}^{\infty} \mathds{1}(n \geq cN - 1) | \langle z,n \rangle |^2 \zeta(z) \de z \nonumber \\
	&=\int_{\{ |z|^2 \leq c' N \} } \exp(-|z|^2) \left[ \sum_{n \geq cN-1} \frac{|z|^{2n}}{n!} \right] \zeta(z) \de z.
	\label{eq:AppD5}
\end{align}
To come to the second line, we also used the identity $| \langle z,n \rangle |^2 = \exp(-|z|^2) |z|^{2n} / n!$. An application of Taylor's theorem with an explicit form of the remainder allows us to see that the series on the r.h.s. is bounded from above by $\exp(|z|^2) |z|^{2 cN}/(cN !)$ (here we assume for the sake of simplicity that $c \in \mathbb{N}$). An application of Stirling's approximation formula therefore shows
\begin{equation}
	\int_{\{ |z|^2 \leq c' N \} } \tr[ \mathds{1}(\mathcal{N}_0 + 1 \geq cN) |z \rangle \langle z | ] \zeta(z) \de z \leq \frac{(c'N)^{cN}}{cN!} \lesssim N^{-1/2} \left( \frac{c' e}{c} \right)^{cN},
	\label{eq:AppD6} 
\end{equation}
which is exponentially small in $N$ as long as $c' e/c < 1$ holds. When we put \eqref{eq:lbng}--\eqref{eq:AppD4} and \eqref{eq:AppD6} together, we find the bound
\begin{equation}
	\tr\left[ \cN\big[\mathds{1}(\cN_0\geq cN)+\mathds{1}(\cN_\textsc{B}\geq cN)+\mathds{1}(\cN_\textsc{I}\geq cN)\big] \Gamma_0 \right] \lesssim 1.
	\label{eq:AppD7} 
\end{equation}
In combination with \eqref{eq:lbng} and the assumption $\delta_{\mathrm{H}} > 0$ this proves $\tr [\cN \G] - \tr [\cN \G_0]  \gtrsim -N^{\delta_{\mathrm{H}} }$. We put this result and \eqref{eq:AppD5b} together and obtain a proof of \eqref{eq:NG}. 
\end{proof}
\end{appendix}

\vspace{0.5cm}
(Chiara Boccato) Università degli Studi di Milano \\ 
Via Saldini 50, 20133 Milano, Italy \\  
E-mail address: \texttt{chiara.boccato@unimi.it}

\vspace{0.1cm}

\noindent (Andreas Deuchert) Institute of Mathematics, University of Zurich \\ 
Winterthurerstrasse 190, 8057 Zurich, Switzerland \\ 
E-mail address: \texttt{andreas.deuchert@math.uzh.ch}\\

\vspace{-0.33cm}

\noindent (David Stocker) \\
E-mail address: \texttt{david95$\textunderscore$stocker@hotmail.com} \\

\end{document}